%% file: main.tex
\documentclass[lettersize,journal]{IEEEtran}

\usepackage{amsmath,amsfonts,amssymb,amsthm}
\usepackage{algorithmic}
\usepackage{algorithm}
\usepackage{array}
\usepackage[caption=false,font=normalsize,labelfont=sf,textfont=sf]{subfig}
\usepackage{textcomp}
\usepackage{stfloats}
\usepackage{url}
\usepackage{verbatim}
\usepackage{graphicx}
\usepackage{cite}
\usepackage{enumitem}
\usepackage{xparse}
\usepackage{booktabs}
\usepackage[table]{xcolor} 
\usepackage{makecell}
\usepackage{balance}
\usepackage[percent]{overpic}

\newcolumntype{C}{>{\centering\arraybackslash}p{1.2em}}

\usepackage{hyperref}

\definecolor{gtbg}{RGB}{210,245,210}
\definecolor{mdsbg}{RGB}{255,230,150}
\definecolor{cedbg}{RGB}{255,210,210}
\definecolor{shade}{RGB}{248,248,248}
\definecolor{good}{RGB}{0,120,0}
\definecolor{bad}{RGB}{180,0,0}

\newcommand{\ok}[1]{\textcolor{good}{\textbf{#1}}}
\newcommand{\err}[1]{\textcolor{bad}{\textbf{#1}}}

\theoremstyle{plain} 
\newtheorem{theorem}{Theorem}[section]
\newtheorem{proposition}[theorem]{Proposition}
\newtheorem{lemma}[theorem]{Lemma}

\theoremstyle{definition} 
\newtheorem{definition}[theorem]{Definition}

\theoremstyle{remark}

\makeatletter
\renewcommand{\thetable}{\arabic{table}}    
\renewcommand{\fnum@table}{\textbf{Table \thetable.}}
\makeatother

\newcommand{\appendixheader}{%
  \clearpage
  \section*{Appendix}
  \begin{center}
    \vspace{0.5em}
    \rule{0.4\linewidth}{0.4pt}\hspace{1em}$\blacklozenge$\hspace{1em}\rule{0.4\linewidth}{0.4pt}
    \vspace{1.5em}
  \end{center}%
}

\begin{document}
	
	\input{notations.tex}

	\title{Continuous Edit Distance, Geodesics and Barycenters of Time-varying Persistence Diagrams}
	
	\author{Sebastien Tchitchek, Mohamed Kissi, Julien Tierny
		\thanks{S. Tchitchek, M. Kissi, and J. Tierny are with the CNRS and Sorbonne
Universite.
			E-mail: \{sebastien.tchitchek, mohamed.kissi\}@etu.sorbonne-universite.fr, julien.tierny@sorbonne-universite.fr} %
		\thanks{Manuscript received XXXX XX, XXXX; revised XXXXX XX, XXXX.}
        }
	
	\markboth{Journal of \LaTeX\ Class Files,~Vol.~14, No.~8, August~2021}%
	{Shell \MakeLowercase{\textit{et al.}}: A Sample Article Using IEEEtran.cls for IEEE Journals}
	
	
	\maketitle

	\input{abstract.tex}

\input{introduction_Section1_}
\input{preliminaries_Section2_}
\input{continuous_edit_distance_between_time_varying_persistence_diagrams_Section3_}
\input{continuous_edit_distance_geodesics_between_time_varying_persistence_diagrams_Section4_}
\input{continuous_edit_distance_barycenters_of_time_varying_persistence_diagrams_Section5_}
\input{applications_Section6_}
\input{results_Section7_}

\input{conclusion_Section8_}

	\section*{Acknowledgments}
	\noindent
	\small{This work is partially supported by the European Commission grant
ERC-2019-COG “TORI” (ref. 863464, \url{https://erc-tori.github.io/}).}
		
		\bibliographystyle{IEEEtran}
		\bibliography{refs}

\appendixheader

\input{Supplementary_material}

	\end{document}

%% file: notations.tex

\newcommand{\julien}[1]{#1}

\newcommand{\sebastien}[1]{#1}

\newcommand{\sebastienBis}[1]{#1}

\renewcommand{\sectionautorefname}{Sec.}
\renewcommand{\subsectionautorefname}{Sec.}
\renewcommand{\subsubsectionautorefname}{Sec.}

\renewcommand{\figureautorefname}{Fig.}
\renewcommand{\tableautorefname}{Tab.}

\newcommand{\N}{N}  

\newcommand{\RNB}{\mathbb{R}}      

\newcommand{\NNB}{\mathbb{N}}      

\newcommand{\algebra}{\mathcal{M}}  

\newcommand{\borelian}{\mathcal{B}}

\newcommand{\candidate}{\ensuremath{B}}

\newcommand{\candidateTwo}{\ensuremath{C}}

\newcommand{\frechetEnergy}{\mathcal{E}}

\newcommand{\measure}{\mu}    

\newcommand{\completion}{\mathcal{Z}}

\newcommand{\scalarField}{\ensuremath{f}}   

\newcommand{\variableTime}{\ensuremath{t}}  

\newcommand{\manifold}{\mathcal{M}}  

\newcommand{\dimensionvariable}{\ensuremath{d}}

\newcommand{\timedScalarFieldsSeq}{\ensuremath{U}}  

\newcommand{\PDSeq}{\ensuremath{V}}  

\newcommand{\sampleTVPD}{\ensuremath{V}}  

\newcommand{\divisibilitySymbol}{\mid} 

\newcommand{\sublevelset}{L^{-}_{w}}     

\newcommand{\PDS}{\mathcal{D}}   

\newcommand{\PD}{\ensuremath{X}} 

\newcommand{\birthPair}{\ensuremath{b}} 

\newcommand{\deathPair}{\ensuremath{d}} 

\newcommand{\persistancepair}{\ensuremath{x}} 

\newcommand{\persistancepairTwo}{\ensuremath{y}} 

\newcommand{\PDBijection}{\psi}        

\newcommand{\floorCost}{\ensuremath{c}} 

\newcommand{\wasserstein}{\ensuremath{W}}  

\newcommand{\diago}{\Lambda}  

\newcommand{\diagoProj}[1]{\pi\left(#1\right)}    

\newcommand{\deletionSet}{\mathcal{D}} 

\newcommand{\substitutionSet}{\mathcal{S}} 

\newcommand{\insertionSet}{\mathcal{I}} 

\newcommand{\morseIndex}{\mathcal{I}} 

\newcommand{\persistencePairNumber}{\ensuremath{I}} 

\newcommand{\persistencePairNumberBis}{\ensuremath{J}} 

\newcommand{\persistencePairNumberBisBis}{\ensuremath{K}} 

\newcommand{\metricSpace}{\ensuremath{\mathcal{X}}}

\newcommand{\metricDistance}{\ensuremath{d}}

\newcommand{\paramDelta}{\Delta}     

\newcommand{\paramDeltaone}{\Delta_1}  

\newcommand{\paramDeltatwo}{\Delta_2}  

\newcommand{\paramWeight}{\alpha}    

\newcommand{\paramPenalty}{\beta}            

\newcommand{\scalarVariableOne}{\varepsilon} 

\newcommand{\scalarVariableTwo}{\gamma}

\newcommand{\paramIntervalSize}{\eta}        

\newcommand{\paramSGD}{\rho_{\mathrm{s}}}

\newcommand{\paramGGD}{\ensuremath{k}}

\newcommand{\paramGGDBis}{\ensuremath{M}} 

\newcommand{\stepGre}{\rho_{\mathrm{g}}}

\newcommand{\paramGD}{\ensuremath{T}}

\newcommand{\paramGDBis}{\ensuremath{t}}  

\newcommand{\WtwoGeodesic}{\gamma}

\newcommand{\metricGeodesic}{\gamma}

\newcommand{\boundarySet}{\ensuremath{A}}

\newcommand{\PASet}{\mathcal{A}}

\newcommand{\initDPOne}{\ensuremath{A}}

\newcommand{\initDPTwo}{\ensuremath{B}}

\newcommand{\Subdspace}{\ensuremath{s}}   

\newcommand{\TVPDspace}{\ensuremath{\mathcal{S}}}  

\newcommand{\CTVPDspace}{\ensuremath{C}}   

\newcommand{\ITVPDspace}{\ensuremath{I}}   

\newcommand{\PCTVPDspace}{\mathrm{PC}}     

\newcommand{\Img}{\mathrm{Im}}   

\newcommand{\Dom}{\mathrm{Dom}}   

\newcommand{\SubdInterval}{I}    

\newcommand{\TVPDp}{\ensuremath{P}} 

\newcommand{\TVPDq}{\ensuremath{Q}}

\newcommand{\TVPDf}{\ensuremath{F}}

\newcommand{\TVPDg}{\ensuremath{G}}

\newcommand{\TVPDX}{\ensuremath{X}}

\newcommand{\TVPDY}{\ensuremath{Y}}

\newcommand{\functionf}{\ensuremath{f}}

\newcommand{\functiong}{\ensuremath{g}}

\newcommand{\partialAssignment}{\ensuremath{f}}

\newcommand{\partialAssignmentBis}{\ensuremath{g}}

\newcommand{\partialAssignmentBisBis}{\ensuremath{h}}

\newcommand{\dilatedWtwoGeodesic}{\mathcal{H}}

\newcommand{\geodesicCED}{\ensuremath{G}}

\newcommand{\TDAP}{TDA}

\newcommand{\TDAPP}{TDA\ }

\newcommand{\localDistance}{\ensuremath{d}}

\newcommand{\localDistanceAppendix}{\ensuremath{D}}

\newcommand{\assignmentCost}{\operatorname{cost}} 

\newcommand{\TVPDP}{TVPD\ }

\newcommand{\TVPDPP}{TVPD}

\newcommand{\TVPDsP}{TVPDs\ } 

\newcommand{\TVPDsPP}{TVPDs} 

\newcommand{\TVPDM}{\mathrm{TVPD}}

\newcommand{\CEDP}{CED\ }  

\newcommand{\CEDPP}{CED}  

\newcommand{\CEDM}{\mathrm{CED}}

\newcommand{\DP}{\delta}

\newcommand{\dimension}{\ensuremath{d}}

\newcommand{\intervalBounda}{\ensuremath{a_{i}}}

\newcommand{\intervalBoundaBis}{\ensuremath{a_{i,n}}}

\newcommand{\intervalBoundaBisBis}{\ensuremath{a_{i,n+1}}}

\newcommand{\intervalBoundb}{\ensuremath{b_{i}}}

\newcommand{\intervalBoundc}{\ensuremath{c_{j}}}

\newcommand{\intervalBoundd}{\ensuremath{d_{j}}}

\newcommand{\criticalPointOne}{\ensuremath{c}}

\newcommand{\criticalPointTwo}{\ensuremath{c'}}

\newcommand{\limitl}{\ensuremath{l}}

\newcommand{\variablei}{\ensuremath{i}}

\newcommand{\variablej}{\ensuremath{j}}

\newcommand{\variablek}{\ensuremath{k}}

\newcommand{\variableK}{\ensuremath{K}}

\newcommand{\variableKLips}{\ensuremath{K}}

\newcommand{\variablel}{\ensuremath{l}}

\newcommand{\variableL}{\ensuremath{L}}

\newcommand{\variableM}{\ensuremath{M}}

\newcommand{\variablen}{\ensuremath{n}}

\newcommand{\variableN}{\ensuremath{N}}

\newcommand{\variabler}{\ensuremath{r}}

\newcommand{\variables}{\ensuremath{s}}

\newcommand{\variablev}{\ensuremath{v}}

\newcommand{\variablew}{\ensuremath{w}}

\newcommand{\variablex}{\ensuremath{x}}

\newcommand{\variablexBis}{\ensuremath{x}}

\newcommand{\variablexBisBis}{\ensuremath{x}}

\newcommand{\variabley}{\ensuremath{y}}

\newcommand{\variableyMetricSpace}{\ensuremath{y}}

\newcommand{\variablez}{\ensuremath{z}}

\newcommand{\variableZ}{\ensuremath{Z}}

\newcommand{\rayon}{\ensuremath{r}}

\newcommand{\ball}{\mathcal{B}}

\newcommand{\intervalI}{\ensuremath{I}}

\newcommand{\averageNumberDeltaS}{\ensuremath{n}}

\newcommand{\averageNumberPersistencePair}{\ensuremath{p}}

\newcommand{\averageNumberPersistencePairPerTVPD}{\ensuremath{P}}

\newcommand{\sampleSize}{\ensuremath{N}}

\newcommand{\bigLandau}{\ensuremath{\mathcal{O}}}


%% file: abstract.tex
	\begin{abstract} We introduce the \emph{Continuous Edit Distance} (\CEDPP), a
geodesic and elastic distance for \emph{time-varying persistence diagrams}
(\TVPDsPP).
\julien{The}
	\CEDP extends edit-distance ideas to \julien{\TVPDsP}
	by combining local
substitution costs with penalized deletions/insertions, controlled by two
parameters: \(\paramWeight\) (trade-off between temporal misalignment and
diagram discrepancy) and \(\paramPenalty\) (gap penalty).
We \julien{also}
provide an explicit
construction of
\CEDPP-geodesics. Building on these ingredients, we present two practical
barycenter solvers—\julien{one} stochastic and \julien{one} greedy—that
monotonically decrease the \CEDP Fréchet energy.
Empirically, \julien{the} \CEDP is
robust to additive perturbations
\julien{(both temporal and spatial),}
recovers temporal shifts, and supports
\julien{temporal pattern}
search.
\julien{On real-life datasets,}
\julien{the}
\CEDP \julien{achieves clustering performance comparable or better than}
standard
elastic dissimilarities, while our
\julien{clustering based on \CEDPP-barycenters yields superior
classification results.}
\julien{Overall, the}
\CEDP equips \TVPDP
analysis with a principled distance, interpretable geodesics, and practical
barycenters,
\julien{enabling}
alignment, \julien{comparison,} averaging, and
clustering directly in the space of \TVPDsPP.
A C++ implementation \julien{is provided for}
reproducibility \julien{at the following address 
\href{https://github.com/sebastien-tchitchek/ContinuousEditDistance}{https://github.com/sebastien-tchitchek/ContinuousEditDistance}}.
    
    \end{abstract}

	\begin{IEEEkeywords}
    Topological data analysis, persistent homology, time-varying data.
	\end{IEEEkeywords}

%% file: introduction_Section1_.tex
		\section{Introduction}\label{sec:introduction}

\IEEEPARstart{A}{dvancements} in data simulation and acquisition techniques have
led to
\julien{an}
ever-increasing data complexity and volume. Consequently, traditional
analytical approaches often become inadequate, as they were not designed to
handle such intricate and voluminous data. Thus, the introduction of
abstractions capable of summarizing and interpreting this data becomes necessary
to enable its analysis.

\julien{Topological data analysis (TDA) \cite{book} provides a set of methods
for the efficient encoding of the topological features of a dataset into
concise topological representations, such as persistent diagrams
\cite{edelsbrunner02}, merge and contour trees \cite{carr00,
gueunet_tpds19}, Reeb graphs \cite{Parsa12, gueunet_egpgv19}, or Morse-Smale
complexes \cite{Defl15,gyulassy_vis18}. These representations have been
successfully employed in a variety of applications such as
combustion \cite{bremer_eScience09, gyulassy_ev14},
material sciences \cite{beiNuclear16,soler2019ranking},
fluid dynamics \cite{kasten_tvcg11, nauleau_ldav22},
chemistry \cite{Malgorzata23, daniel_vis2025}, astrophysics
\cite{sousbie11, shivashankar2016felix}.
or geometry processing \cite{vintescu_eg17} and
data science \cite{ChazalGOS13, topoMap}. Among the descriptors studied in
TDA, we focus in this work on a very popular representation: the persistence
diagram.}

%

In addition to the increasing geometrical complexity and volume of datasets, users can face time-varying data, where a given phenomenon is represented not by a single dataset but by a succession of datasets. In the framework of persistent homology, if each dataset is converted into a persistence diagram, the user obtains a time-indexed family of persistence diagrams. We will refer to each such family as a time-varying persistence diagram (\TVPDPP) hereafter.

In that context, \TVPDsP have recently found use as a means to study data that
evolves over time. Some applications have been protein folding trajectory
analysis \cite{cohen2006vines}, music classification
\cite{bergomi2020homological}, time-varying scalar field analysis
\cite{soler2019ranking}, and analysis of dynamic functional brain
connectivity \cite{yoo2016topological}.

However, to integrate \TVPDsP into traditional data analysis pipelines,
additional tools are required for their study. In particular, essential
components such as an elastic distance, geodesics, and barycenter computation
methods are currently lacking. Indeed, a distance enables the use of specific
search methods \cite{chavez2001searching}, piecewise constant approximations
\cite{marteau2008time}, and machine-learning methods with theoretical guarantees
\cite{cover1967nearest, carlsson2010characterization,hastie2009elements}.
Geodesics allow, for example, principal geodesic analysis
\cite{fletcher2004principal, pont2022principal} to facilitate variability
analysis and visualization. Barycenters are necessary for performing $k$-means
clustering and to obtain the corresponding centroids
\cite{vidal2019progressive}. Equally important, to deal with the asynchronous
nature of
\julien{a given phenomenon}
over time,
\julien{one needs}
an elastic distance---that is, a
distance that remains robust to time shifts and time dilatation.

This paper addresses this problem by introducing a geodesic and elastic distance
between \TVPDsPP, called Continuous Edit Distance (\CEDPP), and whose
barycenters can be computed.
As its name suggests, \julien{the} \CEDP can be seen as an extension of the
Edit distance with Real Penalty (ERP) \cite{chen2004marriage}, \julien{from
numerical time series to TVPDs.}
Intuitively, our
distance measures the minimal cost to transform one \TVPDP into another,
through elementary operations---deletion, substitution, and insertion.
\julien{Moreover,}
\julien{the}
\CEDP
\julien{can be used as}
a visual comparison tool,
\julien{capturing similarities between}
two \TVPDsPP.


\subsection{Related work}\label{sec:related_work}

The literature related to our work can be classified into two main families: \textit{(i)} topological methods, and \textit{(ii)} elastic dissimilarities.

\noindent\textbf{(i) Topological methods:} Several tools and methods have been
developed for analysing persistence diagrams and \TVPDsPP. Building on
optimal-transport concepts
\cite{villani2008optimal,ambrosio2008gradient,mileyko2011probability} the
Wasserstein distance
\julien{is a popular metric}
for persistence diagrams.
This
distance exhibits stability properties that make it particularly well suited to
the topological analysis  scalar fields
\cite{cohen2005stability}. Based on the Wasserstein metric, several works have
aimed to compute a representative barycenter for a persistence diagram sample.
The pioneering algorithm was proposed by Turner et al. with Frechet mean
computation \cite{turner2013frechetmeansdistributionspersistence}, while later
Lacombe et al. \cite{lacombe2018large} introduced an entropy-regularised
optimal-transport variant, and Vidal et al. \cite{vidal2019progressive}
presented a progressive computation method.

Cohen-Steiner et al. \cite{cohen2006vines} introduced vineyards as the
piecewise-linear trajectories traced by the points of a persistence diagram as
the underlying data evolves piecewise-linearly over time. Vineyards, which
reveal how individual birth-death pairs evolve through time, provide an
interpretable topological summary for dynamic data
\cite{soler2019ranking,cohen2006vines}, with their
piecewise-linear structure being exploited to track persistence pairs
efficiently.

Turner \cite{turner2013frechetmeansdistributionspersistence} defines vineyards as equally long continuous \TVPDsP and, by encoding them as vineyard modules whose behaviour is captured by a matrix representation, furnishes under certain conditions a tractable framework for analysing their evolving topological structure.

Munch et al.,\cite{munch2015probabilistic} replace the classical Fréchet mean with a probabilistic Fréchet mean (PFM), a probability measure on diagram space.  Evaluated along equal-length continuous vineyards, the PFM yields a continuous mean vineyard (avoiding the discontinuities of the classical mean) and can be stably computed on bootstrap samples of time-varying point clouds, making it a tool for statistical analysis.

\noindent\textbf{(ii) Elastic dissimilarities:}  In this section, we briefly review the principal elastic dissimilarities proposed in the literature, from foundational works to more recent contributions.

Firstly,
\julien{the}
Fréchet distance \cite{frechet1906quelques, alt1995computing} is originally
defined as an elastic distance between curves by considering the minimal leash
length required to traverse them in a continuous and monotonic fashion.
Naturally suited for comparing time-parameterized trajectories, it captures both
spatial proximity and ordering of points. However,
\julien{the}
Fréchet distance is known to
be
sensitive to small perturbations, making it unstable in the presence
of noise or irregular sampling \cite{driemel2013jaywalking}. This sensitivity
limits its direct applicability in real-world time series analysis, especially
in domains where robustness is essential.

Levenshtein et al. \cite{levenshtein1966binary} introduce
\julien{the}
edit distance as
\julien{an}
elastic distance between
\julien{character}
strings, defined as the minimum
number of character insertions, deletions, and substitutions required to
transform one string into another \cite{navarro2001guided}. Initially designed
for applications in error correction and computational linguistics, it has since
been widely used in fields such as bioinformatics, natural language processing,
and information retrieval \cite{gusfield1997algorithms}.

Alternatively, Sakoe and Chiba \cite{sakoe2003dynamic} define dynamic time
warping (DTW), an elastic dissimilarity between time series, as an approach to
continuous speech recognition.
\julien{Due }to its efficiency it
is
extensively
applied to other domains \cite{berndt1994using, rath2003word,
vial2009combination}. Subsequently, several methods of time series averaging for
DTW have been developed and applied
\cite{gupta1996nonlinear,petitjean2011global,petitjean2012summarizing,
schultz2018nonsmooth}. However, DTW is not a distance, as it does not satisfy
the triangle inequality nor identity of indiscernibles.

Following \cite{levenshtein1966binary}, various extensions and approximations of
\julien{the} edit distance have been proposed to improve computational
efficiency or adapt it to specific data types. Among these, the Edit distance
with Real Penalty (ERP), introduced by Chen and Ng \cite{chen2004marriage} is
an adaptation of the edit distance to time series
\cite{zhang2010classification,bagnall2017great}. An arbitrary reference element
is first fixed, then one seeks to transform a time series into another at
minimal cost, through elementary operations — deletion, substitution, and
insertion. A deletion or insertion of an element from a time series has a cost
proportional to its deviation from the reference element. A substitution between
two elements from different time series has a cost proportional to their
difference.

Later,
\julien{the}
Time Warp Edit Distance (TWED) is introduced by Marteau \cite{marteau2008time}
as an elastic distance for time series that simultaneously accounts for both
temporal distortions and amplitude variations.
\julien{The}
TWED gives an alternative to the
ERP by incorporating time stamps into the cost scheme, and by modifying the
costs of the elementary operations.
\julien{The}
TWED has been successfully applied in time
series classification \cite{serra2014empirical, marteau2008time}, and clustering
\cite{tang2015flight}, particularly in contexts requiring metric consistency.
\julien{However, the} TWED, despite being a distance, is \sebastien{not well 
adapted to the geometry of the space of persistence diagrams. Indeed, even when 
\julien{the}
TWED is adapted to the persistence-diagram setting, the cost of deleting (or 
inserting) an element does not depend on its persistence. By contrast, in an 
ERP-based 
\julien{model,}
\julien{if}
the empty diagram is used as a reference, the 
deletion/insertion cost of an element 
\julien{becomes}
proportional to its persistence. This 
is consistent with the standard geometric interpretation in TDA, where more 
persistent features are both more significant and more expensive to remove.}

Soft Dynamic Time Warping (Soft-DTW) is introduced by Cuturi and Blondel
\cite{cuturi2017soft} as a smooth, differentiable relaxation of the classical
DTW, obtained by replacing the hard-minimum operator in DTW’s cost with a
parameter-controlled soft-minimum. Controlled by this parameter
\(\scalarVariableTwo>0\), Soft-DTW converges to the exact DTW distance as
\(\scalarVariableTwo\to 0\), while for
\julien{a}
finite \(\scalarVariableTwo\) it yields
a differentiable objective whose gradients can be efficiently computed. This
differentiability enables seamless integration as a loss function in
gradient-based optimisation pipelines, leading to applications in time-series
averaging, prototype learning, and end-to-end neural network training
\cite{tagliaferri2023applications,cuturi2017soft, ma2023n400}. Although Soft-DTW
is an elastic dissimilarity, it does not constitute a distance, violating
identity of indiscernibles and triangle inequality, but its smoothness provides
favourable optimisation landscapes compared to the classical
DTW\cite{tagliaferri2023applications}.

\subsection{Contributions}\label{sec:contributions}

This paper makes the following new contributions:

\begin{itemize}

\item
\noindent\textit{A practical distance between \TVPDsPP:} We extend the Edit
Distance with Real penalty \cite{chen2004marriage} to
\TVPDsPP. Unlike other state-of-the-art distances
that can be adapted to \TVPDsP framework, our distance allows the computation of
barycenter\julien{s} and geodesic\julien{s} between \TVPDsPP.

\item
\noindent\textit{An \julien{algorithm} for computing geodesics between
\TVPDsPP:} Given our metric, we present a simple three-step approach for
computing geodesics between \TVPDsPP.

\item
\noindent\textit{\julien{An approach} for computing the barycenter of \julien{a
set} of \TVPDsP:} We extend two popular minimization schemes for computing
barycenters of \julien{a set of} \TVPDP:
\julien{a deterministic scheme}
(imitating greedy
subgradient descent), and
\julien{a stochastic one}
(imitating stochastic subgradient descent). Both \julien{versions} are
iterative with monotone improvement and come with practical stopping criteria.

\item
\noindent\textit{An application to pattern matching:} We present
an application
to \julien{temporal} pattern matching between \TVPDsPP.

\item
\noindent\textit{An application to clustering:} We illustrate \julien{the
practical relevance of our barycenters in clustering problems.}

\item
\noindent\textit{Implementation:} We provide a C++ implementation of our
algorithms that can be used for reproducibility (\julien{available at this address \href{https://github.com/sebastien-tchitchek/ContinuousEditDistance}{https://github.com/sebastien-tchitchek/ContinuousEditDistance}}).
\end{itemize}

%% file: preliminaries_Section2_.tex
		\section{Preliminaries}\label{sec:preliminaries}

\noindent This section presents the theoretical background required for
\julien{the presentation}
our
work. It
\julien{formalizes}
our input data
\julien{and introduces persistent diagrams and a typical metric for their
comparison.}
We refer the reader to standard textbooks
\cite{book,oudot2015persistence} for an introduction to
\sebastien{TDA}.

\subsection{Input data}\label{sec:input_data}

We define a timed PL-scalar field as an ordered pair $(\functionf,
\variableTime)$, with $	\variableTime\in\RNB$, and $\functionf$ a
piecewise-linear (PL) scalar field $\functionf:\manifold_\functionf\to\RNB$
defined on a PL $(\dimension_{\manifold_\functionf})$--manifold
$\manifold_\functionf$ with $\smash{\dimension_{\manifold_\functionf}\le 3}$
\julien{in our applications}.

Let \( \paramDelta \in \,(\,0, +\infty) \,\). Each input in our study is a
sequence of timed PL-scalar fields
$\timedScalarFieldsSeq_{\variablen}=\bigl((\functionf_\variablen,
\variableTime_\variablen)\bigr)_{0\leq \variablen\leq
\variableN_\timedScalarFieldsSeq}$, with
$\variableN_\timedScalarFieldsSeq\in\NNB^*$, such that $\paramDelta
\,\mid\,(\variableTime_{\variableN_\timedScalarFieldsSeq}-	\variableTime_0)$
(i.e. $\exists \, \variablek \in \NNB^*,  \,
(\variableTime_{\variableN_\timedScalarFieldsSeq}-	\variableTime_0)=\variablek
\cdot \paramDelta \, $). If $\variablen\in\NNB$, for any threshold
$\variablew\in\RNB$, we write
$\sublevelset(\functionf_\variablen)=\functionf^{-1}_{\variablen}\bigl((
-\infty, \variablew]\bigr)$ for the sub-level set of $\functionf_\variablen$ at
$\variablew$. As $\variablew$ grows, the topology of
$\sublevelset(\functionf_\variablen)$ changes only at the critical values of
$\functionf_\variablen$. The corresponding critical points $\criticalPointOne\in
\manifold_{\functionf_{\variablen}}$ can be classified by their Morse index
$\morseIndex(\criticalPointOne)$: $0$ for minima, $1$ for $1$-saddles,
${{\dimension_{\manifold_{\functionf_{\variablen}}  }} -1}$ for
{$(\dimension_{\manifold_{\functionf_{\variablen}}}-1)$-saddles}, and
$\dimension_{\manifold_{\functionf_{\variablen}}  }$ for maxima. We assume, in
practice \cite{edelsbrunner2001hierarchical,edelsbrunner1990simulation},
\julien{that} all critical points are isolated and non-degenerate. Following
the Elder rule \cite{book}, every topological feature (i.e., connected
components, cycles, and voids) that appears during the sweep
$\variablew:-\infty\to+\infty$ of $\sublevelset(\functionf_\variablen)$ can be
associated with a pair of critical points
$(\criticalPointOne,\criticalPointTwo)$ such that
$\functionf_{\variablen}(\criticalPointOne)<\functionf_{\variablen}(
\criticalPointTwo)$ and $\mathcal \intervalI(\criticalPointOne)=\mathcal
\intervalI(\criticalPointTwo)-1$. The older point $\criticalPointOne$ marks the
birth of the feature, whereas the younger point $\criticalPointTwo$ signals its
death. The pair $(\criticalPointOne,\criticalPointTwo)$ is therefore called a
persistence pair. For instance, when two connected components merge at a
critical point $\criticalPointTwo$, the component that appeared last (the
youngest) vanishes while the oldest persists. Representing each persistence pair
by the coordinates
$(\birthPair,\deathPair)=\bigl(\functionf_{\variablen}(\criticalPointOne),
\functionf_{\variablen}(\criticalPointTwo)\bigr)$ produces a two--dimensional
multiset known as the persistence diagram, denoted $\PD_\variablen$. As shown in \autoref{fig:noNoisy-NoisyPersistenceDiagrams}, prominent topological features
correspond to pairs $(\birthPair,\deathPair)$ far from the diagonal $\Lambda
=\{(\birthPair,\deathPair)\in \RNB^2 \mid \birthPair=\deathPair\}$—that is,
pairs whose lifespan $\deathPair-\birthPair$ (called their persistence) is
large—whereas pairs generated by small-amplitude noise accumulate near the
diagonal. By repeating this procedure for each timed PL-scalar field
$(\functionf_\variablen,	\variableTime_\variablen)$ of
$\timedScalarFieldsSeq_\variablen=\bigl((\functionf_\variablen,
\variableTime_\variablen)\bigr)_{{0\leq \variablen\leq
\variableN_\timedScalarFieldsSeq}}$, we then obtain a sequence of timed
persistence diagrams $\PDSeq_\variablen=\bigl((\PD_\variablen,
\variableTime_\variablen)\bigr)_{{0\leq \variablen\leq \variableN_\PDSeq}}$,
with $\variableN_\PDSeq=\variableN_\timedScalarFieldsSeq$.

\begin{figure}
	\centering
\includegraphics[width=\linewidth]{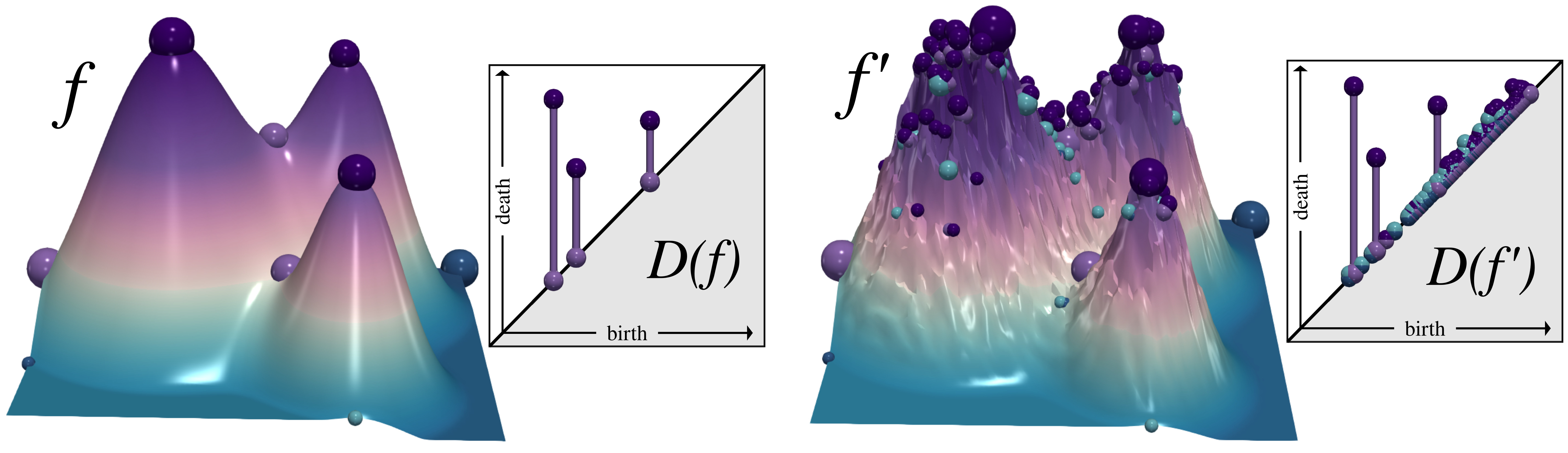}
	\caption{Persistence diagrams $\PDS(\functionf)$ and $\PDS(\functionf')$ of a
noise-free scalar field $\functionf$ (left), and of the same scalar field with
an additive background noise $\functionf'$ (right). The persistence diagrams
show the three main peaks as pairs with large persistence, whereas in the
noise-corrupted diagram $\PDS(\functionf')$ the many pairs with very small
persistence capture only the background
noise.}\label{fig:noNoisy-NoisyPersistenceDiagrams}
\end{figure}

\sebastien{Although our applications focus on timed PL-scalar fields,
the
framework
developed in
\julien{this}
paper
only takes as input such sequences
of timed persistence diagrams. As a result, it is compatible directly with other
types of input data that give rise to timed
persistence diagrams
\julien{(e.g., time-varying point clouds equipped with their
associated Vietoris--Rips filtrations).}}

In what follows, we
enumerate the $\persistencePairNumberBisBis$ points of a persistence diagram
$\PD$ as
$\PD=\{\persistancepair^{1},\dots,\persistancepair^{\persistencePairNumberBisBis
}\}$, and denote $\PD_{\varnothing}=\{\}$ the empty persistence diagram.

\subsection{\julien{Metric for} persistence
diagrams}\label{sec:persistence_diagrams}

Let
$\PD_1=\{\persistancepair_1^1,\dots,\persistancepair_1^{
\persistencePairNumberBisBis_1}\}$ and
$\PD_2=\{\persistancepair_2^1,\dots,\persistancepair_2^{
\persistencePairNumberBisBis_2}\}$ be two persistence diagrams.  To equalise
their cardinalities, we augment each diagram with the diagonal projections of
the off–diagonal points of the other diagram: $\PD_1^{\ast} = \PD_1 \cup
\bigl\{\pi(\persistancepair)\mid \persistancepair\in
\PD_2\setminus\Lambda\bigr\}, $ $ \PD_2^{\ast} = \PD_2 \cup
\bigl\{\pi(\persistancepair)\mid \persistancepair\in
\PD_1\setminus\Lambda\bigr\}$, where the projection $\pi(\birthPair,\deathPair)$
is defined by
$\pi(\birthPair,\deathPair)=\bigl((\birthPair+\deathPair)/2,(\birthPair+
\deathPair)/2\bigr)$.  The sizes of the augmented diagrams coincide, and we set
$\persistencePairNumberBisBis:=|\PD_1^{\ast}|=|\PD_2^{\ast}|$
\julien{and note}
$\PD_{1}^{*}=\{{}^{}\persistancepair_{*1}^{1},\dots,\persistancepair_{*1}^{
\persistencePairNumberBisBis}\}$, and
$\PD_{2}^{*}=\{{}^{}\persistancepair_{*2}^{1},\dots,\persistancepair_{*2}^{
\persistencePairNumberBisBis}\}$.

Let $\intervalI_\persistencePairNumberBisBis=\{1,\dots,\persistencePairNumberBisBis\}$.  A bijection $\psi:\intervalI_\persistencePairNumberBisBis\to \intervalI_\persistencePairNumberBisBis$ specifies a one–to–one matching between the points of $\PD_{1}^{*}$ and $\PD_{2}^{*}$.  We equip $\RNB^2$ with the cost
\( \floorCost(\persistancepair,\persistancepairTwo)= 0, \text{ if }
\persistancepair\in\Lambda\text{ and }\persistancepairTwo\in\Lambda,\,\text{and
}\floorCost(\persistancepair,\persistancepairTwo)=\lVert
\persistancepair-\persistancepairTwo \rVert_2^{\,2}  \text{ otherwise}\),
\julien{where}
$\lVert .\rVert_2$  \julien{is} the \julien{Euclidian} distance
\julien{in}
$\RNB^2$. The
$2$–Wasserstein distance
between the
persistence diagrams $\PD_1^{*}$ and $\PD_2^{*}$  is defined by,
\begin{equation}
\wasserstein_2(\PD_1^{*},\PD_2^{*})=\min_{\psi:\,{bijective}\atop{
\intervalI_\persistencePairNumberBisBis\to
\intervalI_\persistencePairNumberBisBis}}\Bigl(\sum_{\variablej=1}^{
\persistencePairNumberBisBis}
\floorCost\bigl(\persistancepair_{*1}^{\variablej},\persistancepair_{*2}^{\psi(
\variablej)}\bigr)\Bigr)^{1/2}.
	\label{eq:wasserstein}
\end{equation}

Geometrically, $\wasserstein_2$ is the least root–mean–square cost required to transport $\PD_{1}^{*}$ to $\PD_{2}^{*}$ while allowing any point to slide onto the diagonal at zero expense (see \autoref{fig:2-2-L2-PD-Wasserstein_Distance}).

\begin{figure}
	\centering
	\includegraphics[width=\linewidth]{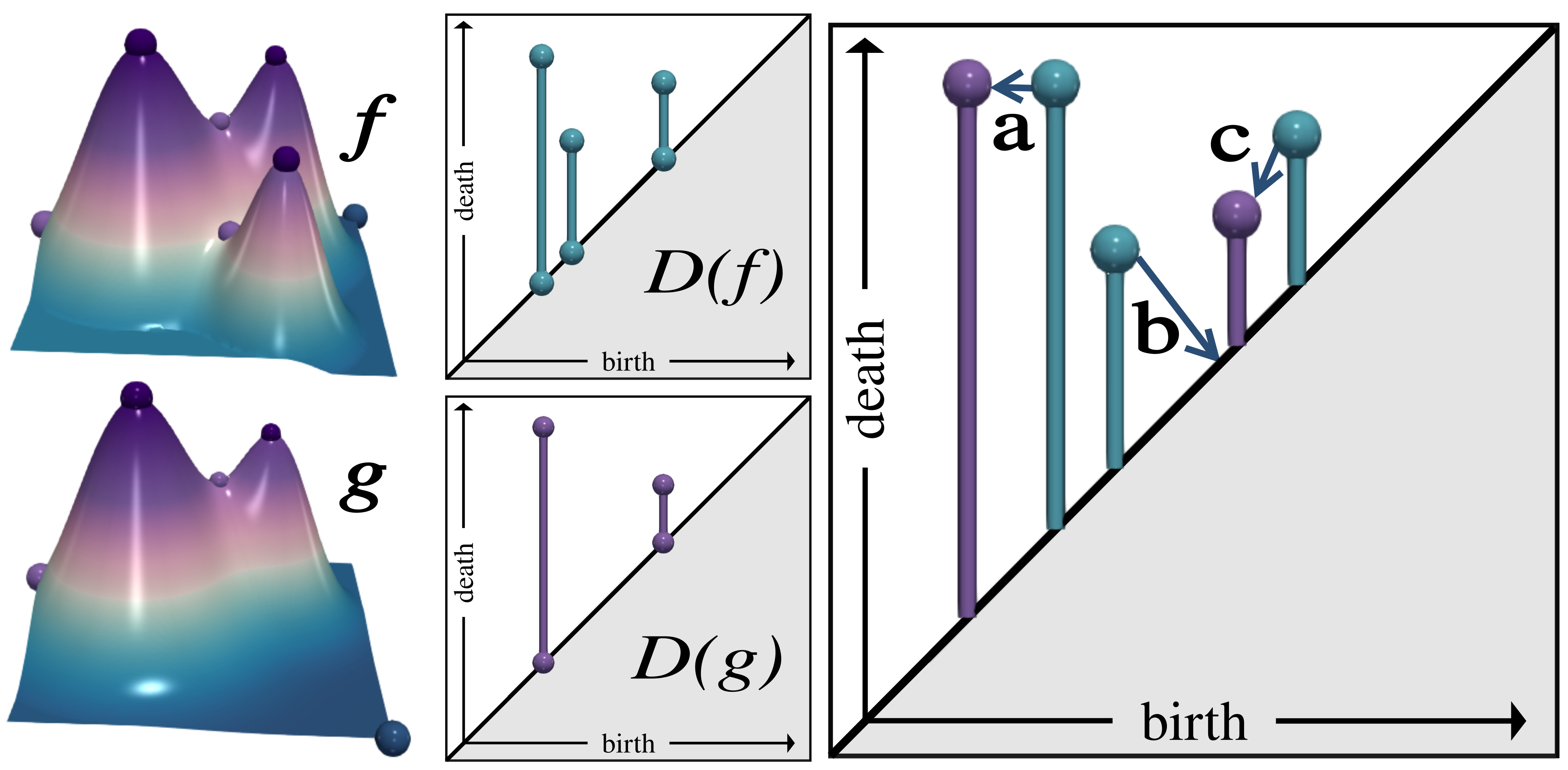}
	\caption{Left: two toy scalar fields $\functionf$ (top) and $\functiong$
(bottom). Center: their persistence diagrams $\PDS(\functionf)$ and
$\PDS(\functiong)$. Right: the optimal 2-Wasserstein matching $\psi$ between
$\PDS(\functionf)$ and $\PDS(\functiong)$, represented by arrows. For
readability, only the persistence pairs of the unaugmented diagrams and the
off-diagonal matchings are displayed. The sum $a^{2}+b^{2}+c^{2}$ of the arrow
lengths equals
$\wasserstein_{2}\bigl(\PDS(\functionf),\PDS(\functiong)\bigr)^{2}$.}
\label{fig:2-2-L2-PD-Wasserstein_Distance}
\end{figure}

%% file: continuous_edit_distance_between_time_varying_persistence_diagrams_Section3_.tex
		\section{Continuous edit distance between \TVPDsPP}\label{sec:continuous_edit_distance_between_TVPDs}

\noindent In this section, we introduce our distance. We begin by specifying the
space of \TVPDsPP. Next, we provide an overview of \julien{the} \CEDP
and establish its basic properties.
\julien{We}
conclude by explaining how to obtain a \TVPDP from our input data, and by
presenting practical methods for computing the \CEDPP.

\subsection{\TVPDsPP}\label{sec:TVPDs}

We introduce a subfamily of \TVPDsPP---still referred to as \TVPDsPP---defined by piecewise-measurable functions (cf. Appendix \ref{app:appendixA} for measure-theoretic preliminaries) that satisfy certain conditions. Indeed, let \( (\PDS, \wasserstein_2) \) be the metric space of persistence diagrams with the $\wasserstein_2$ distance, provided with its Borel sigma-algebra \( \borelian(\PDS) \), \sebastien{and set a fixed non-empty subset \( \boundarySet \subsetneq \PDS \) (whose role and interpretation in the CED construction will be detailed in \autoref{sec:overview})}. Fix  a constant step $\paramDelta\in \,(\,0, +\infty) \,$, and let \( \TVPDspace^\paramDelta \) be the space of functions \( \TVPDf \) from \(dom\,\TVPDf = \bigcup_{\variablei \in \{1, \dots, \variableN_\TVPDf^\paramDelta\}} \intervalI_\variablei^\TVPDf \subset \RNB\) to \( \PDS \), with \( \variableN_\TVPDf^\paramDelta \in \NNB^* \), and \( ( \intervalI_\variablei^\TVPDf )_{\variablei \in \{1, \dots, \variableN_\TVPDf^\paramDelta\}} \) a disjoint family of intervals of \( \RNB \), such that:

\begin{itemize}
	\item \( \forall \variablei \in \{1, \dots, \variableN_\TVPDf^\paramDelta\}, \, \sup \intervalI_\variablei^\TVPDf - \inf \intervalI_\variablei^\TVPDf = \paramDelta \),
	\item \( \forall (\variablei, \variablei') \in \{1, \dots, \variableN_\TVPDf^\paramDelta\}^2,\ \forall (\variablex, \variablex') \in \intervalI_\variablei^\TVPDf \times \intervalI_{\variablei'}^\TVPDf,\ \variablei < \variablei' \Rightarrow \variablex < \variablex' \),
	\item \( \forall \variableTime \in dom\,\TVPDf,\, \wasserstein_2(\TVPDf(\variableTime), \boundarySet)>0 \),
	\item \( \forall \variablei \in \{1, \dots, \variableN_\TVPDf^\paramDelta\},\, \TVPDf|_{\intervalI_\variablei^\TVPDf} : \bigl( \intervalI_\variablei^\TVPDf,\ \sebastien{\borelian(\RNB)|_{\intervalI_\variablei^\TVPDf}} \bigr) \to \bigl( \PDS,\ \borelian(\PDS) \bigr) \) is \sebastien{Lebesgue}-measurable,
	\item \( \forall \variablei \in \{1, \dots, \variableN_\TVPDf^\paramDelta\},\quad \sebastien{\sup\{ \wasserstein_2(a,b) \mid a \in \mathrm{Im}(\TVPDf|_{\intervalI_\variablei^\TVPDf}),\ b } \in  \sebastien{\mathrm{Im}(\TVPDf|_{\intervalI_\variablei^\TVPDf}) \} < \infty }\).
\end{itemize}

We will omit the superscript $\paramDelta$ in the notation $\variableN_\TVPDf^\paramDelta$ whenever no ambiguity arises. Let \( (\TVPDf ,\TVPDg)\in (\TVPDspace^\paramDelta)^2 \), \( \variablei \in \{1, \dots, \variableN_\TVPDf\}, \variablej \in \{1, \dots, \variableN_\TVPDg\}\). \sebastien{We denote \( \TVPDf_\variablei := \TVPDf|_{\intervalI_\variablei^\TVPDf} \),} and we call \( \TVPDf_\variablei \) a \( \paramDelta \)-subdivision. We consider that two \( \paramDelta \)-subdivisions \( \TVPDf_\variablei, \TVPDg_\variablej \) are equal if: \( \inf \intervalI_\variablei^\TVPDf = \inf \intervalI_\variablej^\TVPDg \), \( \sup \intervalI_\variablei^\TVPDf = \sup \intervalI_\variablej^\TVPDg \), and \( \lambda ( \{ \variablex \in \intervalI_\variablei^\TVPDf \cap \intervalI_\variablej^\TVPDg \mid \TVPDf_\variablei(\variablex) \neq \TVPDg_\variablej(\variablex) \})  = 0  \). We denote \( \Subdspace^\paramDelta \) the \( \paramDelta \)-subdivision space.

In a similar way, let \(  (\TVPDf, \TVPDg) \in (\TVPDspace^\paramDelta)^2 \), with  \(\variableN_\TVPDf =  \variableN_\TVPDg\), we consider that $\TVPDf$, $\TVPDg$ are equal if, \( \forall \variablei \in \{1, \dots, \variableN_\TVPDf\},\,  \TVPDf_\variablei =  \TVPDg_\variablei\). More generally, two measurable applications $\TVPDf$ and $\TVPDg$ from $\metricSpace\subset\RNB $ to $\PDS$ are considered equal in our framework whenever \( \lambda ( \{ \variablex \in \metricSpace,\,  \TVPDf(\variablex) \neq \TVPDg(\variablex) \})  = 0  \).

Then, we call \TVPDsP elements of $\TVPDspace^\paramDelta$ (see
\autoref{fig:TVPD}). Observe that, under this definition—and since every
continuous map is Lebesgue-measurable (cf. Appendix \ref{app:appendixA})—any mapping $\TVPDf$ that
satisfies all the stated conditions, with the fourth condition replaced by,
\(\forall \variablei \in \{1, \dots,\variableN_\TVPDf\},\,
\TVPDf|_{\intervalI_\variablei^\TVPDf} : (\intervalI_\variablei^\TVPDf,\left|
\cdot \right|)\to( \PDS,\ \wasserstein_2)\) is continuous, belongs to
$\TVPDspace^\paramDelta$ and therefore is a \TVPDPP. In the
\julien{remaining,}
any such \TVPDPP, with the additional conditions $\forall \variablei
\in \{1,
\dots,\variableN_\TVPDf\},\,\exists(\limitl_\variablei^+,\limitl_\variablei^-)
\in\PDS^{\,2},\,\text{lim}_{	\variableTime\to (\inf
\intervalI_\variablei^\TVPDf)^+}\TVPDf|_{\intervalI_\variablei^\TVPDf}=
\limitl_\variablei^+\text{ and } \, \text{lim}_{	\variableTime\to (\sup
\intervalI_\variablei^\TVPDf)^-}\TVPDf|_{\intervalI_\variablei^\TVPDf}=
\limitl_\variablei^-$, will be referred to as continuous. We denote
$\CTVPDspace^\paramDelta$ this subset of $\TVPDspace^\paramDelta$ that contains
all continuous \TVPDsPP. Note that, as a result, every Cohen–Steiner
\julien{\cite{cohen2006vines}}
or Turner
\julien{\cite{turner2013frechetmeansdistributionspersistence}}
vineyard is a continuous \TVPDPP.

Finally, remark that for any $\TVPDp \in \TVPDspace^\paramDelta$, thanks to the imposed conditions, we can unambiguously denote $\TVPDp$ as a sequence $(\TVPDp_\variablei)_{1 \leq \variablei < \variableN_\TVPDp}.$ Moreover, if  $\variableN \in \{1, \dots, \variableN_\TVPDp\}$,  then $(\TVPDp_\variablei)_{1 \leq \variablei < \variableN}$  stay obviously in $\TVPDspace^\paramDelta$.

\begin{figure}[!t]  
	\centering
	\includegraphics[width=\linewidth]{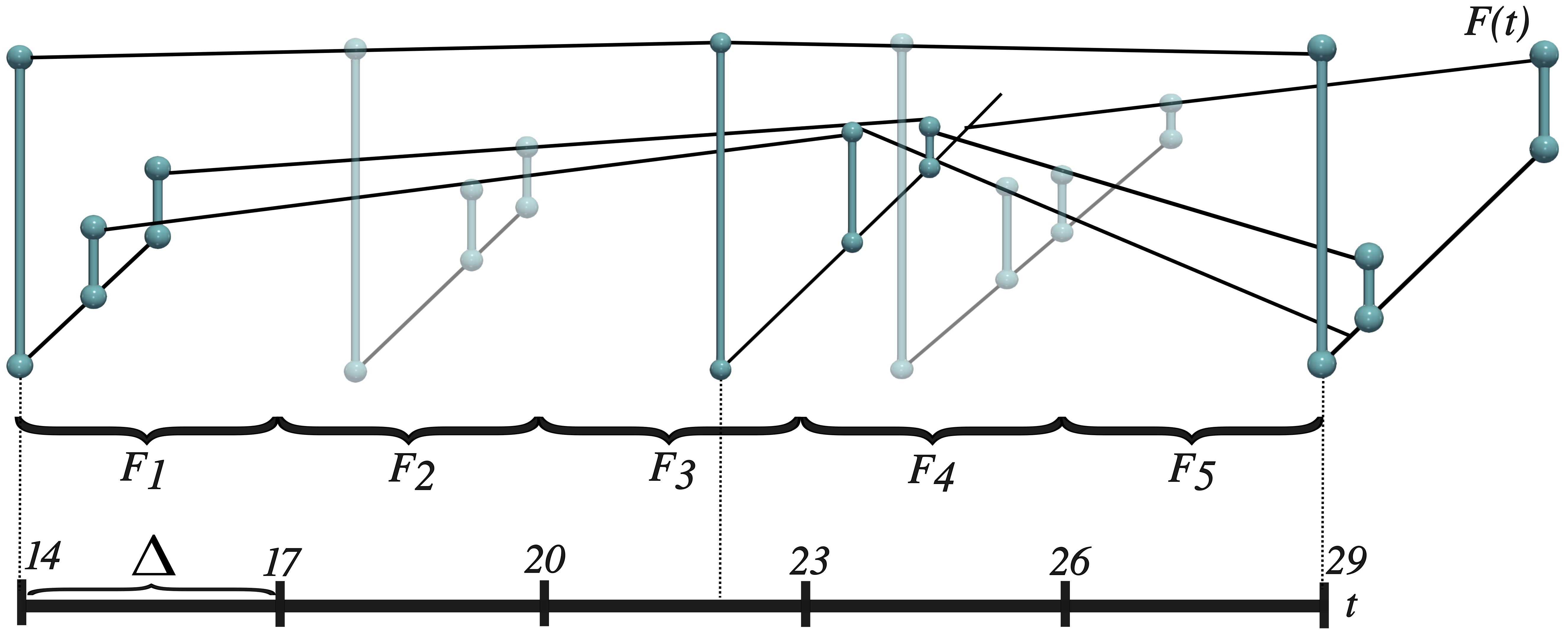}
	\caption{Fix $\boundarySet = \{\PD_{\varnothing}\}\subsetneq \PDS$ as the
singleton containing the empty persistence diagram, and set $\paramDelta = 3$.
The figure depicts a continuous function $\TVPDf\,\colon
dom\,\TVPDf=[14,29]=\!\bigcup_{\variablei=1}^{5} \intervalI_\variablei^\TVPDf
\longrightarrow \PDS,
\intervalI_1^\TVPDf=[14,17),\;\intervalI_2^\TVPDf=[17,20),\;\intervalI_3^\TVPDf=
[20,23),\;\intervalI_4^\TVPDf=[23,26),\;\intervalI_5^\TVPDf=[26,29],$ where
every $\intervalI_\variablei^\TVPDf$ has length
	\julien{$\paramDelta = 3$.}
Opaque diagrams correspond to the input timed persistence diagrams, whereas semi-transparent diagrams are intermediate diagrams inserted by interpolation along $2$-Wasserstein geodesics between successive inputs (see \autoref{sec:input_TVPD}). Since
	$\TVPDf$ is continuous on $[14,29]$, it is measurable on
$[14,29]$, and
on each $\intervalI_\variablei^\TVPDf$. Then every $\TVPDf_\variablei$ is measurable, which implies that each
$\TVPDf_\variablei$ is Lebesgue–measurable, and thus $\TVPDf$ is a TVPD.}\label{fig:TVPD}
\end{figure}

\subsection{Overview}\label{sec:overview}

\sebastien{Let \( \paramWeight \in (0,1) \) and \( \paramPenalty \in (0,1] \) be fixed parameters}. In what follows, we present an overview of the \(\CEDM^\paramDelta_{\paramWeight,\paramPenalty}\) construction developed across \autoref{sec:local_distance} and \autoref{sec:global_distance}.

The \(\CEDM^\paramDelta_{\paramWeight,\paramPenalty}\) distance combines two conceptual layers.

\noindent
\textbf{(i) A local metric on $\paramDelta$-subdivisions}. Let two \TVPDsP $\TVPDp=(\TVPDp_1,\dots,\TVPDp_{\variableN_\TVPDp})\in \TVPDspace^{\paramDelta}$ and $\TVPDq=(\TVPDq_1,\dots,\TVPDq_{\variableN_\TVPDq})\in \TVPDspace^{\paramDelta}$. \autoref{sec:local_distance} defines a metric $\localDistance^{\paramWeight}_{\paramDelta}$ on the space
$\Subdspace^{\paramDelta}$ of $\paramDelta$-subdivisions. For two elements $\TVPDp_\variablei\in \Subdspace^\paramDelta$ from $\TVPDp$ and $\TVPDq_\variablej\in \Subdspace^\paramDelta$ from $\TVPDq$, $\localDistance^{\paramWeight}_{\paramDelta}(\TVPDp_\variablei,\TVPDq_\variablej)$ blends spatial cost—the overall $\wasserstein_2$-distance between the corresponding images $Im(\TVPDp_\variablei)=\{\TVPDp(	\variableTime):	\variableTime\in \intervalI_\variablei^\TVPDp\}$ and $Im(\TVPDq_\variablej)=\{\TVPDq(\variableTime):	\variableTime\in \intervalI_\variablej^\TVPDq\}$---with temporal cost---the shift between the intervals $\intervalI_\variablei^\TVPDp$ and $\intervalI_\variablej^\TVPDq$---weighted respectively by $1-\paramWeight$ and $\paramWeight$. \sebastien{Here, the parameter \(\paramWeight\) controls the trade-off between temporal and spatial
contributions to the local metric $\localDistance^{\paramWeight}_{\paramDelta}$. In this formulation,} $\localDistance^{\paramWeight}_{\paramDelta}(\TVPDp_\variablei,\TVPDq_\variablej)$ is interpreted as the cost of substituting $\TVPDp_\variablei$ with $\TVPDq_\variablej$ (or vice versa).
Additionally, $\localDistance^{\paramWeight}_{\paramDelta}(\TVPDp_\variablei,\boundarySet)$ (resp. $\localDistance^{\paramWeight}_{\paramDelta}(\TVPDq_\variablej,\boundarySet)$) measures the overall $\wasserstein_2$-distance of the image of the $\paramDelta$-subdivision $\TVPDp_\variablei$ (resp $\TVPDq_\variablej$) to the set \boundarySet; this value therefore represents the cost of
deleting that $\paramDelta$-subdivision \sebastien{into
$\boundarySet$ }(resp.\ inserting \sebastien{it from $\boundarySet$}). \sebastien{Therefore, $\boundarySet$ serves as a reference modeling deletions and insertions; for example, a typical choice is $\boundarySet = \{\PD_{\varnothing}\}$.}

\noindent
\textbf{(ii) A global edit distance with penalty between entire time varying persistence diagrams.}
\autoref{sec:global_distance} lifts
the local metric $\localDistance^{\paramWeight}_{\paramDelta}$ to whole \TVPDsP via a $\paramDelta$-partial assignment
$\partialAssignment:\operatorname{dom}\partialAssignment\subset\{1,\dots,\variableN_\TVPDp^\paramDelta\}\rightarrow\operatorname{Im}\partialAssignment\subset\{1,\dots,\variableN_\TVPDq^\paramDelta\},$
required to be strictly increasing. Such an assignment realises three elementary operations: Substitution of $\TVPDp_\variablei$ by $\TVPDq_{\variablej}$ (if $\partialAssignment(\variablei)=\variablej$) with cost $\localDistance^{\paramWeight}_{\paramDelta}(\TVPDp_\variablei,\TVPDq_{\partialAssignment(\variablei)})$; Deletion of $\TVPDp_\variablei$ (if $\variablei\notin\operatorname{dom}\partialAssignment$) with cost $\paramPenalty\cdot\,\localDistance^{\paramWeight}_{\paramDelta}(\TVPDp_\variablei,\boundarySet)$; Insertion of $\TVPDq_\variablej$ (if $\variablej\notin\operatorname{Im}\partialAssignment$) with cost $\paramPenalty\cdot\,\localDistance^{\paramWeight}_{\paramDelta}(\TVPDq_\variablej,\boundarySet)$; where \sebastien{the parameter} $\paramPenalty$ controls the penalty for unmatched $\paramDelta$-subdivisions. $\CEDM^{\paramDelta}_{\paramWeight,\paramPenalty}(\TVPDp,\TVPDq)$ is the minimum total cost over all $\paramDelta$-partial assignments. Here, $\CEDM^{\paramDelta}_{\paramWeight,\paramPenalty}(\TVPDp,\TVPDq)$ is interpreted as the minimal cost of converting $\TVPDp$ into $\TVPDq$ (or vice versa) by means of a $\paramDelta$-partial assignment. 


\noindent
\textbf{Key properties.}
The pair $(\TVPDspace^{\paramDelta},\CEDM^{\paramDelta}_{\paramWeight,\paramPenalty})$
is a metric space \sebastien{(see Appendix \ref{thm:metric_space_global})}. An $\bigLandau(\variableN_\TVPDp^\paramDelta \cdot \variableN_\TVPDq^\paramDelta)$ dynamic-programming scheme
(\autoref{sec:CED_DPC}) evaluates the distance and produces an optimal partial assignment. Finally, continuous \TVPDsP can be approximated arbitrarily well (in \CEDPP, \autoref{sec:pc_TVPD}) by piecewise-constant ones, so the metric is amenable to practical computation on sampled data (\autoref{sec:input_TVPD}).

\subsection{Distance between \( \paramDelta \)-subdivision\julien{s}}
\label{sec:local_distance}

Let \( (p, q) \in (\Subdspace^\paramDelta)^2 \), then \( \exists (\TVPDp, \TVPDq) \in (\TVPDspace^\paramDelta)^2,\ \exists (\variablei,\variablej) \in \{1, \dots, \variableN_\TVPDp\} \times \{1, \dots, \variableN_\TVPDq\} \), such that \( p = \TVPDp_\variablei \) and \( q = \TVPDq_\variablej\). Denoting \(\intervalBounda = \inf(\intervalI_\variablei^\TVPDp), \,\intervalBoundb =\sup(\intervalI_\variablei^\TVPDp),\, \intervalBoundc = \inf(\intervalI_\variablej^\TVPDq), \,\intervalBoundd =\sup(\intervalI_\variablej^\TVPDq)\), we define \( \localDistance^\paramWeight_\paramDelta(p, q) = \localDistance^\paramWeight_\paramDelta(\TVPDp_\variablei, \TVPDq_\variablej) :=\int_{\intervalBounda}^{\intervalBoundb} (1 - \paramWeight) \cdot \wasserstein_2\big( \TVPDp(\variableTime),\ \TVPDq(\variableTime + \intervalBoundc - \intervalBounda) \big) + \paramWeight \cdot |(\variableTime+\intervalBoundc - \intervalBounda)-(\variableTime)|\,dt= \int_{\intervalBounda}^{\intervalBoundb} (1 - \paramWeight) \cdot \wasserstein_2\big( \TVPDp(\variableTime),\ \TVPDq(\variableTime + \intervalBoundc - \intervalBounda) \big) + \paramWeight \cdot |\intervalBoundc - \intervalBounda|\,dt
\). With this definition, we have the result that \( (\Subdspace^\paramDelta,\localDistance^\paramWeight_\paramDelta) \) is a metric space, and so $\localDistance^\paramWeight_\paramDelta$ is a distance on $\Subdspace^\paramDelta$ \sebastien{(see Appendix \ref{lem:metric_space_local})}. Intuitively, the distance $\localDistance^\paramWeight_\paramDelta$ can be understood as the successive sum of spatial and temporal distances between the elements composing each $\paramDelta$-subdivision. The relative contribution of the spatial part is weighted by the parameter $1-\paramWeight$, while that of the temporal part is weighted by $\paramWeight$.

Moreover, with the same notation, we define \(\localDistance^\paramWeight_\paramDelta(p, \boundarySet) = \localDistance^\paramWeight_\paramDelta(\TVPDp_\variablei, \boundarySet) := \int_{\intervalBounda}^{\intervalBoundb} (1 - \paramWeight) \cdot \wasserstein_2\bigl(\TVPDp(\variableTime), \boundarySet\bigr)\,dt\), with  $\forall \variablex \in \PDS,\, \wasserstein_2(\variablexBis, \boundarySet)=\text{inf}\,\bigl(\wasserstein_2(\variablexBis,\variabley),\,\variabley\in \boundarySet\bigr)$. In this case, $\localDistance^\paramWeight_\paramDelta(\TVPDp_\variablei, \boundarySet)$ can be viewed as the successive sum of spatial distances from the elements of the $\paramDelta$-subdivision $\TVPDp_\variablei$ to the set $\boundarySet$. Then we have, as a result, that $\localDistance^\paramWeight_\paramDelta(\TVPDp_\variablei,\TVPDq_\variablej)+\localDistance^\paramWeight_\paramDelta(\TVPDq_\variablej,\boundarySet) \geq \localDistance^\paramWeight_\paramDelta(\TVPDp_\variablei, \boundarySet)$ \sebastien{(see Appendix \ref{lem:triangular_inequality_local_distance_lemma})}.

\subsection{From distance between subdivisions to distance between \TVPDsPP}\label{sec:global_distance}

Set $ (\TVPDp, \TVPDq) \in (\TVPDspace^\paramDelta)^2$.  We call
$\paramDelta$-partial assignment (we omit the $\paramDelta$ when the context is
clear) between  $\TVPDp$  and  $\TVPDq$  any function $\partialAssignment$  from
 dom $\partialAssignment \subset \{1, 2, \dots, \variableN_\TVPDp\}$ to Im
$\partialAssignment$ $\subset \{1, 2, \dots, \variableN_\TVPDq\}$, such that
\partialAssignment~  is strictly increasing, i.e.  $\forall (\variablei,
\variablej)\in$ $(\text{dom }\partialAssignment)^2$, $ \variablei < \variablej
\Rightarrow \partialAssignment(\variablei) < \partialAssignment(\variablej).$ We
denote $\PASet^\paramDelta(\TVPDp, \TVPDq)$ the set of the partial assignments
between $\TVPDp$ and $\TVPDq$ in $\TVPDspace^\paramDelta$, we can see that
specifying an $\partialAssignment \in \PASet^\paramDelta(\TVPDp, \TVPDq)$
directly yields an assignment $\partialAssignment^{-1} \in
\PASet^\paramDelta(\TVPDq, \TVPDp)$.

Then, the $\CEDM^{\paramDelta}_{\paramWeight,\paramPenalty}$ between $\TVPDp=(\TVPDp_\variablei)_{1 \leq \variablei \leq \variableN_\TVPDp}$ and $\TVPDq=(\TVPDq_\variablej)_{1 \leq \variablej \leq \variableN_\TVPDq}$ is defined as: 
\begin{align}
	\CEDM^{\paramDelta}_{\paramWeight,\paramPenalty} (\TVPDp,\TVPDq)=&\mathop{\min}\limits_{\substack{\partialAssignment\in \PASet^\paramDelta(\TVPDp,\TVPDq)}} \text{cost}^\paramDelta_{(\TVPDp,\TVPDq)}(\partialAssignment):= \notag  \\ \mathop{\min}\limits_{\substack{\partialAssignment\in \PASet^\paramDelta(\TVPDp,\TVPDq)}}&\Big(\sum\limits_{\variablei \in \text{ dom }\partialAssignment} \,\,\localDistance^\paramWeight_{\paramDelta}(\TVPDp_\variablei,\TVPDq_{\partialAssignment(\variablei)})\; \tag{2}\\ +&\;\;\sum\limits_{\variablei\, \notin \text{ dom }\partialAssignment}\,\, \paramPenalty\cdot \localDistance^\paramWeight_{\paramDelta}(\TVPDp_\variablei,\boundarySet)\; \tag{3}\\\displaybreak[2] +&\;\:\:\,\sum\limits_{\variablej\, \notin \text{ Im }\partialAssignment}\paramPenalty\cdot \localDistance^\paramWeight_{\paramDelta}(\TVPDq_\variablej,\boundarySet) \Big)\tag{4}
\end{align}

Through this definition, the $\CEDM^{\paramDelta}_{\paramWeight,\paramPenalty}$ between two \TVPDsP is the minimal cost to transform the first \TVPDP into the second (see \autoref{fig:CED}), via a partial assignment, and three elementary operation types—substitution (line 2), deletion (line 3), and insertion (line 4). A substitution replaces a subdivision of the first \TVPDP with the subdivision assigned to it in the second \TVPDPP, at a cost equal to the distance $\localDistance^\paramWeight_\paramDelta$ between the two subdivisions. A deletion removes a subdivision from the first \TVPDPP, at a cost equal to $\paramPenalty$ times the distance $\localDistance^\paramWeight_\paramDelta$ of that subdivision to the set \boundarySet. An addition inserts, into the first \TVPDPP, a subdivision from the second \TVPDPP, at a cost equal to \(\paramPenalty\) times the distance $\localDistance^\paramWeight_\paramDelta$ of that subdivision to the set \(\boundarySet\).

We observe that, if $\variableN_\TVPDp=\variableN_\TVPDq=\variableN$, a direct consequence of the definition is $\CEDM^{\paramDelta}_{\paramWeight,\paramPenalty} (\TVPDp,\TVPDq) \leq \sum_{\variablei \in \{1,\, \dots \,,\,\variableN\}} \,\,\localDistance^\paramWeight_{\paramDelta}(\TVPDp_\variablei,\TVPDq_\variablei)$. Indeed, the identity function $Id:\{1,\, \dots \,,\,\variableN\}\to\{1,\, \dots \,,\,\variableN\},\,\variablei\to \variablei$ is a partial assignment $\bigl(i.e., Id\in\PASet^\paramDelta(\TVPDp,\TVPDq)\bigr)$, then $\CEDM^{\paramDelta}_{\paramWeight,\paramPenalty} (\TVPDp,\TVPDq) \leq $ cost$^\paramDelta_{(\TVPDp,\TVPDq)}(Id)=\sum_{\variablei \in \{1,\, \dots \,,\,\variableN\}} \,\,\localDistance^\paramWeight_{\paramDelta}(\TVPDp_\variablei,\TVPDq_{\variablei})$. 

Note also that for any $(\paramDeltaone,\paramDeltatwo)\in\,(0,\infty)^{\,2},$ such that $\paramDeltaone\mid\paramDeltatwo$, if $(\TVPDp,\TVPDq)\in \TVPDspace^{\paramDeltatwo}\times \TVPDspace^{\paramDeltatwo}$ then as a result $(\TVPDp,\TVPDq)\in \TVPDspace^{\paramDeltaone}\times \TVPDspace^{\paramDeltaone}$, and $\CEDM^{\paramDeltaone}_{\paramWeight,\paramPenalty}(\TVPDp,\TVPDq)\leq$ $\CEDM^{\paramDeltatwo}_{\paramWeight,\paramPenalty}(\TVPDp,\TVPDq)$. Roughly speaking, the inequality follows from the fact that the partial assignment $\partialAssignmentBis:$ dom $\partialAssignmentBis \subset \{1, 2, \dots, \variableN_\TVPDp^{\paramDeltatwo}\}$ $\to$ Im $\partialAssignmentBis$ $\subset \{1, 2, \dots, \variableN_\TVPDq^{\paramDeltatwo}\}$ that realises $\CEDM^{\paramDeltatwo}_{\paramWeight,\paramPenalty}(\TVPDp,\TVPDq) $ $\bigl($i.e., such that $\CEDM^{\paramDeltatwo}_{\paramWeight,\paramPenalty}(\TVPDp,\TVPDq)$ = cost$^{\paramDeltatwo}_{(\TVPDp,\TVPDq)}(\partialAssignmentBis)\bigr)$, possesses an equivalent partial assignment $\partialAssignmentBis':$ dom $\partialAssignmentBis' \subset \{1, 2, \dots, \variableN_\TVPDp^{\paramDeltaone}\}$ $\to$ Im $\partialAssignmentBis'$ $\subset \{1, 2, \dots, \variableN_\TVPDq^{\paramDeltaone}\}$ that realises the same cost in $\TVPDspace^{\paramDeltaone}$ $\bigl($i.e., cost$^{\paramDeltaone}_{(\TVPDp,\TVPDq)}(\partialAssignmentBis')=$ cost$^{\paramDeltatwo}_{(\TVPDp,\TVPDq)}(\partialAssignmentBis)$; \sebastien{see Appendix \ref{prop:monotonicity_Delta}$\bigr)$}. Therefore, $\CEDM^{\paramDeltaone}_{\paramWeight,\paramPenalty}(\TVPDp,\TVPDq)=\mathop{\min}_{\substack{\partialAssignment\in \PASet^{\paramDeltaone}(\TVPDp,\TVPDq)}} \text{cost}^{\paramDeltaone}_{(\TVPDp,\TVPDq)}(\partialAssignment)\leq$ cost$^{\paramDeltaone}_{(\TVPDp,\TVPDq)}(\partialAssignmentBis')=$ cost$^{\paramDeltatwo}_{(\TVPDp,\TVPDq)}(\partialAssignmentBis)=$ $\CEDM^{\paramDeltatwo}_{\paramWeight,\paramPenalty}(\TVPDp,\TVPDq)$. Moreover, by identical reasoning, we observe that for every partial assignment  $\partialAssignmentBisBis\in\PASet^{\paramDeltatwo}(\TVPDp,\TVPDq)$, there exists a partial assignment $\partialAssignmentBisBis'\in\PASet^{\paramDeltaone}(\TVPDp,\TVPDq)$, such that cost$^{\paramDeltaone}_{(\TVPDp,\TVPDq)}(\partialAssignmentBisBis')=$ cost$^{\paramDeltatwo}_{(\TVPDp,\TVPDq)}(\partialAssignmentBisBis)$. Then, the set $\PASet^{\paramDeltaone}(\TVPDp,\TVPDq)$ is at least as extensive as the set $\PASet^{\paramDeltatwo}(\TVPDp,\TVPDq)$; in the sense that $\{\text{cost}^{\paramDeltatwo}_{(\TVPDp,\TVPDq)}(\partialAssignment)\in\RNB_+\,\vert\,\partialAssignment\in\PASet^{\paramDeltatwo}(\TVPDp,\TVPDq)\}\subsetneq\{\text{cost}^{\paramDeltaone}_{(\TVPDp,\TVPDq)}(\partialAssignment)\in\RNB_+\,\vert\,\partialAssignment\in\PASet^{\paramDeltaone}(\TVPDp,\TVPDq)\}$. It is therefore informative to choose $\paramDelta$ small, since a lower parameter value can reveal new, finer assignments. 

\begin{figure}[!t]  
	\centering
	\includegraphics[width=\linewidth]{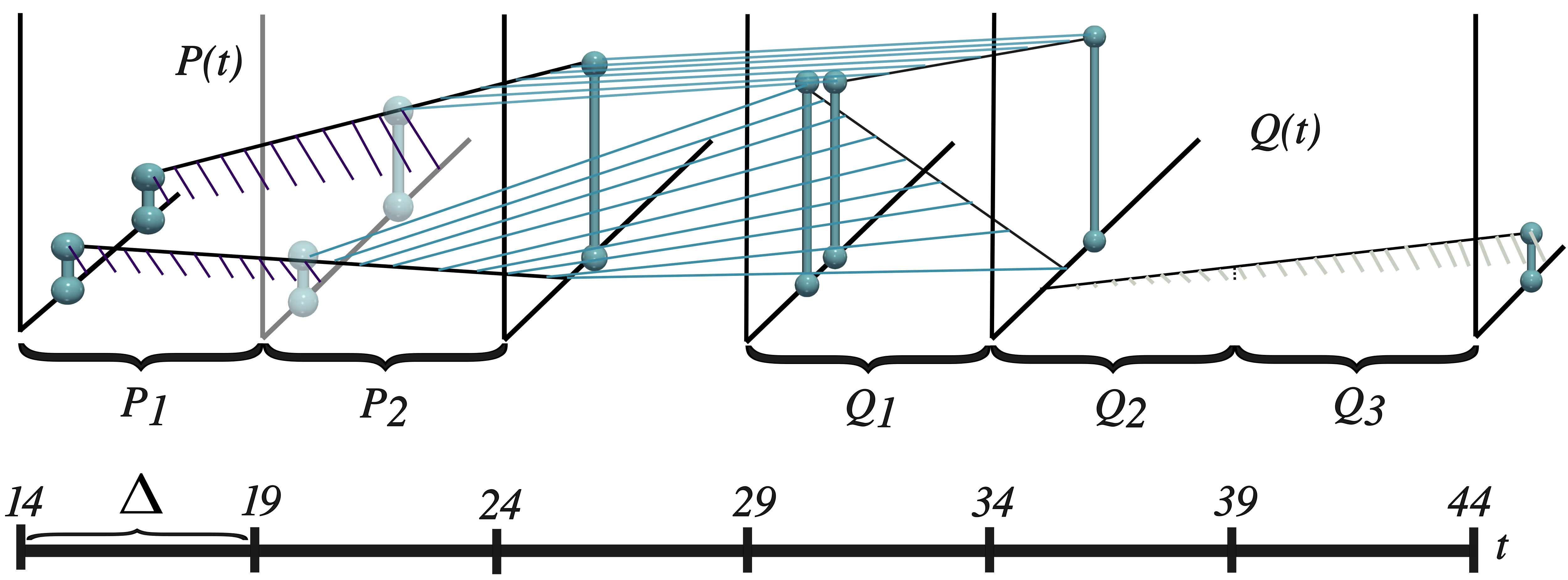}
	\caption{Schematic illustration of the $\CEDM_{\paramWeight,\paramPenalty}^{\paramDelta}$, with $\boundarySet= \{\PD_{\varnothing}\}$ and $\paramDelta = 5$, between two \TVPDsP \TVPDp and \TVPDq. The optimal $\paramDelta$-partial assignment is the function $\partialAssignment:dom\,\partialAssignment=\{2\}\rightarrow\{1\}$. Accordingly, $\CEDM_{\paramWeight,\paramPenalty}^{\paramDelta}(\TVPDp,\TVPDq)$ decomposes into deletion of $\TVPDp_1$ $\bigl($purple hatched strip, cost $\paramPenalty\cdot\,\localDistance_{\paramDelta}^{\paramWeight}(\TVPDp_1,\boundarySet)\bigr)$; substitution of $\TVPDp_2$ by $\TVPDq_1$ (blue hatched strip, cost $\localDistance_{\paramDelta}^{\paramWeight}(\TVPDp_2,\TVPDq_1)$); insertions of $\TVPDq_2$ and $\TVPDq_3$ (gray hatched strips, cost $\paramPenalty\cdot\bigl(\localDistance_{\paramDelta}^{\paramWeight}(\TVPDq_2,\boundarySet)+\localDistance_{\paramDelta}^{\paramWeight}(\TVPDq_3,\boundarySet)\bigr)$. The colored and hatched regions therefore sum to $\CEDM_{\paramWeight,\paramPenalty}^{\paramDelta}(\TVPDp,\TVPDq)$.}\label{fig:CED}
\end{figure}

\subsection{Computation via Dynamic Programming}\label{sec:CED_DPC}

We provide in this subsection a computation method by dynamic programming \cite{bellman1966dynamic}, illustrated in \autoref{fig:CED_Geodesic_Temp}, for the $\CEDM^{\paramDelta}_{\paramWeight,\paramPenalty}$ between two \TVPDsPP: Let $\TVPDp = (\TVPDp_\variablei)_{1 \leq \variablei \leq \variableN_\TVPDp} \in \TVPDspace^\paramDelta$, $\TVPDq = (\TVPDq_\variablej)_{1 \leq \variablej \leq \variableN_\TVPDq} \in \TVPDspace^\paramDelta$. In this subsection, we will note for $\variablev = (\variablex, \variabley) \in \RNB^2$, $\variablev_1 = \variablex$, $\variablev_2 =\variabley$, $\delta_\paramDelta^{\paramWeight,\paramPenalty}(\variablev) =\delta_\paramDelta^{\paramWeight,\paramPenalty }\left((\TVPDp_\variablei)_{0 \leq \variablei \leq \variablev_1}, (\TVPDq_\variablej)_{0 \leq \variablej \leq \variablev_2} \right)$. Then, we define recursively, $\forall \variableK \in \{1, \dots, \variableN_\TVPDp\}$, $\forall \variableK' \in \{1, \dots, \variableN_\TVPDq\}$,

\medskip

$\begin{aligned}&\delta_\paramDelta^{\paramWeight,\paramPenalty} \left(\variableK,\variableK' \right)= \min \Bigr\{ \delta_\paramDelta^{\paramWeight,\paramPenalty} \left(\variableK-1,\variableK'\right)+ \paramPenalty\cdot \localDistance_\paramDelta^\paramWeight (\TVPDp_\variableK,\boundarySet),\\ & \delta_\paramDelta^{\paramWeight,\paramPenalty}\left(\variableK-1, \variableK'-1 \right) + \localDistance_\paramDelta^\paramWeight (\TVPDp_\variableK, \TVPDq_{\variableK'}),\\& \delta_\paramDelta^{\paramWeight,\paramPenalty} \left(\variableK, \variableK'-1\right)+ \paramPenalty \cdot \localDistance_\paramDelta^\paramWeight(\TVPDq_{\variableK'},\boundarySet) \Bigr\}\end{aligned}$

\medskip

\noindent with initialization 
$\delta_\paramDelta^{\paramWeight,\paramPenalty}\bigl(0,0\bigr)=0,\forall \variableK\in\{1,\dots,\variableN_\TVPDp\},\,\delta_\paramDelta^{\paramWeight,\paramPenalty}\bigl(\variableK,0\bigr)=\delta_\paramDelta^{\paramWeight,\paramPenalty}\bigl(\variableK-1,0\bigr)+\paramPenalty \cdot \localDistance_\paramDelta^\paramWeight(\TVPDp_\variableK,\boundarySet),
\text{ and } \forall \variableK'\in\{1,\dots,\variableN_\TVPDq\},\,\delta_\paramDelta^{\paramWeight,\paramPenalty}\bigl(0,\variableK'\bigr)=\delta_\paramDelta^{\paramWeight,\paramPenalty}\bigl(0,\variableK'-1\bigr)+\paramPenalty \cdot \localDistance_\paramDelta^\paramWeight(\TVPDq_{\variableK'},\boundarySet).$ With this definition in place, we have the result that $\CEDM^{\paramDelta}_{\paramWeight,\paramPenalty}(\TVPDp, \TVPDq) = \delta_\paramDelta^{\paramWeight,\paramPenalty} \left(\variableN_\TVPDp,  \variableN_\TVPDq\right)$ \sebastien{(see Appendix \ref{prop:distance_DP_equality})}. 

Moreover, once the recursive computation of $\delta_\paramDelta^{\paramWeight,\paramPenalty} \left( (\TVPDp_\variablei)_{0 \leq \variablei \leq \variableN_\TVPDp}, (\TVPDq_\variablej)_{0 \leq \variablej \leq \variableN_\TVPDq} \right)$ has been carried out, another result is that the partial assignment $\partialAssignment$ achieving the minimal cost defining $\CEDM^{\paramDelta}_{\paramWeight,\paramPenalty}(\TVPDp,\TVPDq)$, i.e. such that $\CEDM^{\paramDelta}_{\paramWeight,\paramPenalty}(\TVPDp, \TVPDq)=$ cost$^\paramDelta_{(\TVPDp,\TVPDq)}(\partialAssignment)$, can be determined by the following procedure:

We define recursively, with initializations $\initDPOne_0 = \emptyset$, $\initDPTwo_0 = (\variableN_\TVPDp,\variableN_\TVPDq)$,

\begin{itemize}
	\item $\initDPTwo_{\variablez+1} = \initDPTwo_\variablez - (1, 0)$ if 
	\[
	\delta_\paramDelta^{\paramWeight,\paramPenalty}(\initDPTwo_\variablez) = \delta_\paramDelta^{\paramWeight,\paramPenalty}\bigl(\initDPTwo_\variablez - (1,0)\bigr) +\paramPenalty\cdot \localDistance^\paramWeight_\paramDelta(\TVPDp_{\initDPTwo_{\variablez_1}}, \boundarySet),
	\]
	
	\item $\initDPTwo_{\variablez+1} = \initDPTwo_\variablez - (1, 1)$ if 
	\[
	\delta_\paramDelta^{\paramWeight,\paramPenalty}(\initDPTwo_\variablez) = \delta_\paramDelta^{\paramWeight,\paramPenalty}\bigl(\initDPTwo_\variablez - (1,1)\bigr) + \localDistance^\paramWeight_\paramDelta(\TVPDp_{\initDPTwo_{\variablez_1}}, \TVPDq_{\initDPTwo_{\variablez_2}}),
	\]
	
	\item $\initDPTwo_{\variablez+1} = \initDPTwo_\variablez - (0, 1)$ if 
	\[
	\delta_\paramDelta^{\paramWeight,\paramPenalty}(\initDPTwo_\variablez) = \delta_\paramDelta^{\paramWeight,\paramPenalty}\bigl(\initDPTwo_\variablez - (0,1)\bigr) + \paramPenalty \cdot \localDistance^\paramWeight_\paramDelta(\TVPDq_{\initDPTwo_{\variablez_2}}, \boundarySet),
	\]

	\item \(\initDPOne_{\variablez+1} = \initDPOne_\variablez \cup \initDPTwo_\variablez \quad \text{if} \quad \initDPTwo_{\variablez+1} = \initDPTwo_\variablez - (1,1)\),
	\item \(\initDPOne_{\variablez+1} = \initDPOne_\variablez \:\:\quad\qquad\text{if }\initDPTwo_{\variablez+1} = \initDPTwo_\variablez - (1, 0)\),
	\item \(\initDPOne_{\variablez+1} = \initDPOne_\variablez \:\:\quad\qquad\text{if } \initDPTwo_{\variablez+1} = \initDPTwo_\variablez - (0, 1)\).   
\end{itemize}	

\medskip

We stop when $\initDPTwo_\variableZ = (0,0)$ for some $\variableZ \in \NNB$, and
then we have as \julien{a} result $\partialAssignment = \initDPOne_\variableZ$
(see \autoref{fig:CED_Geodesic_Temp} \sebastien{and Appendix \ref{lem:dp_backtracking_correct})}.

\begin{figure*}[!t]
  \centering
  \begin{overpic}[width=\linewidth]{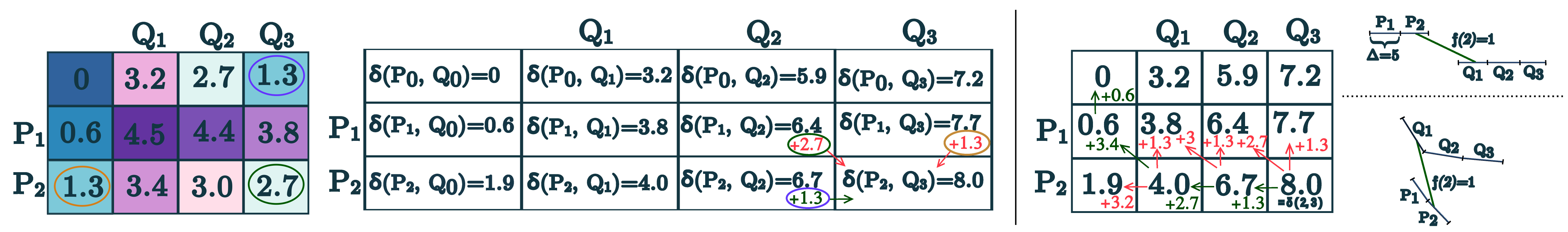}
    \put(1,13){\textbf{(a)}}
    \put(21,13){\textbf{(b)}}
    \put(66,13){\textbf{(c)}}
    \put(97,13){\textbf{(d)}}
  \end{overpic}
	\caption{(a) Cost matrix of the \TVPDsP $\TVPDp=(\TVPDp_{1},\TVPDp_{2})$ and
$\TVPDq=(\TVPDq_{1},\TVPDq_{2},\TVPDq_{3})$ in \autoref{fig:CED}: the first
column and first row contain the initialization terms
\(\localDistance_{\paramDelta}^{\paramWeight}(\cdot,\boundarySet)\); the
remaining cells are the pairwise costs
\(\localDistance_{\paramDelta}^{\paramWeight}(\TVPDp_\variablei,
\TVPDq_\variablej)\). Cell color encodes the cost magnitude. (b) Dynamic
programming computation of the
$\CEDM_{\paramWeight,\paramPenalty}^{\paramDelta}$ between \TVPDsP $\TVPDp$ and
$\TVPDq$.
\julien{For the example of the bottom right entry of the matrix, to}
	compute $\DP_\paramDelta^{\paramWeight,\paramPenalty} \left(
(\TVPDp_\variablei)_{1 \leq \variablei\leq 2}, (\TVPDq_\variablej)_{1 \leq
\variablej \leq 3} \right)$, we take the minimum among
$\DP_\paramDelta^{\paramWeight,\paramPenalty} \left((\TVPDp_\variablei)_{1 \leq
\variablei \leq 1},(\TVPDq_\variablej)_{1 \leq \variablej \leq 3} \right)+
\paramPenalty\cdot \localDistance_\paramDelta^\paramWeight
(\TVPDp_2,\boundarySet)$ (deletion, vertical arrow),
$\DP_\paramDelta^{\paramWeight,\paramPenalty}\left((\TVPDp_\variablei)_{1\leq
\variablei \leq 1}, (\TVPDq_\variablej)_{1 \leq \variablej\leq 2} \right) +
\localDistance_\paramDelta^\paramWeight (\TVPDp_2, \TVPDq_{3})$ (substitution,
diagonal arrow), and $\DP_\paramDelta^{\paramWeight,\paramPenalty} \left(
(\TVPDp_\variablei)_{1 \leq \variablei \leq 2},(\TVPDq_\variablej)_{1 \leq
\variablej \leq 2}\right)+ \paramPenalty \cdot
\localDistance_\paramDelta^\paramWeight(\TVPDq_{3},\boundarySet)$ (insertion,
horizontal arrow). Since the minimum is reached by the insertion term,
$\DP_\paramDelta^{\paramWeight,\paramPenalty} \left( (\TVPDp_\variablei)_{1 \leq
\variablei\leq 2}, (\TVPDq_\variablej)_{1 \leq \variablej \leq 3} \right)=
\CEDM_{\paramWeight,\paramPenalty}^{\paramDelta}(\TVPDp,\TVPDq)$ is set to that
value. (c) Procedure to recover the optimal partial assignment
$\partialAssignment$ that attains
$\CEDM_{\paramWeight,\paramPenalty}^{\paramDelta}(\TVPDp,\TVPDq)$.  Each cell of
the table displays the value
$\DP_{\paramDelta}^{\paramWeight,\paramPenalty}\big((\TVPDp_\variablei)_{1\le
\variablei\le \variabler},(\TVPDq_\variablej)_{1\le \variablej\le \variables}
\big)$;
arrows indicate the three candidate predecessors (deletion  $\uparrow$,
substitution $\nwarrow$, insertion  $\leftarrow$). A green arrow marks the
predecessor realizing the minimum, so the reverse scan yields
$\partialAssignment\colon dom\,\partialAssignment=\{2\}\!\to\!\{1\}$ with
$\partialAssignment(2)=1$. (d,top) Temporal representation of the \TVPDsP
$\TVPDp$ and $\TVPDq$ and their optimal partial assignment $\partialAssignment$
(horizontal axis is time, ticks are the subdivision boundaries of length
$\paramDelta$). Each \TVPDP is represented by its domain of definition, and
green straight line segment\julien{s} represent the optimal assignment
$\partialAssignment$; the vertical spacing between the \TVPDsP is solely to ease
the visualization of \(\partialAssignment\). (d, bottom)
$\localDistance^\paramWeight_\paramDelta$-MDS \sebastien{(multidimensional scaling)} embedding of \TVPDsP $\TVPDp$ and
$\TVPDq$ in $\RNB^{3}$ (each point along a curve represents a persistence
diagram), with the optimal assignment $\partialAssignment$ again drawn as green
straight segment.}
\label{fig:CED_Geodesic_Temp}
\end{figure*}

\subsection{Construction of \julien{a} \TVPDP from
\julien{an input diagram sequence}}
\label{sec:input_TVPD}

In this subsection, we describe how to obtain a \TVPDP from
\julien{an input sequence of persistence diagrams.}

Let
$\PDSeq_\variablen\bigl((\PD_\variablen,\variableTime_\variablen)\bigr)_{{0\leq
\variablen\leq \variableN_\PDSeq}}$ be the sequence of timed persistence
diagrams
(cf. \autoref{sec:input_data}). We
now turn this discrete sequence into a continuous application
\(\TVPDf:([\variableTime_0,\variableTime_{\variableN_\PDSeq}],| \cdot|
)\to(\PDS,\wasserstein_2)\) in three substeps.

\subsubsection{\julien{Contiguous}
geodesics}
For every $\variablen\in\{0,\dots,\variableN_\PDSeq-1\}$ we compute the 2-Wasserstein distance
\(
\wasserstein_2(\PD_\variablen,\PD_{\variablen+1})
\)
and select a $\wasserstein_2$–geodesic
\(
\WtwoGeodesic_\variablen\colon[0,1]\to\PDS
\)
satisfying
\(\WtwoGeodesic_\variablen(0)=\PD_\variablen\) and \(\WtwoGeodesic_\variablen(1)=\PD_{\variablen+1}\).
Such a geodesic exists because \((\PDS,\wasserstein_2)\) is a
geodesic metric space~\cite{turner2013frechetmeansdistributionspersistence}.

\subsubsection{Temporal interpolation}
We map each real time \(\variableTime\in[\variableTime_\variablen,\variableTime_{\variablen+1})\) to the
geodesic parameter
\(
\lambda_\variablen(\variableTime)=\dfrac{\variableTime-\variableTime_\variablen}{\variableTime_{\variablen+1}-\variableTime_\variablen}\in[0,1),
\)
and set $\dilatedWtwoGeodesic_\variablen(\variableTime)\;=\;\WtwoGeodesic_\variablen\!\bigl(\lambda_\variablen(\variableTime)\bigr).$
Consequently,
\(
\dilatedWtwoGeodesic_\variablen(\variableTime_\variablen)=\PD_\variablen,\; \lim_{\variableTime\to \variableTime_{\variablen+1}}  \dilatedWtwoGeodesic_\variablen(\variableTime)=\PD_{\variablen+1},
\)
and \(\dilatedWtwoGeodesic_\variablen\) is continuous on \(\bigl[\variableTime_\variablen,\variableTime_{\variablen+1}\bigr)\) because geodesics are continuous by definition.

\subsubsection{Piecewise–geodesic application}

Gluing the segments together yields the piecewise–geodesic application
\[
\TVPDf:[\variableTime_0,\variableTime_{\variableN_\PDSeq})\,\to\PDS,
\TVPDf(\variableTime)=
\begin{cases}
	\dilatedWtwoGeodesic_0(\variableTime) \quad \quad  \forall \,\variableTime\in[\variableTime_0,\variableTime_1),\\[4pt]
	\dilatedWtwoGeodesic_1(\variableTime) \quad \quad \forall \,\variableTime\in[\variableTime_1,\variableTime_2),\\
	\;\vdots & \\
	\dilatedWtwoGeodesic_{\variableN-1}(\variableTime) \,\,\,\,\,\, \!\forall \,\variableTime\in[\variableTime_{{\variableN_\PDSeq}-1},\variableTime_{\variableN_\PDSeq}),
\end{cases}
\]
and we define $\TVPDf(\variableTime_{\variableN_\PDSeq})=\PD_{\variableN_\PDSeq}$. By construction, $\TVPDf$ is continuous on $[\variableTime_0,\variableTime_{\variableN_\PDSeq}]$. Therefore, if $\TVPDf(\variableTime)\notin \boundarySet$ for every$ \,\variableTime\in[\variableTime_0,\variableTime_{\variableN_\PDSeq}\,]$, then $\TVPDf\in \TVPDspace^{\paramDelta}$ (but also $\TVPDf\in \CTVPDspace^{\paramDelta})$, and consequently $\TVPDf$ is a \TVPDP in the sense of our definition. We refer to such \TVPDsPP, derived from our input data type, as input \TVPDsP (see \autoref{fig:input_TVPD}). We denote by $\ITVPDspace^{\paramDelta}$ the set of input \TVPDsPP, and we have $\ITVPDspace^{\paramDelta}\subsetneq \CTVPDspace^{\paramDelta}\subsetneq \TVPDspace^{\paramDelta}$. It should be noted that the information contained in $\TVPDf$ includes not only the sequence of timed-persistence diagrams $\PDSeq_\variablen=\bigl((\PD_\variablen,\variableTime_\variablen)\bigr)_{{0\leq \variablen\leq \variableN_\PDSeq}}$, but also the intermediate interpolations $\{ \WtwoGeodesic_\variablen\!\bigl(\lambda_\variablen(\variableTime)\bigr)\in \PDS,\variablen\in\{0,\dots,\variableN_\PDSeq-1\},\variableTime\in (0,1)\,\}$. The application $\TVPDf$ is therefore richer in information than the original sequence. 

\begin{figure}
	\centering
	\includegraphics[width=\linewidth]{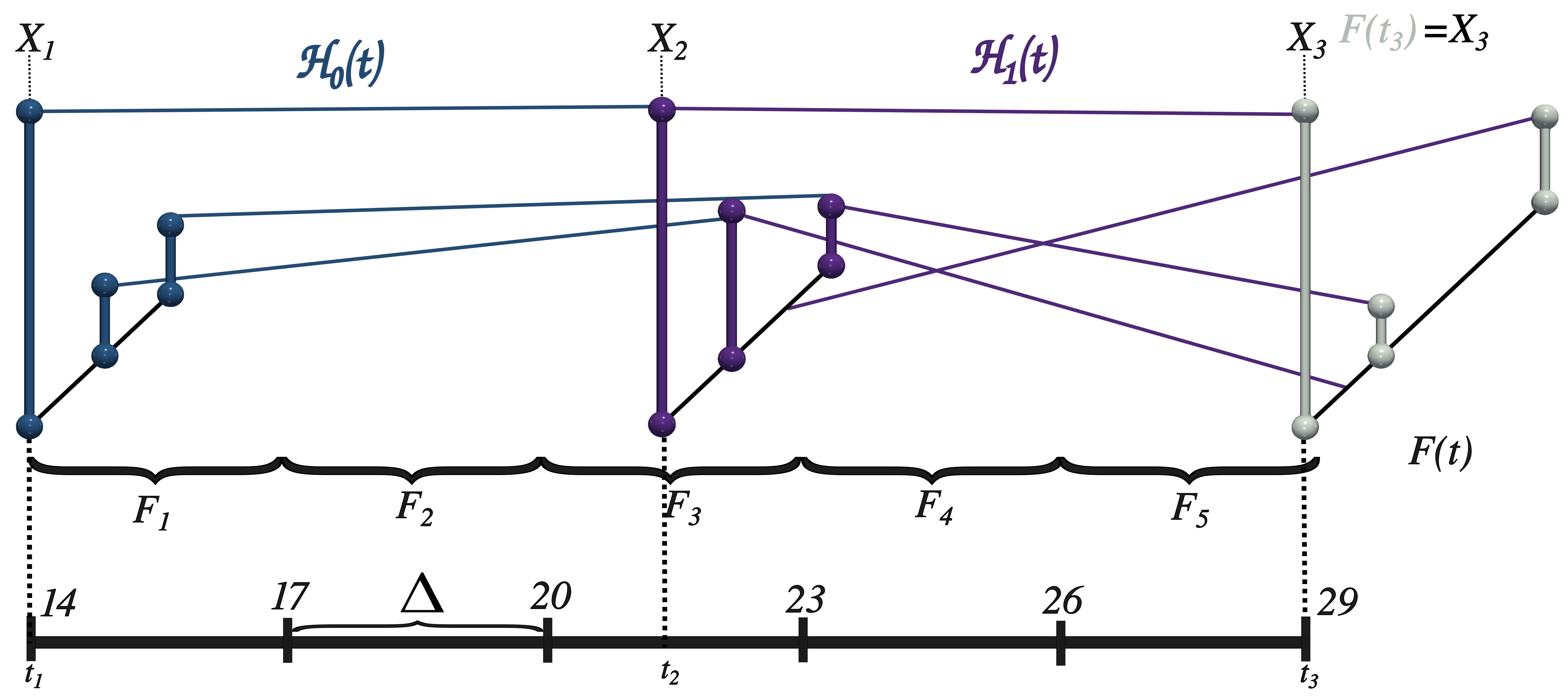}
	\caption{Fix $\boundarySet = \{\PD_{\varnothing}\}$. Illustration of constructing an input \TVPDP $\TVPDf$ from a sequence of timed persistence diagrams $\PDSeq_{\variablen} = \bigl((\PD_{1},\variableTime_{1}),(\PD_{2},\variableTime_{2}),(\PD_{3},\variableTime_{3})\bigr)$. In blue, the dilated $\wasserstein_{2}$-geodesic $\dilatedWtwoGeodesic_{0}(\variableTime)=\WtwoGeodesic_{0}\!\bigl(\lambda_{0}(\variableTime)\bigr)$ joining $\PD_{1}$ to $\PD_{2}$ for $\variableTime\in[\variableTime_{1},\variableTime_{2})$. In purple, the dilated $\wasserstein_{2}$-geodesic $\dilatedWtwoGeodesic_{1}(\variableTime)=\WtwoGeodesic_{1}\!\bigl(\lambda_{1}(\variableTime)\bigr)$ joining $\PD_{2}$ to $\PD_{3}$ for $\variableTime\in[\variableTime_{2},\variableTime_{3})$. In gray, $\TVPDf(\variableTime_{3})$ is define to $\PD_{3}$. The input \TVPDP $\TVPDf=(\TVPDf_{1},\TVPDf_{2},\TVPDf_{3},\TVPDf_{4},\TVPDf_{5})$ is then obtained by gluing the blue, purple, and gray parts (that is, $\TVPDf(\variableTime)=\dilatedWtwoGeodesic_{0}(\variableTime)\text{ for }\variableTime\in[\variableTime_{1},\variableTime_{2}),\TVPDf(\variableTime)=\dilatedWtwoGeodesic_{1}(\variableTime)\text{ for }\variableTime\in[\variableTime_{2},\variableTime_{3})\,,\TVPDf(\variableTime_{3})=\PD_{3})$.}\label{fig:input_TVPD}
\end{figure}

\subsection{Piecewise-constant approximation of input \TVPDsP for practical \CEDP computation}\label{sec:pc_TVPD}

 In order to
 \julien{ease the practical computation of}
 $\CEDM^{\paramDelta}_{\paramWeight,\paramPenalty}$, and later
$\CEDM^{\paramDelta}_{\paramWeight,1}$-geodesic\julien{s},
we introduce in this subsection piecewise-constant
approximations of
\TVPDsP (see \autoref{fig:piecewise-constant_TVPD}).


Let $\TVPDf :dom\,\TVPDf = \bigcup_{\variablei \in \{1, \dots, \variableN_\TVPDf^\paramDelta\}} \intervalI_\variablei^\TVPDf \to \PDS$, be one such continuous \TVPDP of $\TVPDspace^\paramDelta$, that is $\TVPDf\in \CTVPDspace^\paramDelta$. Let $\paramIntervalSize \,\vert \, \paramDelta$, and $\variableM\in\NNB$, such that for every $\variablei \in \{1, \dots, \variableN_\TVPDf^\paramDelta\}$, $\inf \intervalI_\variablei^\TVPDf + \variableM \cdot\paramIntervalSize=\sup \intervalI_\variablei^\TVPDf$, and denote $\tilde{\TVPDf_{\variablei}^{\paramIntervalSize}}:{\intervalI_\variablei^\TVPDf}\to\PDS$ the piecewise-constant application defined by $\tilde{\TVPDf_{\variablei}^{\paramIntervalSize}}(\variableTime)\;=\;\lim_{\variableTime\to \intervalBoundaBis} \TVPDf_{\variablei}(\variableTime),
\quad \forall \,\variableTime\in[\intervalBoundaBis,\intervalBoundaBisBis)\,\cap\,\intervalI_\variablei^\TVPDf,\,$ with $\intervalBoundaBis \;=\; \inf \intervalI_\variablei^\TVPDf + \variablen \cdot\paramIntervalSize,\,\text{for } \variablen=0,\dots,\variableM$. Gluing all the $\tilde{\TVPDf_{\variablei}^{\paramIntervalSize}}$ together, for every $ \variablei \in\,$$ \{1, \dots, \variableN_\TVPDf^\paramDelta\}$, we obtain the piecewise-constant approximation $\tilde{\TVPDf}^{\paramIntervalSize}:dom\,\TVPDf\to\PDS$ of $\TVPDf$, with $\tilde{\TVPDf}^{\paramIntervalSize}\in \TVPDspace^\paramDelta$ as a simple function on each $\intervalI_\variablei^\TVPDf$. As a result \sebastien{(see Appendix \ref{prop:piecewise-constant_TVPD_approximation})}, the subset of $\TVPDspace^\paramDelta$ consisting of piecewise-constant \TVPDsPP, denoted $\PCTVPDspace^\paramDelta$, is dense in the set $\CTVPDspace^\paramDelta$ of continuous \TVPDsPP. Indeed, for every $\TVPDf\in \CTVPDspace^\paramDelta$ and any $\scalarVariableOne>0$, there exists $\paramIntervalSize(\scalarVariableOne)\vert\paramDelta$ such that $\CEDM^{\paramDelta}_{\paramWeight,\paramPenalty}(\TVPDf,\tilde{\TVPDf}^{\paramIntervalSize(\scalarVariableOne)})<\scalarVariableOne$. Because $\ITVPDspace^{\paramDelta}\subsetneq \CTVPDspace^\paramDelta$, we have $\ITVPDspace^\paramDelta \subsetneq \CTVPDspace^\paramDelta\subsetneq \overline{\PCTVPDspace^\paramDelta}$, and then we got the same conclusions for $\ITVPDspace^\paramDelta$. In practice, if \(\TVPDf\colon[\variableTime_0,\variableTime_{\variableN_\PDSeq}]\to\PDS\)
is an input \TVPDP obtained from a timed persistence diagram sequence
\(\PDSeq_\variablen=((\PD_\variablen,\variableTime_\variablen))_{0\le \variablen\le \variableN_\PDSeq}\), it suffices to choose $ \paramIntervalSize(\scalarVariableOne)< \frac{\scalarVariableOne}{(1-\paramWeight)\cdot\,\variableKLips_\PDSeq\cdot\,(\variableTime_{\variableN_\PDSeq}-\variableTime_0)}$, with \(\variableKLips_\PDSeq \;=\;\max_{0\le \variablen\le \variableN_\PDSeq-1}\,\frac{\wasserstein_2(\PD_\variablen,\PD_{\variablen+1})}{\,\variableTime_{\variablen+1}-\variableTime_\variablen\,}\) \sebastien{(cf. Appendix \ref{prop:Lipschitz_TVPD})}.

Since $(\TVPDspace^\paramDelta,\CEDM^{\paramDelta}_{\paramWeight,\paramPenalty})$ is a metric space, applying the triangle inequality and using symmetry property, we obtain $\forall (\TVPDf,\TVPDg)\in (\TVPDspace^\paramDelta )^2,\, \vert \,\CEDM^{\paramDelta}_{\paramWeight,\paramPenalty}(\TVPDf,\TVPDg) -\CEDM^{\paramDelta}_{\paramWeight,\paramPenalty}(\tilde \TVPDf^{\paramIntervalSize(\scalarVariableOne)},\tilde \TVPDg^{\paramIntervalSize(\scalarVariableOne)}) \,\vert \leq \CEDM^{\paramDelta}_{\paramWeight,\paramPenalty}(\TVPDf,\tilde \TVPDf^{\paramIntervalSize(\scalarVariableOne)})+\CEDM^{\paramDelta}_{\paramWeight,\paramPenalty}(\TVPDg,\tilde \TVPDg^{\paramIntervalSize(\scalarVariableOne)})\leq 2\cdot\scalarVariableOne$. In conclusion, we can $(2\cdot\scalarVariableOne)$-approximate the value of $\CEDM^{\paramDelta}_{\paramWeight,\paramPenalty}$ between two input \TVPDsP $\TVPDf$ and $\TVPDg$ (or, more generally, between two continuous \TVPDsPP), by the value of $\CEDM^{\paramDelta}_{\paramWeight,\paramPenalty}$ computed between their piecewise-constant approximations $\tilde \TVPDf^{\paramIntervalSize(\scalarVariableOne)}$ and $\tilde \TVPDg^{\paramIntervalSize(\scalarVariableOne)}$. On a practical level, one may choose $\paramDelta$ to be as small as can be handled computationally, and set a common $\paramIntervalSize = \paramDelta$ for all the \TVPDsP whose pairwise $\CEDM^{\paramDelta}_{\paramWeight,\paramPenalty}$ distances are to be computed, since it is advantageous for both of these parameters to be small. We then have $\localDistance^\paramWeight_\paramDelta(\tilde \TVPDp_{\variablei}^{\paramIntervalSize}, \tilde \TVPDq_{\variablej}^{\paramIntervalSize})=\paramDelta \cdot \Bigl( \wasserstein_{2}\bigl(\TVPDp_{\variablei}(\intervalBounda),\TVPDq_{\variablej}(\intervalBoundc)\bigr)(1-\paramWeight)+\paramWeight\cdot \vert \intervalBoundc-\intervalBounda \vert\Bigr)$, and $\localDistance^\paramWeight_\paramDelta(\tilde \TVPDp_\variablei^{\paramIntervalSize}, \boundarySet) = \paramDelta \cdot (1-\paramWeight)\cdot \wasserstein_{2}\bigl(\TVPDp_{\variablei}(\intervalBounda),\boundarySet\bigr)$.

\begin{figure}
	\centering
	\includegraphics[width=\linewidth]{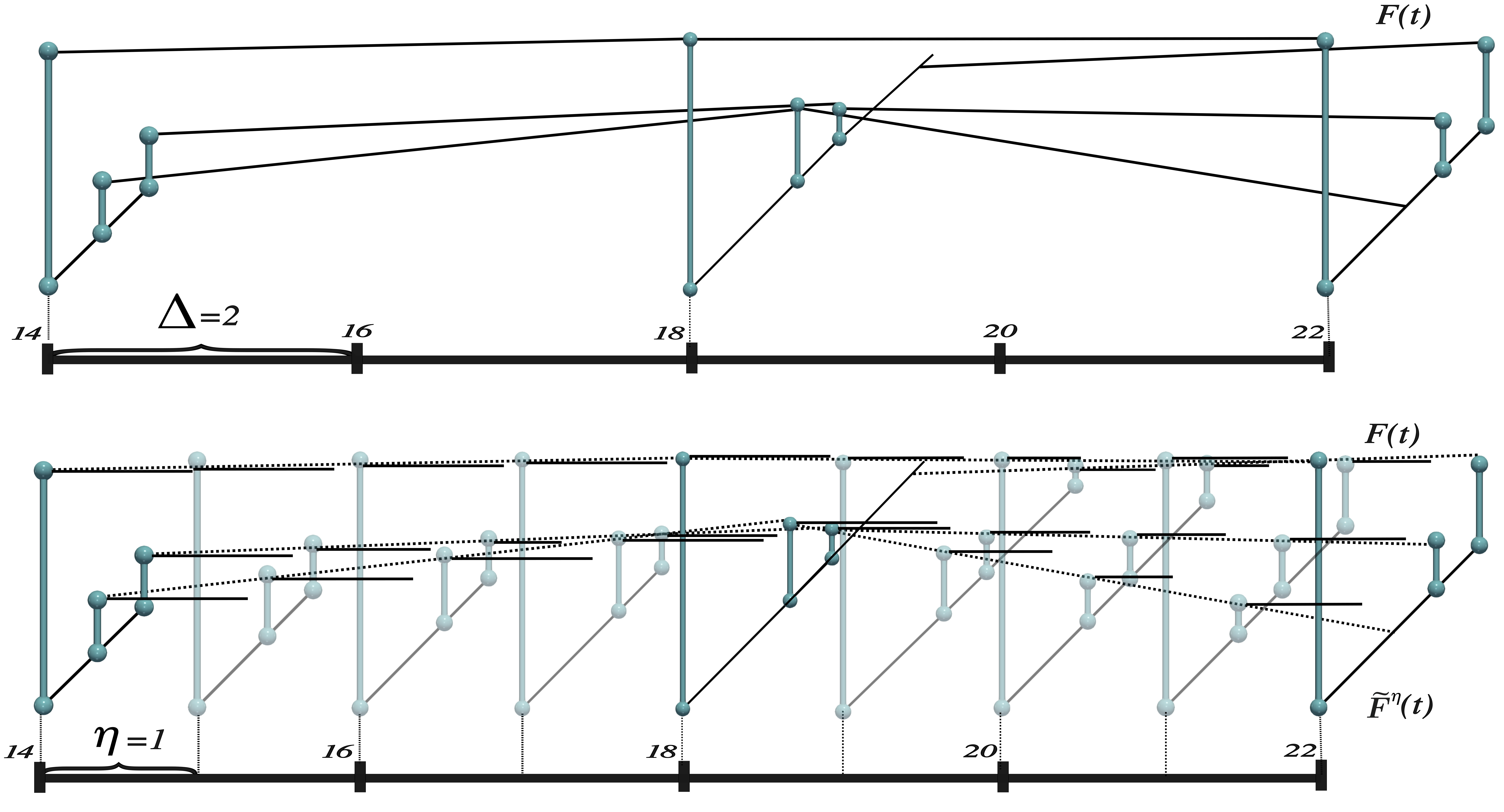}
	\caption{Illustration of the approximation procedure of a continuous \TVPDP
$\TVPDf\in \CTVPDspace^{\paramDelta}$ \julien{(top)} by its piecewise-constant
\TVPDP $\tilde{\TVPDf^{\paramIntervalSize}}\in \PCTVPDspace^{\paramDelta}$ of
parameter $\paramIntervalSize=1$ \julien{(bottom)}.
\julien{A}
smaller $\paramIntervalSize$ produces a correspondingly better
approximation.}\label{fig:piecewise-constant_TVPD}
\end{figure}

%% file: continuous_edit_distance_geodesics_between_time_varying_persistence_diagrams_Section4_.tex
    \section{Continuous edit distance geodesics}
\label{sec:CED_geodesic}

\noindent This section formalizes geodesics for the
$\CEDM^{\paramDelta}_{\paramWeight,1}$ metric on \TVPDsPP. We
\julien{discuss its}
existence under mild assumptions and describe an explicit three-step
construction.

\subsection{Definition and overview}\label{sec:overview_geodesic}

In a metric space \((\metricSpace,\metricDistance)\), a geodesic joining two points $(\variablexBisBis,\variableyMetricSpace)\in \metricSpace^{\,2}$ is a continuous application \(\metricGeodesic : [0,\metricDistance\bigl(\variablexBisBis,\variableyMetricSpace\bigr)]\to\metricSpace\) such that, $\metricGeodesic(0)=\variablexBisBis$, $\metricGeodesic\bigl(\metricDistance(\variablexBisBis,\variableyMetricSpace)\bigr)=\variableyMetricSpace$, and \(\metricDistance\bigl(\variablexBisBis,\variableyMetricSpace\bigr)=\sup \sum\,_{\variablei=0}^{\variablek-1} \metricDistance\,\bigl(\metricGeodesic(\variableTime_{\variablei}),\metricGeodesic(\variableTime_{\variablei+1})\bigr)\), where the supremum is taken over all $\variablek\in \NNB^{*}$, and all sequences $\variableTime_0=0< \variableTime_1 < \dots < \variableTime_\variablek=\metricDistance\bigl(\variablexBisBis,\variableyMetricSpace\bigr)$ in $[0,\metricDistance\bigl(\variablexBisBis,\variableyMetricSpace\bigr)]$. A metric space is said to be geodesic if every pair of points can be joined by at least one geodesic.

As a result, if $\boundarySet$ closed and $\forall \variablexBisBis \in \PDS, \, \{\variableyMetricSpace  \in \boundarySet, \metricDistance(\variablexBisBis,\variableyMetricSpace)=\metricDistance(\variablexBisBis,\boundarySet) \}$ is non-empty, then $(\Subdspace^\paramDelta,\localDistance^\paramWeight_{\paramDelta})$ and $\bigl(\TVPDspace^{\paramDelta},\CEDM^{\paramDelta}_{\paramWeight,1}\bigr)$ are geodesics \sebastien{(see Appendix \ref{lem:delta-subdivision_geodesic_space} and \ref{thr:final})}

Many conditions allow the hypothesis, $\forall \variablexBisBis \in \PDS, \, \{\variableyMetricSpace  \in \boundarySet, \metricDistance(\variablexBisBis,\variableyMetricSpace)=\metricDistance(\variablexBisBis,\boundarySet) \}$ is non-empty, to be satisfied. Some examples are : \boundarySet \, compact; $\forall \variablexBisBis \in \PDS, \boundarySet\cap \ball(\variablexBisBis,\rayon+\epsilon)$ relatively compact $\bigl($with any $\epsilon \in \RNB^*_+$, and $\rayon =\wasserstein_2(\variablexBisBis,\boundarySet)\bigr)$.

Intuitively, a \CEDP geodesic transforms a \TVPDP $\TVPDp$ into a \TVPDP $\TVPDq$ by performing, in continuous time, the elementary operations specified by their optimal $\paramDelta$-partial assignment $\partialAssignment\in\PASet^{\paramDelta}(\TVPDp,\TVPDq)$ : first, starting from $\TVPDp$, deleting unmatched $\paramDelta$-subdivisions of $\TVPDp$; then substituting each matched $\paramDelta$-subdivision of $\TVPDp$ with its matched counterpart in $\TVPDq$; and finally inserting the unmatched $\paramDelta$-subdivisions of $\TVPDq$, yielding $\TVPDq$.

\smallskip
\noindent\textbf{Setting.}
Fix two \TVPDsP \((\TVPDp,\TVPDq)\in (\TVPDspace^{\paramDelta})^2\) with optimal partial assignment \(\partialAssignment\in\PASet^{\paramDelta}(\TVPDp,\TVPDq)\). We decompose the index sets \(\{1,\dots,\variableN^\paramDelta_\TVPDp\}\) and \(\{1,\dots,\variableN^\paramDelta_\TVPDq\}\)  as
\[
\begin{aligned}
	&\mathcal \deletionSet_\TVPDp\;=\;\{\,\variablei\mid \variablei\notin\operatorname{dom}\partialAssignment\}\quad
	&&\text{(deletions)},\\
	&\mathcal \substitutionSet_\TVPDp\;=\;\operatorname{dom}\partialAssignment
	&&\text{(substitutions)},\\
	&\mathcal \insertionSet_\TVPDq\;=\;\{\,\variablej\mid \variablej\notin\operatorname{Im}\partialAssignment\}
	&&\text{(insertions).}
\end{aligned}
\]

The total \CEDP cost splits accordingly
\[
\variableL \;=\; \CEDM^{\paramDelta}_{\paramWeight,1}(\TVPDp,\TVPDq)
\;=\;\variableL_{\deletionSet}+\variableL_{\substitutionSet}+\variableL_{\insertionSet},
\]
where
\(\variableL_{\deletionSet}= \sum_{\variablei\in\deletionSet_\TVPDp}
\localDistance^\paramWeight_{\paramDelta}(\TVPDp_\variablei,\boundarySet)\),
\(\variableL_{\substitutionSet}= \sum_{\variablei\in\substitutionSet_\TVPDp}
\localDistance^\paramWeight_{\paramDelta}(\TVPDp_\variablei,\TVPDq_{\partialAssignment(\variablei)})\),
and
\(\variableL_{\insertionSet}=\sum_{\variablej\in\insertionSet_\TVPDq}
\localDistance^\paramWeight_{\paramDelta}(\TVPDq_\variablej,\boundarySet)\). Define the break-points
\(\variableL_{0/3}=0,\;\variableL_{1/3}=\variableL_{\deletionSet},\;
\variableL_{2/3}=\variableL_{\deletionSet}+\variableL_{\substitutionSet},\;
\variableL_{3/3}=\variableL_{\deletionSet}+\variableL_{\substitutionSet}+\variableL_{\insertionSet}=\variableL\).

\subsection{Steps of the geodesic}

We construct a continuous path
\(\geodesicCED:[0,\variableL]\to \TVPDspace^{\paramDelta}\) by concatenating
\emph{three} uniformly-parameterised segments.

\subsubsection{First step of the geodesic
	(\(\variablel\in[0,\variableL_{1/3}]\))}

Starting from \(\geodesicCED(0)=\TVPDp\), only the \(\paramDelta\)-subdivisions indexed by \(\deletionSet_\TVPDp\) move, that is each \(\TVPDp_\variablei\,(\variablei\in\deletionSet_\TVPDp)\) follows a \(\localDistance^\paramWeight_{\paramDelta}\)-geodesic in \(\Subdspace^{\paramDelta}\) joining it to the set~\(\boundarySet\). All other \(\paramDelta\)-subdivisions of $\TVPDp$ stay fixed. At \(\variablel=\variableL_{1/3}\) every deleted \(\paramDelta\)-subdivision of $\TVPDp$ has collapsed onto \(\boundarySet\), yielding the intermediate \TVPDP
\[
\geodesicCED(\variableL_{1/3})=\TVPDp\;\setminus\;\bigl\{\TVPDp_\variablei\mid \variablei\in\mathcal \deletionSet_\TVPDp\bigr\}.
\]

\subsubsection{Second step of the geodesic \(\bigl(\variablel\in(\variableL_{1/3},\variableL_{2/3}]\bigr)\)} 

From \(\geodesicCED(\variableL_{1/3})\) we simultaneously transport each remaining \(\paramDelta\)-subdivision \(\TVPDp_\variablei\,(\variablei\in\substitutionSet_\TVPDp)\) along a $\localDistance^\paramWeight_{\paramDelta}$-geodesic to its counterpart \(\TVPDq_{\partialAssignment(\variablei)}\). At \(\variablel=\variableL_{2/3}\) we reach
\[
\geodesicCED(\variableL_{2/3})=
\bigl\{\TVPDq_{\partialAssignment(\variablei)}\mid \variablei\in\substitutionSet_\TVPDp\bigr\},
\]
that is, a sub-\TVPDP of \(\TVPDq\) lacking the subdivisions indexed by \(\insertionSet_\TVPDq\).

\subsubsection{Third step of the geodesic \(\bigl(\variablel\in(\variableL_{2/3},\variableL_{3/3}]\bigr)\)}

Finally, for each \(\variablej\in\insertionSet_\TVPDq\) we
``spawn'' the \(\paramDelta\)-subdivision \(\TVPDq_\variablej\) out of~\(\boundarySet\), indeed each  \(\variablej\in\insertionSet_\TVPDq\) is formed by following a $\localDistance^\paramWeight_{\paramDelta}$-geodesic from $\boundarySet$ to $\TVPDq_\variablej$. At \(\variablel=\variableL_{3/3}=\variableL\) all insertions have finished and
\(\geodesicCED(\variableL_{3/3})=\TVPDq\).

\medskip

Then \(\geodesicCED\) is a geodesic
joining \(\TVPDp\) to \(\TVPDq\).
The geodesic is not unique in general, but any optimal assignment
yields at least one \CEDPP-geodesic that follows the
\textbf{delete $\;\to\;$ substitute $\;\to\;$ insert} paradigm
described above (see \autoref{fig:CED_Geodesic}).

\begin{figure*}[!t]  
	\centering
	\includegraphics[width=\linewidth]{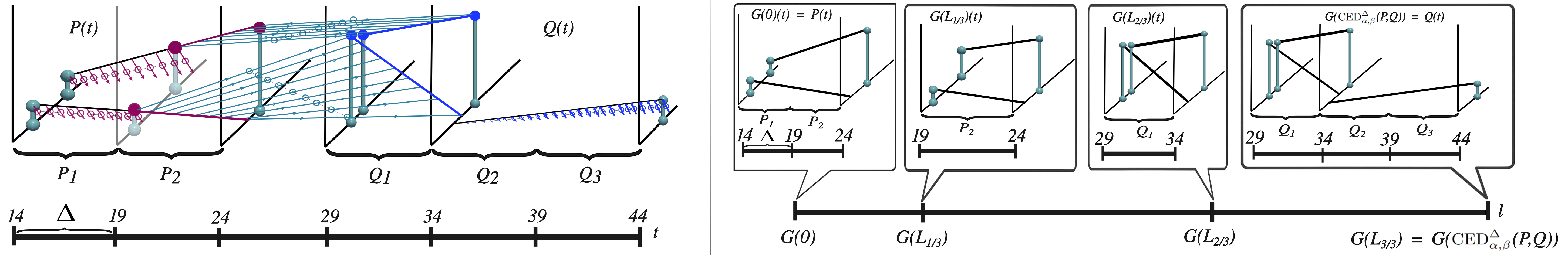}
	\caption{Fix $\boundarySet=\{\PD_{\varnothing}\}$. Left: The figure schematically illustrates the $\CEDM_{\paramWeight,1}^{\paramDelta}$-geodesic between two \TVPDsP $\TVPDp=(\TVPDp_1,\TVPDp_2)$ and $\TVPDq=(\TVPDq_1,\TVPDq_2,\TVPDq_3)$, the optimal $\paramDelta$‑partial assignment being the map $\partialAssignment\colon\operatorname{dom}\partialAssignment=\{2\}\!\to\!\{1\}$. Starting from \TVPDp \,$\bigl($that is, $\geodesicCED(0)=\TVPDp\bigr)$, we first perform the deletion step (in purple on the figure): $\TVPDp_1$ is unmatched under $\partialAssignment$ and therefore moves continuously toward $\boundarySet$ (the persistence diagrams constituting $\TVPDp_1$ head toward the empty diagram), while $\TVPDp_2$ remains fixed. At the end of this step only $\TVPDp_2$ is left $\bigl($that is, $\geodesicCED(\variableL_{1/3})=\TVPDp_2\bigr)$. Second, during the substitution step (in blue on the figure): the $\paramDelta$‑subdivision $\TVPDp_2$ is matched to $\TVPDq_1$ by $\partialAssignment$, then $\TVPDp_2$ travels along the $\localDistance_{\paramDelta}^{\paramWeight}$-geodesic to $\TVPDq_1$. When this motion ends, the intermediate \TVPDP reduces to the single $\paramDelta$‑subdivision $\TVPDq_1$ $\bigl($that is, $\geodesicCED(\variableL_{2/3})=\TVPDq_1\bigr)$. Lastly, during the insertion step (in gray on the figure): the $\paramDelta$-subdivisions $\TVPDq_2$ and $\TVPDq_3$, which are unmatched by $\partialAssignment$, are inserted from $\boundarySet$ while $\TVPDq_3$ stays fixed. After both insertions are complete, the geodesic reaches the target \TVPDP $\TVPDq=(\TVPDq_1,\TVPDq_2,\TVPDq_3)$ $\bigl($that is, $\geodesicCED(\variableL_{3/3})=\TVPDq\bigr)$. Thus the $\CEDM_{\paramWeight,1}^{\paramDelta}$-geodesic from $\TVPDp$ to $\TVPDq$ follows the delete $\rightarrow$ substitute $\rightarrow$ insert scheme encoded by the optimal assignment $\partialAssignment$. Right: The figure illustrates the intermediate \TVPDsP on the geodesic $\geodesicCED$ between $\TVPDp$ and $\TVPDq$ at the main stages, as the variable $\variablel$ varies from $0$ to $\CEDM_{\paramWeight,1}^{\paramDelta}(\TVPDp,\TVPDq).$}\label{fig:CED_Geodesic}
\end{figure*}

%% file: continuous_edit_distance_barycenters_of_time_varying_persistence_diagrams_Section5_.tex
	\section{Continuous edit distance barycenters}
\label{sec:barycenter_CED}

\noindent In this section we seek to minimize the
$\CEDM_{\paramWeight,1}^{\paramDelta}$-Fréchet energy of a \TVPDP $\TVPDX$
w.r.t. a sample $\{\sampleTVPD_1,\dots,\sampleTVPD_\variablen\}$:

$$
\frechetEnergy(\TVPDX)=\sum_{\variablei=1}^\variablen \CEDM_{\paramWeight,1}^{\paramDelta}(\TVPDX,\sampleTVPD_\variablei)^2.
$$

To this end, we use two simple, practical schemes that update $\TVPDX$ along \CEDPP-geodesics and keep the best candidate seen so far.

\textbf{\textit{Stochastic geodesic descent:}} Let the initial step size be
$\paramSGD\in(0,1)$; decrease it linearly over the first $\lfloor \paramGD/2
\rfloor$ iterations down to $0.1\cdot\paramSGD$, then hold it constant until
iteration $\paramGD$.

\noindent\textbf{(i)} Initialize by \sebastien{choosing} $\variablei\in\{1,\dots,\variablen\}$ \sebastien{randomly} and set $\TVPDX\leftarrow \sampleTVPD_\variablei$; record $\candidate\leftarrow \TVPDX$ as the current best.
  
\noindent\textbf{(ii)} Iterate: sample $\variablej\in\{1,\dots,\variablen\}$ \sebastien{randomly}; move $\TVPDX$ a \CEDPP-geodesic step of length $\paramSGD \cdot \CEDM^{\paramDelta}_{\paramWeight,1}(\TVPDX,\sampleTVPD_\variablej)$ from $\TVPDX$ toward $\sampleTVPD_\variablej$. If $\frechetEnergy(\TVPDX)<\frechetEnergy(\candidate)$, update $\candidate\leftarrow \TVPDX$. Decrease $\paramSGD$.
  
\noindent\textbf{(iii)} Stop when the relative energy decrease over the last $\paramGGDBis$ iterations is $<1\%$, or after a fixed iteration cap $\paramGD$; return $\candidate$.

\smallskip

\textbf{\textit{Greedy geodesic descent:}}
Let the initial step size be $\stepGre=1/\paramGGD$, with $\paramGGD\in\NNB, \paramGGD>1$.

\noindent\textbf{(i)} Initialize by \sebastien{choosing} $\variablei\in\{1,\dots,\variablen\}$ \sebastien{randomly} and set $\TVPDX\leftarrow \sampleTVPD_\variablei$; record $\candidate\leftarrow \TVPDX$ as the current best.

\noindent\textbf{(ii)} Iterate: for every $\variablej\in\{1,\dots,\variablen\}$, sample candidates $\candidateTwo_{\variablej,\paramGDBis}$ ($\paramGDBis\in\{0,\dots,\paramGGD\}$) along the \CEDPP-geodesic from $\TVPDX$ to $\sampleTVPD_\variablej$ with step $\paramGDBis\cdot\stepGre$; let $\TVPDX$ be the candidate with the smallest energy among all samples $\candidateTwo_{\variablej,\paramGDBis}$. If $\frechetEnergy(\TVPDX)<\frechetEnergy(\candidate)$, set $\candidate\leftarrow \TVPDX$.
  
\noindent\textbf{(iii)} Stop when the relative energy decrease over the last iteration is $<1\%$, or after a fixed iteration cap $\paramGD$; return $\candidate$.

\smallskip

Both schemes are iterative and use monotone acceptance (the recorded best energy is non-increasing).

%% file: applications_Section6_.tex
      \section{Applications}\label{sec:applications}
    
    \noindent We illustrate two representative utilizations of our framework:
(i)
    \julien{temporal}
    pattern tracking via matching and (ii) topological clustering via
barycenters, which respectively leverage \julien{the} \CEDP and the \TVPDP
barycenter
\julien{computations.}

    \begin{figure*}
	\centering
	\includegraphics[width=\linewidth]{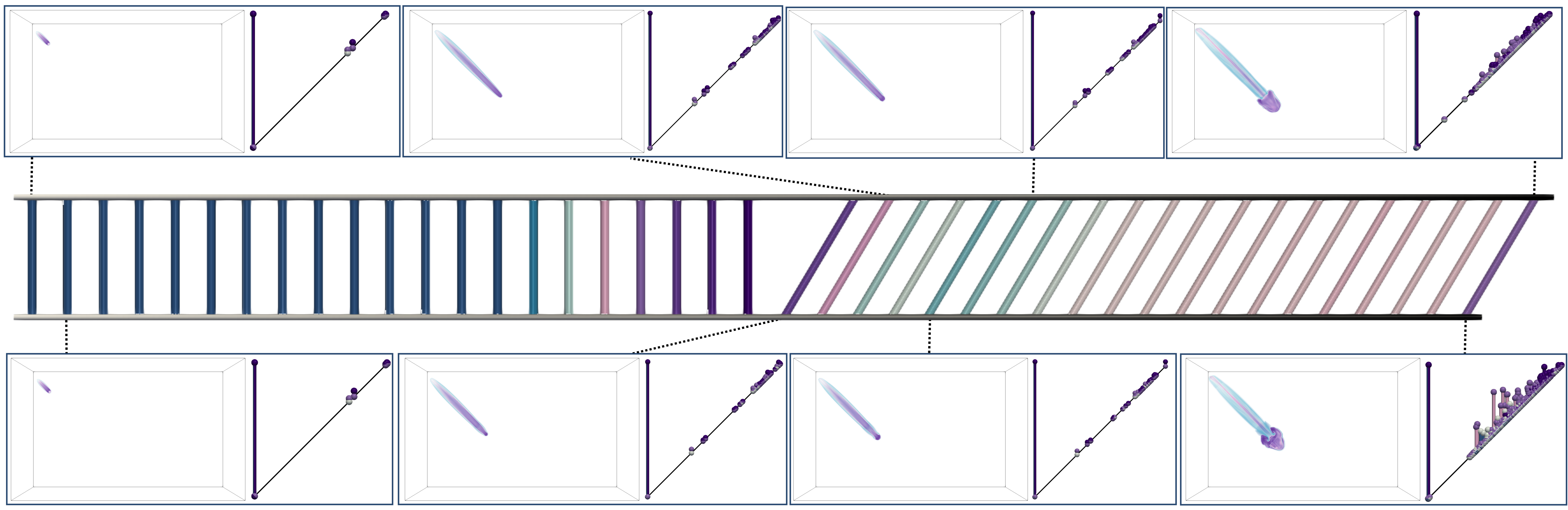}
	\caption{Temporal-shift recovery between two input \TVPDsPP, \emph{YB11} and
\emph{YC11}, from the \emph{Asteroid Impact} dataset. Top/bottom rows show
selected scalar fields and their persistence diagrams for \emph{YB11} (top) and
\emph{YC11} (bottom). The middle strip shows the CED alignment:
time is encoded in grayscale along each sequence (top=\emph{YB11},
bottom=\emph{YC11}); vertical connectors indicate the
\(\paramDelta\)-subdivision matchings, color-coded by assignment cost \sebastien{(from blue for low cost, through pink, to purple for high cost)}.
The explosion occurs at different times—between time steps 6241–6931 in
\emph{YB11} (top, 2nd–3rd snapshots) and 5335–6034 in \emph{YC11} (bottom,
2nd–3rd snapshots). The \CEDP alignment correctly pairs pre- and post-explosion
phases across the two sequences, thus recovering the temporal
shift.}\label{fig:dilation_matching}
    \end{figure*}

    \subsection{\julien{Temporal} pattern tracking via
matching}\label{sec:pattern_tracking_via_matching}
    
    \julien{The}
    \CEDP (\sebastien{\autoref{sec:continuous_edit_distance_between_TVPDs}}) relies on the
optimization of a partial assignment between two input \TVPDsPP. Then, the
resulting matchings can be used
    \julien{for}
    pattern tracking between \TVPDsPP. \autoref{fig:dilation_matching}
and \autoref{fig:subsequence} illustrate this in two ways on input \TVPDsP from the
asteroid impact dataset \cite{taylor2008scivis,imahorn2018visualization}: (i)
temporal-shift recovery by aligning two \TVPDsP with different event times,
shown in \autoref{fig:dilation_matching}; and (ii)
\julien{pattern}
search by aligning a
\TVPDP to one of its sub-\TVPDPP, shown in \autoref{fig:subsequence}. Above
all,
\julien{the}
\CEDP matchings
\julien{can be used as}
a visual comparison tool, allowing to represent where
the similarities lie within two \TVPDsPP. \sebastien{The performance of this
tracking
is discussed
\julien{in}
\sebastien{\autoref{sec:framework_quality}}}.

    \begin{figure}[!t]
	\centering
	\includegraphics[width=\linewidth]{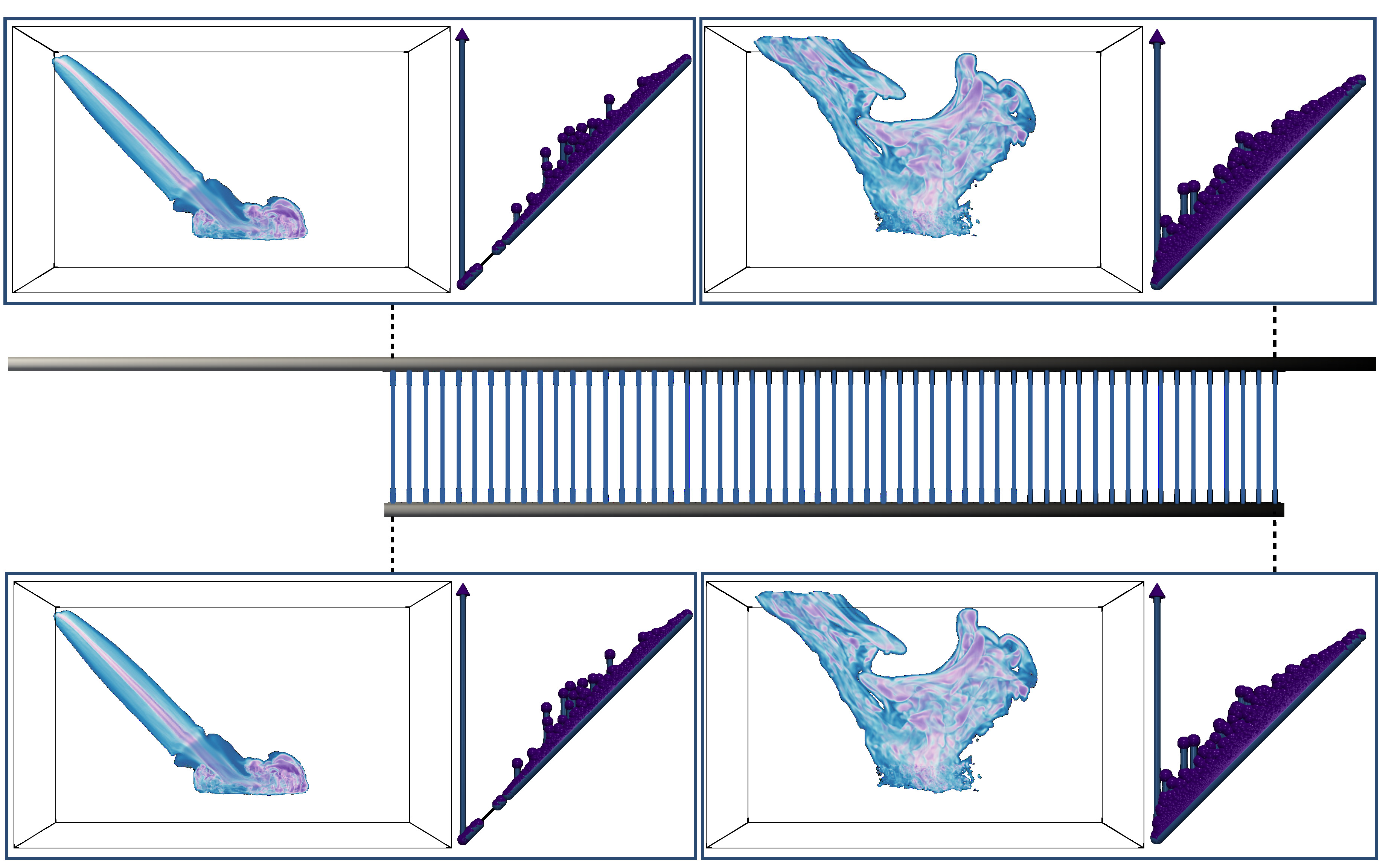}
	\caption{\julien{Temporal pattern}
	search between a \TVPDPP, \emph{YB11}, from the \emph{Asteroid Impact}
dataset and one of its sub-\TVPDsPP. The top and bottom rows show snapshots of
the scalar field and the corresponding persistence diagrams at selected time
steps (top: full \emph{YB11} \TVPDPP; bottom: candidate subsequence). The
central strip visualizes the alignment computed by \CEDPP: both sequences are
laid out along time (\emph{YB11} on top, subsequence on bottom), with time
encoded in grayscale along each sequence; vertical connectors indicate the
\(\paramDelta\)-subdivision matchings, colored by assignment cost (in
this figure uniformly blue, indicating zero cost because the target is an exact
subsequence). Because the target is a true subsequence of the source, the
alignment collapses to a one-to-one mapping on the selected window, effectively
synchronizing the \(\paramDelta\)-subdivisions (identical time stamps on both
sides).
    }\label{fig:subsequence}
    \end{figure}

    \subsection{Barycenters for topological clustering}\label{sec:barycenters_for_topological_clustering}
    
    Clustering partitions a dataset into subsets that are internally close and
mutually well-separated under a task-relevant distance $\metricDistance$. This
yields a principled coarse-graining of the space: it reduces complexity, and
exposes heterogeneity by delineating distinct regimes of behavior. When working
with topological signatures, each cluster captures a typical topological
behavior. To do so within the \TVPDP setting, we instantiate a $k$-means–style
scheme \cite{lloyd1982least,jain2010data} adapted to the \TVPDP geometry: the
centroid operator is given by our \TVPDP barycenter routine (stochastic or
greedy; \sebastien{\autoref{sec:barycenter_CED}}), while pairwise dissimilarities are
evaluated with the \CEDP distance
(\sebastien{\autoref{sec:continuous_edit_distance_between_TVPDs}}).
\autoref{tab:clustering} reports clustering results obtained with this
strategy—using both the stochastic and the greedy barycenter variants—on several
acquired datasets (sea-surface
height~\cite{vidal2019progressive},
\julien{VESTEC~\cite{flatken2023vestec}}
    and
asteroid impact). For comparison, we also report an MDS-based clustering
baseline: for each dissimilarity (L2, Fréchet, TWED, DTW, \CEDPP), we compute
the pairwise distance matrix between \TVPDsPP, apply \sebastien{2D} MDS to this matrix, and
then run $k$-means on the resulting embedding. \sebastien{The evaluation of the clustering performance is discussed at the end of \autoref{sec:framework_quality}.}

    \sebastienBis{\autoref{fig:VESTEC_kmeans} provides a qualitative view of the \CEDPP-based $k$-means on the sample of the four TVPDs from the VESTEC dataset. With $k=2$ and the stochastic barycenter variant, the algorithm separates the TVPDs into two groups that coincide with the ground truth (runs 1–2 vs.\ 3–4). In the MDS embedding induced by the $d^\alpha_\Delta$ distance, each cluster is organized around its CED barycenter, and the overlaid optimal partial matchings show that the TVPDs in a given group are consistently aligned with their centroid across time.}

    \begin{table}[t]
      \refstepcounter{table}\label{tab:clustering}%
      \noindent\textbf{Table \thetable.} Comparison of two clustering pipelines on real \TVPDP datasets (SSH, VESTEC, Asteroid Impact).
  \emph{MDS-clustering}: for each dissimilarities (L2, Fr\'echet, TWED, DTW, \CEDPP), we build the pairwise distance matrix,
  embed the data by classical MDS, then run $k$-means (with $k$ equal to the number of classes).
  \emph{\CEDPP-clustering}: our approach based on \CEDP with two barycenter optimizers (stochastic and greedy).
  Entries are numbers of misclassified subsequences w.r.t.\ the ground truth (lower is better);
  -- indicates not applicable.
  Within the MDS block, \CEDP is competitive or superior, and on \emph{Asteroid
Impact} it is the only
  \julien{dissimilarity}
  with zero error.
  For \CEDPP-clustering, one of the two variants reaches the ground truth on all datasets,
  whereas MDS-clustering makes an error on \emph{VESTEC} for every dissimilarities.
    
      \vspace{4pt}\centering\footnotesize
      \setlength{\tabcolsep}{3pt}
      \renewcommand{\arraystretch}{1.1}
    
      \begin{tabular*}{\columnwidth}{@{\extracolsep{\fill}} l ccccc | cc @{}}
        \toprule
        & \multicolumn{5}{c}{\textbf{MDS-clustering}}
        & \multicolumn{2}{c}{\textbf{\CEDPP-clustering}}\\
        \cmidrule(lr){2-6}\cmidrule(l){7-8}
        \textbf{Dataset} & L2 & Fr\'echet & TWED & DTW & \CEDP & \CEDPP-S & \CEDPP-G \\
        \midrule
        SSH             & 2 & \textbf{0} & \textbf{0} & \textbf{0} & \textbf{0} & \textbf{0} & \textbf{0} \\
        VESTEC          & 2 & 2 & 1 & 1 & 1 & \textbf{0} & \textbf{0} \\
        Asteroid impact & -- & 2 & 1 & 1 & \textbf{0} & \textbf{0} & 1 \\
        \bottomrule
      \end{tabular*}
    \end{table}
    
    \begin{figure}[!t]  
	\centering
	\includegraphics[width=\linewidth]{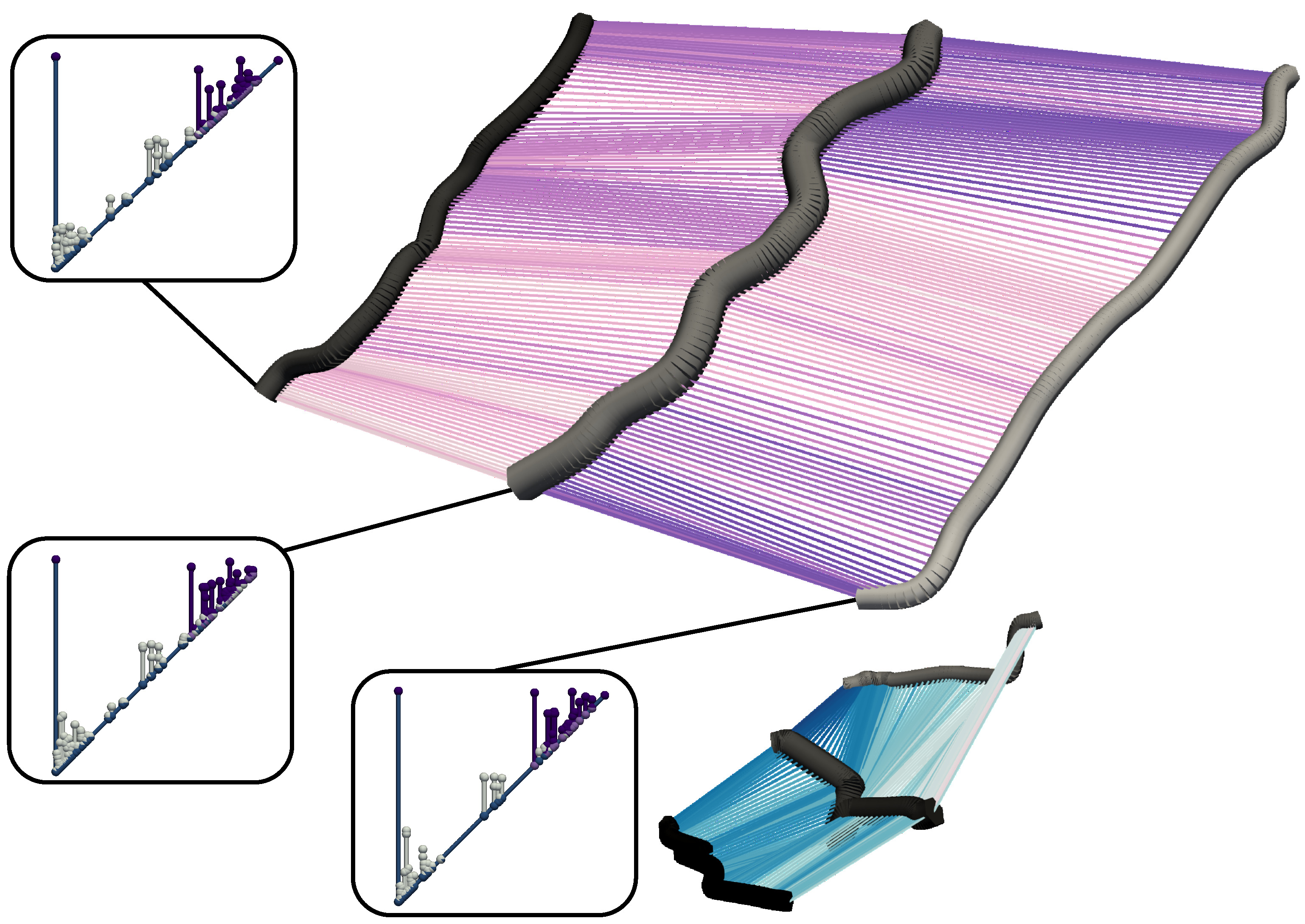}
	\caption{2-means clustering $(k=2)$ of the four VESTEC TVPDs
\julien{(stochastic barycenter routine)}.
Shown are the four TVPDs and the
two cluster centroids returned. For visualization, TVPDs are embedded in
$\mathbb{R}^3$ via an MDS embedding induced by $d^\alpha_\Delta$; For each TVPD,
time is encoded in grayscale, and its optimal \CEDP partial matching to the
assigned centroid is overlaid; with edge colors encoding assignment cost \sebastien{(from blue for low cost, through pink, to purple for high cost)}.
\sebastien{Insets show, for the first cluster and at the initial time step, the
persistence diagrams of the two input TVPDs and the persistence diagram of their
centroid, to which they are matched by CED.} The two clusters
\sebastien{correctly} match the ground truth (runs 1--2 vs.\ 3--4).}
    \label{fig:VESTEC_kmeans}
    \end{figure}

%% file: results_Section7_.tex
    \section{Results}

    \noindent This section reports experimental results executed on a
workstation equipped with an Intel Xeon CPU (2.9 GHz; 16 cores; 64 GB RAM) and
using TTK
    \cite{ttk17, ttk19, leguillou_tech24}
for persistence diagram computation
\cite{guillou2023discretemorsesandwichfast, leguillou_tpds25}
and matching \cite{vidal2019progressive}. Our method is implemented in
C++
as TTK modules. We performed
the experiments on an ensemble of simulated and acquired 2D/3D datasets—some
reused from prior work (sea-surface
height\sebastien{\cite{vidal2019progressive}},
VESTEC\sebastien{\cite{flatken2023vestec}}) and another drawn from the 2018
SciVis contest (asteroid
impact\sebastien{\cite{taylor2008scivis,imahorn2018visualization}}). For each
experiment whose results are reported in this section, we used
\(\boundarySet=\{\PD_{\varnothing}\}\).
Detailed specifications of these datasets are provided in Appendix C. The parameters \(\paramWeight\), \(\paramDelta\), \(\eta\), and \(\beta\) in all our experiments are set according to the specifications detailed in Appendix~D; these specifications can also be read as a practical guideline for choosing these parameters in applications of \(\CEDM^\paramDelta_{\paramWeight,\paramPenalty}\) to other TVPD datasets.

    \subsection{Framework quality}\label{sec:framework_quality}

    \julien{The}
    \CEDP is a metric (\sebastien{\autoref{sec:continuous_edit_distance_between_TVPDs}}). \autoref{fig:stability} empirically evaluates its robustness to additive
noise (with one temporal-only noise experiment and one spatial-only noise
experiment) for several values of parameter $\paramPenalty$. From a reference
time series of timed PL–scalar fields
$\timedScalarFieldsSeq=\bigl((\functionf_\variablen,
\variableTime_\variablen)\bigr)_{\variablen=0}^{\variableN_\timedScalarFieldsSeq
}$, we generate, for each experiment, 25 noisy sequences
$\timedScalarFieldsSeq^{(\scalarVariableOne)}$ with increasing values
$\scalarVariableOne$. For the temporal-only noise experiment, for each value of
$\scalarVariableOne$ and each
$\variablen\in\{0,\dots,\variableN_\timedScalarFieldsSeq\}$, we add uniform
noise of amplitude $\scalarVariableOne$ to $\variableTime_\variablen$, thereby
obtaining $\timedScalarFieldsSeq^{(\scalarVariableOne)}$. We then report
$\CEDM\big(\TVPDM(\timedScalarFieldsSeq),\,
\TVPDM(\timedScalarFieldsSeq^{(\scalarVariableOne)})\big)$ as a function of
$\scalarVariableOne$.
\julien{The}
curves
grow approximately linearly, indicating
    \julien{a}
    stability of
    \julien{the}
    \CEDP to temporal jitter for reasonable noise levels. For the
spatial-only noise experiment, for each value of $\scalarVariableOne$ and each
$\variablen\in\{0,\dots,\variableN_\timedScalarFieldsSeq\}$, we add
\julien{a}
uniform
noise of amplitude $\scalarVariableOne$ to the values of $\functionf_\variablen$
(that is, from $\functionf_\variablen$ a noisy version
$\functionf_{\variablen}^{\scalarVariableOne}$ is created such that $\vert \vert
\functionf_\variablen-\functionf_{\variablen}^{\scalarVariableOne}\vert \vert
\leq \scalarVariableOne$), yielding
$\timedScalarFieldsSeq^{(\scalarVariableOne)}$. As above, we report
$\CEDM\big(\TVPDM(\timedScalarFieldsSeq),\,
\TVPDM(\timedScalarFieldsSeq^{(\scalarVariableOne)})\big)$ versus
$\scalarVariableOne$; empirically, the dependence on $\scalarVariableOne$
unfolds in four regimes. For small noise amplitude ($\scalarVariableOne<12\%$),
the curve shows a near-linear baseline, the persistence pair birth and death
times (of persistence diagrams \CEDPP-matched at the same time step) shift
almost independently. For intermediate amplitudes ($12\%<\scalarVariableOne<20\%$),
the curve becomes visibly convex as the first combinatorial events (argmax
flips, elder inversions, changes of the killing cell) start to appear and
increase in frequency, which steepens the slope. A threshold-crossing kink then
emerges in a narrow window ($20\%<\scalarVariableOne<25\%$) when many pairs switch
almost simultaneously. Beyond this, for $\paramPenalty=1$ the growth becomes approximately linear as \sebastien{the typical $\wasserstein_2$ contribution between diagrams matched by CED stabilizes}, whereas for ($\paramPenalty<1$) additional slope breaks appear, coinciding with switches of
the \CEDP matching regime (substitutions vs. deletions/insertions). Across both
settings, $\paramPenalty$ tunes tolerance to noise—the slopes and breakpoints
shift as $\paramPenalty$ decreases. These two experiments illustrate that,
\julien{in practice, the}
\CEDP
varies smoothly and predictably under input perturbations, showing its
robustness to additive noise.
    
    \autoref{fig:convergence} illustrates the decrease and convergence of the
$\CEDM_{\paramWeight,1}^{\paramDelta}$-Fréchet energy during the iterative
computation of a \CEDP barycenter for 16 synthetic input \TVPDsPP,
using both the stochastic and the greedy variants. In both
cases, the energy decreases monotonically and stabilizes.

    \sebastien{Next, we illustrate on \autoref{fig:dilation_matching} the
recovery of a temporal shift between two input \TVPDsP from the asteroid impact
dataset (simulations of asteroid–ocean interactions with atmospheric airbursts).
The two runs, YB11 and YC11, share the same early scenario but the airburst
occurs markedly earlier in YC11 than in YB11. In the alignment strip,
\julien{the} \CEDP produces a coherent block of matchings that follows the
pre-explosion phase in both sequences, then shifts to align the post-explosion
regime, rather than simply matching snapshots with similar time stamps. This
behavior shows that \julien{the} \CEDP effectively recovers temporal shifts,
supporting its use for temporal alignment in \TVPDsPP.}
\autoref{fig:subsequence} illustrates the practical relevance of \julien{the}
\CEDP for motif search by aligning an input \TVPDP to one of its own
subsequences. The \CEDPP-induced alignment correctly retrieves the \TVPDP and
subsequence correspondence and synchronizes the $\paramDelta$-subdivisions
(matched subdivisions share identical time stamps), thereby validating \CEDP for
motif search in \TVPDsPP.

    \autoref{tab:clustering} reports the clustering results. On the MDS-based clustering baseline, \CEDP is consistently on par with, or superior to, competing dissimilarities; notably, on Asteroid Impact it is the only one to achieve zero errors with respect to the ground truth. With \CEDPP-clustering, the ground-truth partition is always recovered by at least one of the two variants (stochastic or greedy), whereas the MDS-based pipeline misclassifies on VESTEC for every dissimilaritie. Taken together, the MDS-based results support the use of \CEDP as a distance for \TVPDP analysis, and the \CEDPP-clustering results validate the stochastic and greedy variants as effective methods for clustering \TVPDsPP.

    \begin{figure}[!t]  
      \centering
      \includegraphics[width=\linewidth]{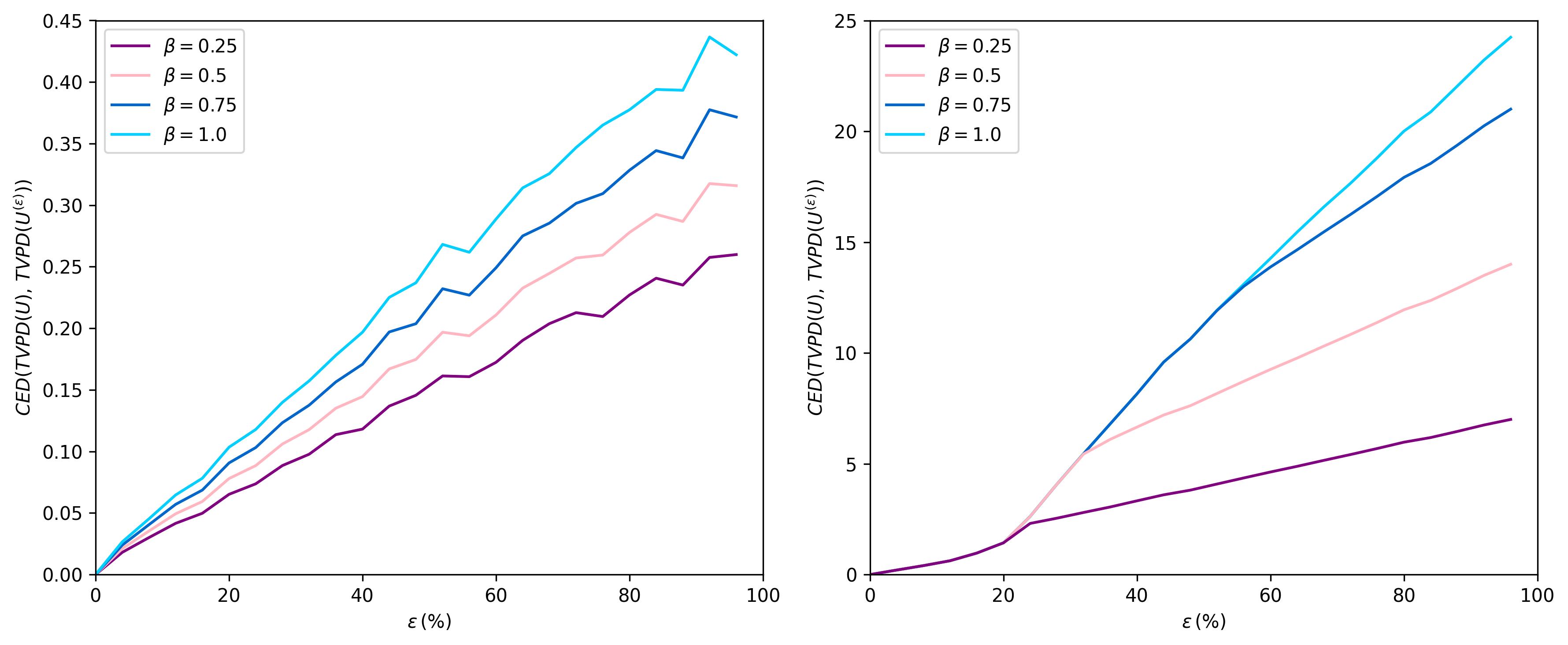}
      \caption{Empirical robustness of \julien{the} \CEDP to \julien{an}
additive noise \julien{$\scalarVariableOne$}: \emph{temporal-only} noise (left)
and \emph{spatial-only} noise (right), for several values of the parameter
\(\paramPenalty\).
\julien{For the temporal-only noise (left), the curves grow approximately
linearly. For the spatial-only noise (right), the distance increase follows
four phases (discussed in the main text).}
  Across both settings, \(\paramPenalty\) tunes the tolerance to noise---the
slopes and \julien{breakpoint shifts} as \(\paramPenalty\) varies.}
      \label{fig:stability}
    \end{figure}

    \begin{figure}[!t]  
	\centering
	\includegraphics[width=\linewidth]{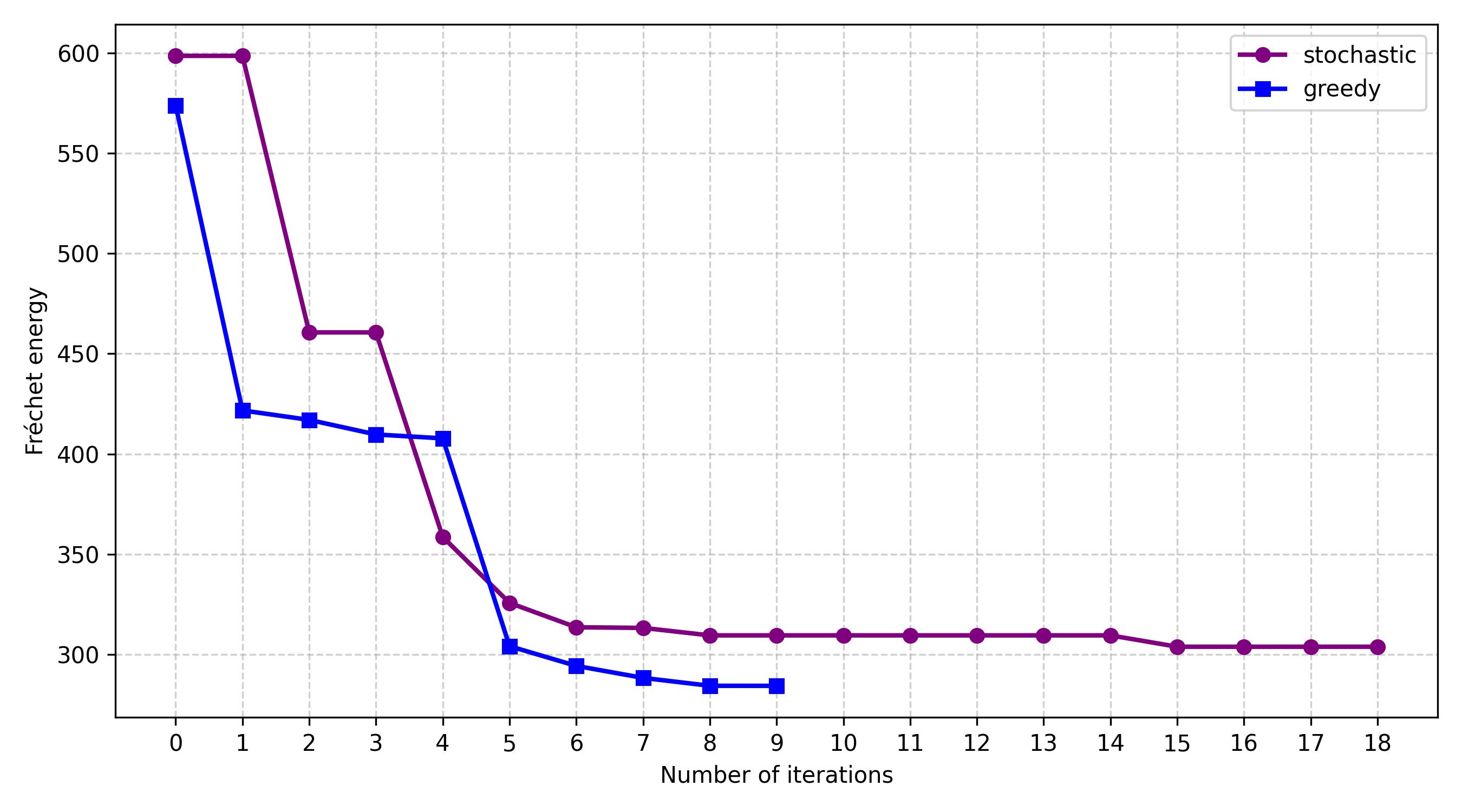}
	\caption{Convergence of the Fr\'echet energy across iterations during the
computation of a \CEDP barycenter for 16 synthetic input \TVPDsPP.
In both cases, the objective
decreases and then stabilizes, indicating convergence to a \CEDP barycenter of
the 16 \TVPDsP (returned as the final output).}\label{fig:convergence}
    \end{figure}
    
    \subsection{Time performance}
    
    \autoref{tab:sequential} reports the practical time performance of our
implementation for $\CEDM^{\paramDelta}_{\paramWeight,1}$ barycenter
computation, for both the stochastic and the greedy variants. We observe that
the running time depends on the number of input \TVPDsP in the sample,
$\sampleSize$, and on the average number of persistence pairs per input \TVPDPP,
$\averageNumberPersistencePairPerTVPD = \averageNumberDeltaS^{\paramDelta}\cdot
\averageNumberPersistencePair$ (where $\averageNumberDeltaS^{\paramDelta}$
denotes the average number of $\paramDelta$-subdivisions per input \TVPDP in the
sample and $\averageNumberPersistencePair$ the average number of persistence
pairs per $\paramDelta$-subdivision of its piecewise-constant approximation).
Indeed, each iteration of our barycenter algorithm requires $\sampleSize$
$\CEDM$-geodesic computations and $\sampleSize\!\cdot\!(\paramGGD-1)$ $\CEDM$
distance computations for the greedy version (for a step size $\paramSGD =
1/\paramGGD$), and a single $\CEDM$-geodesic computation for the stochastic
version. To carry out the dynamic programming
(\sebastien{\autoref{sec:CED_DPC}}) used to compute each
$\CEDM^{\paramDelta}_{\paramWeight,1}$ (or each
$\CEDM^{\paramDelta}_{\paramWeight,1}$ geodesic) between two input \TVPDsP
$\TVPDX$ and $\TVPDY$ from the sample, one must first compute the distance
$\localDistance^{\paramWeight}_{\paramDelta}$ between every pair of
$\paramDelta$-subdivisions of $\tilde \TVPDX^{\paramDelta}$ and $\tilde
\TVPDY^{\paramDelta}$. If $\persistencePairNumber$ denotes the number of
persistence pairs in a $\paramDelta$-subdivision of $\tilde
\TVPDX^{\paramDelta}$ and $\persistencePairNumberBis$ that of $\tilde
\TVPDY^{\paramDelta}$, then computing
$\localDistance^{\paramWeight}_{\paramDelta}$ between them with
    \julien{an exact assignment}
    algorithm takes
$\bigLandau\!\big((\persistencePairNumber{+}\persistencePairNumberBis)^3\big)$.
The computation of $\CEDM^{\paramDelta}_{\paramWeight,1}$ between $\TVPDX$ and
$\TVPDY$ is then performed by dynamic programming over the
$\paramDelta$-subdivisions of $\tilde \TVPDX^{\paramDelta}$ and $\tilde
\TVPDY^{\paramDelta}$, yielding a time complexity
$\bigLandau\!\big(\variableN_\TVPDX^{\paramDelta}\!\cdot
\variableN_\TVPDY^{\paramDelta}\big)$. Once
$\CEDM^{\paramDelta}_{\paramWeight,1}(\TVPDX,\TVPDY)$ has been computed, to
place oneself somewhere on their \CEDPP-geodesic costs at most
$\bigLandau\!\big((\variableN_\TVPDX^{\paramDelta}+
\variableN_\TVPDY^{\paramDelta})\cdot
\averageNumberPersistencePairPerTVPD\big)$. Consequently, the end-to-end cost of
one iteration is on the order of $\bigLandau\!\Bigl((\paramGGD\cdot
\sampleSize)\big((\averageNumberDeltaS^{\paramDelta})^2 \cdot (8
\averageNumberPersistencePair^{3} {+} 1)
+\averageNumberDeltaS^{\paramDelta}\cdot2\averageNumberPersistencePair\big)\Bigr
)$ for the greedy algorithm and
$\bigLandau\!\big((\averageNumberDeltaS^{\paramDelta})^2 \cdot (8
\averageNumberPersistencePair^{3} {+} 1)+
\averageNumberDeltaS^{\paramDelta}\cdot2 \averageNumberPersistencePair\big)$ for
the stochastic algorithm. This explains the increase in running time
observed for the stochastic variant from the VESTEC to the SSH dataset—its
complexity is independent of $\sampleSize$ and depends only on
$\averageNumberPersistencePairPerTVPD$, \sebastien{which grows by slightly less than a
factor of two, while the runtime grows by slightly more. By contrast, the greedy
variant exhibits a marked increase, consistent with its additional scaling in
$\sampleSize$}.

    \begin{table}[t]
  \refstepcounter{table}\label{tab:sequential}%
  \noindent\textbf{Table \thetable.}
  \julien{Running} times (in seconds, 5 run average) for the barycenter
computation using our \emph{greedy} and \emph{stochastic} variants.
  Here, \(\sampleSize\) is the number of input \TVPDsP in the sample, and \(\averageNumberPersistencePairPerTVPD = \averageNumberDeltaS^{\paramDelta}\!\cdot\! \averageNumberPersistencePair\) is the average number of persistence pairs per input \TVPDPP, where \(\averageNumberDeltaS^{\paramDelta}\) denotes the average number of \(\paramDelta\)-subdivisions per input \TVPDP and \(\averageNumberPersistencePair\) the average number of persistence pairs per \(\paramDelta\)-subdivision (in the piecewise-constant approximation). The execution times show that the greedy variant scales with both \(\averageNumberPersistencePairPerTVPD\) and \(\sampleSize\), whereas the stochastic variant scales primarily with \(\averageNumberPersistencePairPerTVPD\).

  \vspace{4pt}\centering
  \resizebox{\columnwidth}{!}{%
  \begin{tabular}{l r r r r r}
    \toprule
    \textbf{Dataset} & \textbf{$\sampleSize$} & \textbf{$\averageNumberDeltaS^{\paramDelta}$} &

\textbf{$\averageNumberPersistencePairPerTVPD=\averageNumberDeltaS^{\paramDelta}
\!\cdot\! \averageNumberPersistencePair$} & \textbf{\julien{Greedy}} &
\textbf{\julien{Stochastic}} \\
    \midrule
    Asteroid Impact    & 6 & 40.56  & 1,313,868 & 19,834.21  & 3,731.03 \\
    Sea Surface Height & 8 & 50  &  82,680  & 4,036.67  & 1,661.90 \\
    VESTEC             & 4 & 120 &  \sebastien{36,242}  & \sebastien{722.46}  & \sebastien{903.81} \\
    \bottomrule
  \end{tabular}%
  }
\end{table}

    \subsection{Limitations}

    A practical limitation of \CEDP is its runtime when used within a clustering
pipeline, which entails many barycenter evaluations—especially for samples with
a large average number of persistence pairs per \TVPDsP
$\averageNumberPersistencePairPerTVPD$. Indeed, as shown in
\autoref{tab:sequential}, barycenter times on acquired datasets are substantial
and scale roughly with $\averageNumberPersistencePairPerTVPD$ for both the
stochastic and the greedy variants. However, several strategies can drastically
reduce the cost of each \CEDP computation—and thus the total time required for
clustering. First, following Vidal
et~al.\cite{vidal2019progressive}, one can efficiently
approximate the Wasserstein distance between two persistence diagrams by
discarding pairs below a given persistence threshold, which significantly lowers
$\averageNumberPersistencePairPerTVPD$ (and therefore the barycenter runtime)
while preserving the signal carried by salient features. Second, locality
constraints such as the Sakoe–Chiba band \cite{keogh2005exact}, originally
proposed for DTW, can be directly adapted to \CEDPP: restricting the warping
path to a diagonal band and computing costs only inside it reduces the dynamic
programming complexity from $\bigLandau(n\cdot m)$ to $\bigLandau(n\cdot w)$
(band half-width $w$), leading to substantial speedups for the \CEDP
computation, and therefore the barycenter runtime—at the expense of possibly
missing large optimal
assignments.

\sebastien{Another limitation, of a theoretical nature but with practical implications, is that our construction of \(\CEDM^\paramDelta_{\paramWeight,\paramPenalty}\)-geodesics is currently restricted to the case $\paramPenalty = 1$; extending it to other values of $\paramPenalty$ is left for future work.}

\sebastien{\julien{Finally, another limitation}
is the reliance on a common subdivision step when 
comparing two input \TVPDsP via \julien{\CEDPP}. When their temporal extents 
differ significantly, any choice of shared subdivision step $\paramDelta$ 
induces a nontrivial lower bound on the CED between them, proportional (through 
the penalty parameter $\paramPenalty$) to the difference in their durations: a 
long \TVPDP $\TVPDp$ can only be matched to a much shorter one $\TVPDq$ up to a 
large number of unmatched $\paramDelta$-subdivisions of $\TVPDp$, each incurring 
a positive deletion cost. In such cases, the distance reflects not only 
differences in topological patterns but also differences in temporal length, 
which may or may not be desirable depending on the application. \sebastienBis{In scenarios where one wishes to discount such temporal-length effects, a proper value of $\beta$ needs to be adjusted. However, as mentioned above, this is only available for distance computation, and not for geodesics or barycenters.}
}

%% file: conclusion_Section8_.tex
\section{Conclusion}\label{sec:conclusion}

\noindent In this paper we presented the \emph{Continuous Edit Distance}, a geodesic, elastic distance for \TVPDsPP. We established metric and geodesic properties, derived efficient dynamic-programming routines to evaluate the \CEDPP, and provided an explicit \emph{delete\(\rightarrow\)substitute\(\rightarrow\)insert} construction of \CEDPP-geodesics. On top of this, we proposed two practical barycenter solvers—stochastic and greedy—with monotone Fréchet-energy decrease and simple stopping rules, and we released a C++ implementation within TTK.

Empirically, \CEDP is robust to additive perturbations (approximately linear response to temporal jitter; piecewise-linear to spatial noise), yields interpretable alignments that recover temporal shifts and enable motif search. When used within a $k$-means–style pipeline on acquired datasets (sea-surface height, VESTEC, asteroid impact), \CEDP delivers clustering quality on par with or superior to standard elastic dissimilarities, and our \CEDPP-based clustering attains the ground truth with at least one of the two barycenter variants. 

A practical limitation is runtime in clustering pipelines, where many barycenter evaluations are required for samples with large average numbers of persistence pairs. We outlined straightforward accelerations: persistence thresholding to shrink diagram sizes while preserving the signal carried by salient features, and locality constraints (e.g., Sakoe–Chiba bands) to reduce dynamic-programming cost. 

Future work includes multi–step–ahead forecasting with \CEDP as the training objective, enabling early-warning and counterfactual topological analysis of time-varying phenomena that commonly exhibit temporal dilation and shifts (e.g., ocean dynamics, atmospheric processes, functional brain connectivity). We also plan a practical \CEDPP-based linearization of the \TVPDP space—mapping neighborhoods to low-dimensional coordinates—so that standard tools (dimensionality reduction, trend analysis) are directly applicable to \TVPDsPP. \sebastien{Finally, a perspective is to improve computational efficiency by exploiting parallelism in our barycenter solvers, in particular for the greedy geodesic descent scheme, where the candidates sampled along geodesics from the current iterate to each input \TVPDP can be evaluated independently.}

Overall, \CEDP equips \TVPDP analysis with a principled distance, interpretable geodesics, and practical barycenters enabling standard geometric workflows (alignment, averaging, clustering) directly in the space of \TVPDsPP.

%% file: Supplementary_material.tex
	\appendices

	\section{Measure theory preliminaries}\label{app:appendixA}
        \label{app:preliminaries}
        \noindent We used concepts from measure theory in the formalization of the CED and the elements to which it applies; we therefore succinctly recall some necessary definitions and results. We refer the reader to the textbook \cite{bogachev2007measure} for a measure theory exposition.
        
        Let \(X\) be a set. A family \(\mathcal{M}\) of subsets of \(X\) is called a \(\sigma\)-algebra on \(X\) if it satisfies the following properties: (\textbf{i}) \(X \in \mathcal{M}\); (\textbf{ii}) If \(A \in \mathcal{M}\), then \((X \setminus A) \in \mathcal{M}\); (\textbf{iii}) If \(A_n \in \mathcal{M}\) for every \(n \in \mathbb{N}\), then $\cup_{n\in\mathbb N}\,A_n \in \mathcal{M}$. The elements of \(\mathcal{M}\) are referred to as the measurable sets of \(X\), and the ordered pair \((X,\mathcal{M})\) is called a measurable space.

        For \(F\) a family of subsets of a set \(X\), we define  \( \sigma(F)\) as the intersection of the \( \sigma\)-algebras $\mathcal{M}$ on X that contain \(F\), and call \(\sigma(F)\) the \(\sigma\)-algebra generated by \(F\); it is the smallest \(\sigma\)-algebra on \(X\) that contains \(F\).
        
        If \((X, d)\) is a metric space, the Borel \(\sigma\)-algebra on \(X\) is the \(\sigma\)-algebra, denoted $\mathcal{B}(X)$, generated by the open balls of \(X\): \( \mathcal{B}(X) = \sigma\bigl(\{ B(x, r) \mid x \in X,\ r > 0 \}\bigr)\), with $B(x, r)=\{\,y\in\mathcal{X},\, d(x,y)<r\}$.

        Considering that \((X,\mathcal{M})\) and \((Y,\mathcal{N})\) are measurable spaces, an application \(f : X \to Y\) is said to be measurable with respect to \(\mathcal{M}\) and \(\mathcal{N}\) if \(f^{-1}(A) \in \mathcal{M}\) for all \(A\in \mathcal{N}\). In this case, we also say that $f : (X,\mathcal{M}) \to (Y,\mathcal{N})$ is measurable.
        
        Moreover, for $(X,d)$ and $(Y,d')$ two metric spaces, a result of measure theory is that a continuous application \(f : X \to Y\) is then measurable with respect to $\mathcal{B}(X)$ and $\mathcal{B}(Y)$. When the choice of metrics is not clear from the context, we will say that $f : (X,d) \to (Y,d')$ is continuous, in order to make continuity with respect to $d$ and $d'$ explicit.
        
        Unless otherwise specified, a function \(f : X \to Y\) is said to be measurable if it is measurable with respect to $\mathcal{B}(X)$ and $\mathcal{B}(Y)$. 
        
        Fix \((X,\mathcal{M})\) a measurable space. A function \(\mu : \mathcal{M} \to [0,+\infty]\) is a positive measure if \(\mu(\emptyset) = 0\), and for any countable family \((A_n)_{n\in\mathbb{N}}\) of pairwise disjoint sets in \(\mathcal{M}\), we have \(\mu\left(\bigcup_{n\in\mathbb{N}} A_n \right) =\sum_{n\in\mathbb{N}}\mu(A_n)\). The triple \((X,\mathcal{M},\mu)\) is called a measure space.

        There exists a unique positive measure \(\lambda\) on \(\bigl(\mathbb{R}, \mathcal{B}(\mathbb{R})\bigr)\), called Lebesgue measure, such that \(\lambda\bigl([a,b]\bigr)=\lambda\bigl((a,b]\bigr)=\lambda\bigl([a,b)\bigr)=\lambda\bigl((a,b)\bigr) = b - a\) for all \(a, b \in \mathbb{R}\) with \(a < b\).
        
        Let \((X,\mathcal{M},\mu)\) be a measure space. A set \(A \subset X\) is defined as negligible if \(A \in \mathcal{M}\) and \(\mu(A) = 0\). The measure \(\mu\) is said to be complete if every subset of a negligible set is also measurable and negligible. Suppose \(\completion{(\mathcal{M}})\) be the collection of all sets \(E \subset X\) for which there exist \(A,B \in \mathcal{M}\) such that \(A \subset E \subset B\) and \(\mu(B \setminus A) = 0\). Define \(\mu^*(E) = \mu(A)\). Then \(\completion{(\mathcal{M}})\) is a \(\sigma\)-algebra on \(X\), \(\mu^*\) is a complete measure extending \(\mu\), and \(\bigl(X,\completion{(\mathcal{M}}),\mu^*\bigr)\) is called the completion of $\mathcal{M}$ for \(\mu\). In practice, we still denote $\mu^{*}$ by $\mu$.

        \sebastien{Let $Y$ be a metric space. If $I \in\mathcal{B}(\RNB)$, we said a map $f : I \to Y$ is Lebesgue--measurable if it is measurable with respect to $\completion{\bigl(\borelian(\RNB)|_{\intervalI}\bigr)}$, the completion of $\borelian(\RNB)|_{\intervalI}$ for the Lebesgue measure $\lambda$, and $\borelian(Y)$, that is,
        \(f^{-1}(B) \in \completion{\bigl(\borelian(\RNB)|_{\intervalI}\bigr)}
          \quad \text{for every } B \in \borelian(Y).
        \) In this case, we also say that $f : (I,\borelian(\RNB)|_{\intervalI}) \to (Y,\borelian(Y))$ is Lebesgue--measurable. As a classical result, if \(f : I \to Y\) is continuous, then $f$ is measurable. Moreover, another result is that if \(f : I \to Y\) is measurable, then $f$ is Lebesgue-measurable. So, if \(f : I \to Y\) is continuous, then $f$ is Lebesgue-measurable.}

	\section{Proofs}

    \textbf{Remark.} \textit{All the abstract results proved in this section apply in particular to the space \( (\PDS_2, \wasserstein_2) \)  of persistence diagrams with finite 2-th moment (implicitly assumed in the main manuscript), endowed with the 2-Wasserstein distance, which is a geodesic Polish metric space \cite{mileyko2011probability,turner2013frechetmeansdistributionspersistence} (in particular separable). In the main manuscript, we moreover take $A \subset E$ to be a singleton, so that for every $x \in E$ the set
    \(
      \{ y \in A : d(x,y) = d(x,A) \}
    \)
    is non-empty (indeed, it coincides with $A$ itself). Hence all the standing assumptions on $(E,d)$ and on $A$ are automatically satisfied in the setting used in the main manuscrit and for the experiments, and the results of the present section apply directly to that case.}

    \medskip

	\noindent Let \( (E, d) \) be a non-empty separable metric space,  
	provided with its Borel sigma-algebra \( \mathcal{B}(E) \), and let \( \alpha \in (0,1) \), \( B \in (0,1]\), and \( A \subsetneq E \) be non-empty.
	
	\bigskip
	Set \( \Delta \in (0, +\infty) \), and let \( \mathcal{S}^\Delta \)  
	be the space of applications \( F \) from a subset
	\[
	dom\,\TVPDf = \bigcup_{\variablei \in \{1, \dots, \variableN_\TVPDf^\paramDelta\}} \intervalI_\variablei^\TVPDf \subset \RNB
	\]
	to \( E \), with \( \variableN_\TVPDf^\paramDelta \in \NNB^* \), and \( ( \intervalI_\variablei^\TVPDf )_{\variablei \in \{1, \dots, \variableN_\TVPDf^\paramDelta\}} \) a disjoint family of intervals of \( \mathbb{R} \),  
	such that:
	
	\begin{itemize}
		\item \( \forall \variablei \in \{1, \dots, \variableN_\TVPDf^\paramDelta\}, \, \sup \intervalI_\variablei^\TVPDf - \inf \intervalI_\variablei^\TVPDf = \paramDelta \),
		\item \( \forall (\variablei, \variablei') \in \{1, \dots, \variableN_\TVPDf^\paramDelta\}^2,\ \forall (\variablex, \variablex') \in \intervalI_\variablei^\TVPDf \times \intervalI_{\variablei'}^\TVPDf,\ \variablei < \variablei' \Rightarrow \variablex < \variablex' \),
		\item \( \forall \variableTime \in dom\,\TVPDf,\, d(\TVPDf(\variableTime), \boundarySet)>0 \),
		\item \( \forall \variablei \in \{1, \dots, \variableN_\TVPDf^\paramDelta\},\, \TVPDf|_{\intervalI_\variablei^\TVPDf} : \bigl( \intervalI_\variablei^\TVPDf,\ \mathcal{Z}(\mathcal{B}(\mathbb{R})|_{I_i^F}) \bigr) \to \bigl( E,\ \borelian(E) \bigr) \) is measurable,
		
		where \( \mathcal{Z}(\mathcal{B}(\mathbb{R})|_{I_i^F}) \) is the completion of  
		\( \mathcal{B}(\mathbb{R})|_{I_i^F} \) for \( \lambda \), the Lebesgue measure,
		
		\item \( \forall \variablei \in \{1, \dots, \variableN_\TVPDf^\paramDelta\},\quad \mathrm{diam}(\mathrm{Im}(F|_{I_i^F})) < \infty \).
	\end{itemize}
	
	We will omit the superscript $\paramDelta$ in the notation $\variableN_\TVPDf^\paramDelta$ whenever no ambiguity arises. Let \( (\TVPDf ,\TVPDg)\in (\TVPDspace^\paramDelta)^2 \), \( \variablei \in \{1, \dots, \variableN_\TVPDf\}, \variablej \in \{1, \dots, \variableN_\TVPDg\}\).

	We denote \( F_i := F|_{I_i^F} \), and we call \( F_i \) a \( \Delta \)-subdivision.
	
	\bigskip
	
	We consider that two \( \Delta \)-subdivisions \( F_i, G_j \) are equal if:
	
	\begin{itemize}
		\item \( \inf I_i^F = \inf I_j^G \),
		\item \( \sup I_i^F = \sup I_j^G \),
		\item \( \lambda\left( \left\{ x \in I_i^F \cap I_j^G \mid F_i(x) \neq G_j(x) \right\} \right) = 0 \),
	\end{itemize}
	
	We denote \( s^\Delta \) the \( \Delta \)-subdivision space.
	
	\bigskip
	
	In a similar way, let \(  (\TVPDf, \TVPDg) \in (\TVPDspace^\paramDelta)^2 \), with  \(\variableN_\TVPDf =  \variableN_\TVPDg\), we consider that $\TVPDf$, $\TVPDg$ are equal if, \( \forall \variablei \in \{1, \dots, \variableN_\TVPDf\},\,  \TVPDf_\variablei =  \TVPDg_\variablei\). More generally, two measurable applications $\TVPDf$ and $\TVPDg$ from $\metricSpace\subset\RNB $ to $E$ are considered equal in our framework whenever \( \lambda ( \{ \variablex \in \metricSpace,\,  \TVPDf(\variablex) \neq \TVPDg(\variablex) \})  = 0  \).
	
	\bigskip
	
	\textbf{Remark :} Moreover if, \[F_i : ( I_i^F, \mathcal{Z}(\mathcal{B}(\mathbb{R})|_{I_i^F})) \longrightarrow (E, \mathcal{B}(E))
	\]
	with \( \inf I_i^F = c,\,  \sup I_i^F = d,\, (c,d) \in \mathbb{R}^2 \), $I_i^F\neq[c,d]$, we can use instead \( \tilde{F}_i \) (still denoted \( F_i \) in practice) :  
	\[
	([c,d], \mathcal{Z}(\mathcal{B}(\mathbb{R})|_{[c,d]})) \longrightarrow (E, \mathcal{B}(E))
	\]
	
	measurable and defined as, \( \forall t \in (c,d),\ \tilde{F}_i(t) = F_i(t) \),\,$\,
	\tilde{F}_i(c) =x, \tilde{F}_i(d) = y, \, \text{with } x,y \text{ some elements of } E.
	$
	
	\bigskip
	
	Indeed, let us suppose \( F \in \mathcal{S}^\Delta \), with \( I_i^F \neq [c,d] \),  
	for \( (c,d,i) \in \mathbb{R}^2 \times \{1, \dots, N_F\} \)
	
	\[
	\text{If } A \in \mathcal{B}(E),
	\]

    \[
\tilde{F}_i^{-1}(A) =
\begin{cases}
F_i^{-1}(A) \cup \{c,d\} & \text{if } x \in A \text{ and } y \in A,\\[0.3em]
F_i^{-1}(A) \cup \{c\}   & \text{if } x \in A \text{ and } y \notin A,\\[0.3em]
F_i^{-1}(A) \cup \{d\}   & \text{if } x \notin A \text{ and } y \in A,\\[0.3em]
F_i^{-1}(A)              & \text{if } x \notin A \text{ and } y \notin A.
\end{cases}
\]
	
	with \( F_i^{-1}(A) \in \mathcal{Z}(\mathcal{B}(I_i^F)) \text{, then with } F_i^{-1}(A) \in \mathcal{Z}(\mathcal{B}([c,d])) \).
	
	Moreover, \( \{c\},\ \{d\},\ \{c,d\} \in \mathcal{Z}(\mathcal{B}([c,d])) \), then
	
	\[
	F_i^{-1}(A) \cup \{c\} \in \mathcal{Z}(\mathcal{B}([c,d])),\quad
	F_i^{-1}(A) \cup \{d\} \in \mathcal{Z}(\mathcal{B}([c,d]))
	\]
	
	\[
	F_i^{-1}(A) \cup \{c,d\} \in \mathcal{Z}(\mathcal{B}([c,d]))
	\] .
	
	And so,  
	\[
	\tilde{F}_i : ([c,d], \mathcal{Z}(\mathcal{B}(\mathbb{R})|_{[c,d]})) \longrightarrow (E, \mathcal{B}(E))
	\]
	
	is measurable and \( \tilde{F}_i = F_i \).
	
	\bigskip
	
	Similarly, if \( F_i : (I_i^F, \mathcal{Z}(\mathcal{B}(\mathbb{R})|_{I_i^F})) \longrightarrow (E, \mathcal{B}(E)) \),
	\bigskip
	
	with \( \inf I_i^F = c,\,  \sup I_i^F = d,\, (c,d) \in \mathbb{R}^2 \), $I_i^F\neq(c,d)$, we can use instead \( \tilde{F}_i \) (still denoted \( F_i \) in practice) :  
	\[
	((c,d), \mathcal{Z}(\mathcal{B}(\mathbb{R})|_{(c,d)})) \longrightarrow (E, \mathcal{B}(E))
	\]
	
	measurable and defined as \( \forall t \in (c,d),\ \tilde{F}_i(t) = F_i(t) \).
	
	\bigskip
	
	Indeed, let us suppose \( F \in \mathcal{S}^\Delta \), with \( I_i^F \neq (c,d) \),  
	for \( (c,d,i) \in \mathbb{R}^2 \times \{1, \dots, N_F\} \)
	
	\[
	\ \text{If } A \in \mathcal{B}(E),
	\]
	
	\[\text{we have }
	\tilde{F}_i^{-1}(A) = F_i^{-1}(A) \cap (c,d)
	\]
	
	with \( F_i^{-1}(A) \in \mathcal{Z}(\mathcal{B}(I_i^F)) \).
	
	\bigskip

	Then \( F_i^{-1}(A) \cap (c,d) \in \mathcal{Z}(\mathcal{B}((c,d)))
	\)
	\text{and }
	\( \tilde{F}_i^{-1}(A) \in \mathcal{Z}(\mathcal{B}((c,d))).
	\)
	
	\bigskip
	
	Finally,  
	\[
	\tilde{F}_i : ((c,d), \mathcal{Z}(\mathcal{B}(\mathbb{R})|_{(c,d)})) \longrightarrow (E, \mathcal{B}(E))
	\]
	
	is measurable and \( \tilde{F}_i = F_i \).
	
	\bigskip
	
	So, in the next we can assume that $I_i^F$ is open, or closed, as we need.
	
	\bigskip
	
	We therefore now assume for the next that if \( F \in \mathcal{S}^\Delta \), then \( \forall i \in \{1, \dots, N_F\}, \) $I_i^F$ is open.
	
	\bigskip
	
	\begin{proposition}\label{prop:W_measurability} 
	Let \( (P, Q) \in (\mathcal{S}^\Delta)^2 \), and \( (i, j) \in \{1, \dots, N_P\} \times \{1, \dots, N_Q\} \).
	
	\bigskip
	
	We denote:
	
	\begin{align*}
	&a_i := \inf(I_i^P),\quad b_i := \sup(I_i^P),\quad
	\\&c_j := \inf(I_j^Q),\quad d_j := \sup(I_j^Q).
	\end{align*}

	\bigskip
	
	Then the application,  
	\[
	W_\alpha^{P_i,Q_j} : \left( I_i^P,\ \mathcal{Z}(\mathcal{B}(\mathbb{R})|_{I_i^P}) \right) \longrightarrow \left( \mathbb{R},\ \mathcal{B}(\mathbb{R}) \right)
	\]
	\[
	t \longmapsto W_\alpha^{P_i,Q_j}(t) = (1 - \alpha) \cdot d\big( P(t),\ Q(t + c_j - a_i) \big) + \alpha \cdot |c_j - a_i|
	\]
	
	is measurable, and Lebesgue-integrable.
	\end{proposition}
	\bigskip
	
\begin{proof}
	
	\bigskip
	
	The application  
	\[
	\mathcal{U} : \left( I_i^P,\ \mathcal{Z}(\mathcal{B}(\mathbb{R})|_{I_i^P}) \right) \longrightarrow \left( I_j^Q,\ \mathcal{Z}(\mathcal{B}(\mathbb{R})|_{I_j^Q}) \right)
	\]
	
	is measurable.
	
	\bigskip
	
	Indeed, let \( M \in \mathcal{Z}(\mathcal{B}(\mathbb{R})|_{I_j^Q}) \),
	
	\bigskip
	
	then \( M = A \cup N \) with \( A \in \mathcal{B}(\mathbb{R})|_{I_j^Q} \) and  
	\( N \) a negligible part of \( \mathcal{B}(\mathbb{R})|_{I_j^Q} \),  
	
	Then
	\[
	\mathcal{U}^{-1}(A \cup N) = \mathcal{U}^{-1}(A) \cup \mathcal{U}^{-1}(N).
	\]

	Since,
	
	\[
	\mathcal{U} : \left( I_i^P,\ \mathcal{B}(\mathbb{R})|_{I_i^P} \right) \longrightarrow \left( I_j^Q,\ \mathcal{B}(\mathbb{R})|_{I_j^Q} \right)
	\]
	\[
	t \longmapsto t + c_j - a_i
	\]
	
	is measurable, as continuous function,  
	\( \forall K \in \mathcal{B}(\mathbb{R})|_{I_j^Q} \),  
	\[
	\mathcal{U}^{-1}(K) \in \mathcal{B}(\mathbb{R})|_{I_i^P}.
	\]
	
	So,
	\[
	\mathcal{U}^{-1}(A) \in \mathcal{B}(\mathbb{R})|_{I_i^P}.
	\]
	
	\bigskip
	
	Moreover, \( N \) is a negligible part of \( \mathcal{B}(\mathbb{R})|_{I_j^Q} \), then  
	because \( \lambda \) is invariant by translation,  
	\[
	\exists T \in \mathcal{B}(\mathbb{R})|_{I_j^Q},\ N \subset T,\ \lambda(T) = 0 \Rightarrow \mathcal{U}^{-1}(N) \subset \mathcal{U}^{-1}(T)
	\]
	\[
	\lambda(\mathcal{U}^{-1}(T)) = 0,\ \text{with } \mathcal{U}^{-1}(T) \in \mathcal{B}(\mathbb{R})|_{I_i^P} \quad (\text{Because,}
	\]
	
	\[
	\mathcal{U} : \left( I_i^P,\ \mathcal{B}(\mathbb{R})|_{I_i^P} \right) \longrightarrow \left( I_j^Q,\ \mathcal{B}(\mathbb{R})|_{I_j^Q} \right)
	\]
	\[
	t \longmapsto t + c_j - a_i
	\]
	is continuous).
	
	\bigskip
	
	So \( \mathcal{U}^{-1}(M) \) is the union of a element of  
	\( \mathcal{B}(\mathbb{R})|_{I_i^P} \) and a negligible part of  
	\( \mathcal{B}(\mathbb{R})|_{I_i^P} \), then
	\[
	\mathcal{U}^{-1}(M) \in \mathcal{Z}(\mathcal{B}(\mathbb{R})|_{I_i^P}).
	\]
	
	And \( \mathcal{U} : \left( I_i^P,\ \mathcal{Z}(\mathcal{B}(\mathbb{R})|_{I_i^P}) \right) \to \left( I_j^Q,\ \mathcal{Z}(\mathcal{B}(\mathbb{R})|_{I_j^Q}) \right) \)
	
	\[
	t \longmapsto t + c_j - a_i
	\]
	
	is measurable.
	
	\medskip
	
	Also \( P|_{I_i^P} : \left( I_i^P,\ \mathcal{Z}(\mathcal{B}(\mathbb{R})|_{I_i^P}) \right) \longrightarrow (E,\ \mathcal{B}(E)) \)
	\bigskip
	and
	\bigskip
	\( Q|_{I_j^Q} : \left( I_j^Q,\ \mathcal{Z}(\mathcal{B}(\mathbb{R})|_{I_j^Q}) \right) \longrightarrow (E,\ \mathcal{B}(E)) \)
	are measurable by definition.
	
	\medskip
	
	Then the application
	
	\[
	\widetilde{Q} = \left( Q|_{I_j^Q} \circ \mathcal{U} \right) : \left( I_i^P,\ \mathcal{Z}(\mathcal{B}(\mathbb{R})|_{I_i^P}) \right) \longrightarrow (E,\ \mathcal{B}(E))
	\]
	
	is measurable as composition of measurable applications.
	
	\medskip
	
	The application \( V : \left( I_i^P,\ \mathcal{Z}(\mathcal{B}(\mathbb{R})|_{I_i^P}) \right) \to \left( E^2,\ \mathcal{B}(E) \otimes \mathcal{B}(E) \right) \)
	
	\[
	t \longmapsto V(t) = \left( P(t),\ \widetilde{Q}(t) \right)
	\]
	
	has its component-wise applications measurable, and  
	the target sigma-algebra is the sigma-algebra product, then \( V \) is measurable.
	
	\bigskip
	
	Because \( d \) is continuous (as a distance) on \( T \) the topology product of  
	\( (E, d) \) and \( (E, d) \), then  
	\[
	d : (E^2, \sigma(T)) \rightarrow (\mathbb{R}, \mathcal{B}(\mathbb{R})) \text{ is measurable,}
	\]
	
	where we denote \( \sigma(T) \) the sigma-algebra generated by \( T \).  
	
	\bigskip
	
	Because \( \sigma(T) = \mathcal{B}(E) \otimes \mathcal{B}(E) \) (since \( E \) is separable),  
	and \( d \) is measurable for \( \sigma(T) \), then \( d \) is measurable  
	for \( \mathcal{B}(E) \otimes \mathcal{B}(E) \).
	
	\bigskip
	
	Since  
	\[
	V : \left( I_i^P,\ \mathcal{Z}(\mathcal{B}(\mathbb{R})|_{I_i^P}) \right) \rightarrow (E^2, \mathcal{B}(E) \otimes \mathcal{B}(E))
	\]  
	and  
	\[
	d : (E^2, \mathcal{B}(E) \otimes \mathcal{B}(E)) \rightarrow (\mathbb{R}, \mathcal{B}(\mathbb{R})) \text{ are measurable,}
	\]  
	then  
	\[
	\widetilde{W} := d \circ V : \left( I_i^P,\ \mathcal{Z}(\mathcal{B}(\mathbb{R})|_{I_i^P}) \right) \rightarrow (\mathbb{R}, \mathcal{B}(\mathbb{R}))
	\]  
	is measurable.
	
	\bigskip
	
	Then $W = (1-\alpha) \, . \, \widetilde{W}$ is measurable, as a product of a measurable function by a constant.
	
	\bigskip
	
	\[
	G : \left( \mathbb{R} \times \mathbb{R},\ \mathcal{B}(\mathbb{R}) \otimes \mathcal{B}(\mathbb{R}) \right) \rightarrow (\mathbb{R}, \mathcal{B}(\mathbb{R}))
	\]
	\[
	(x, y) \longmapsto x + y
	\]
	
	is measurable, as continuous function for  
	\[
	\mathcal{B}(\mathbb{R} \otimes \mathbb{R}) = \mathcal{B}(\mathbb{R}) \otimes \mathcal{B}(\mathbb{R}) \quad \text{since } \mathbb{R} \text{ is separable.}
	\]

	\(\text{Because, } 
	I_i^P \in \mathcal{Z}(\mathcal{B}(\mathbb{R})|_{I_i^P}) \text{ , then }
	\)
	
	\[
	H : \left( I_i^P,\ \mathcal{Z}(\mathcal{B}(\mathbb{R})|_{I_i^P}) \right)
	\longrightarrow (\mathbb{R},\ \mathcal{B}(\mathbb{R}))
	\]
	
	\[
	t \longmapsto \alpha \cdot |c_j - a_i|
	\]
	
	is measurable.
	
	\bigskip
	
	Then,
	
	\[
	W_\alpha^{P_i,Q_j} : \left( I_i^P,\ \mathcal{Z}(\mathcal{B}(\mathbb{R})|_{I_i^P}) \right)
	\longrightarrow (\mathbb{R},\ \mathcal{B}(\mathbb{R}))
	\]
	
	\[
	t \longmapsto G\left( W(t),\ H(t) \right)
	\]
	
	is measurable as composition of measurable functions.

	\[
	\forall t \in I_i^P,
	\]
    \begin{gather}0 \leq W_\alpha^{P_i,Q_j}(t) \leq (1 -\alpha) (  d( P( \frac{a_i + b_i}{2} ),Q( \frac{c_j + d_j}{2} ) )+ \notag\\\mathrm{diam}(\mathrm{Im}(P|_{I_i^P})) + \mathrm{diam}(\mathrm{Im}(Q|_{I_j^Q})))+ \alpha \cdot |c_j - a_i|\end{gather}
	
	(by triangular inequality of \( d \)).
	\bigskip
	Then \( W_\alpha^{P_i,Q_j} \) is Lebesgue-integrable.

        \end{proof}

        \begin{proposition}\label{prop:S_measurability}  Let \( P \in \mathcal{S}^\Delta \), and \( i \in \{1, \dots, N_P\} \). We denote \( a_i = \inf(I_i^P),\ b_i = \sup(I_i^P) \).
	
	Then the application,  
	\[
	S_\alpha^{P_i} : \left( I_i^P,\ \mathcal{Z}(\mathcal{B}(\mathbb{R})|_{I_i^P}) \right) \longrightarrow (\mathbb{R}, \mathcal{B}(\mathbb{R}))
	\]
	\[
	t \longmapsto S_\alpha^{P_i}(t) = (1 - \alpha) \cdot d(P(t), A)
	\]
	
	is measurable, and Lebesgue-integrable.
	\end{proposition}
    
\begin{proof}
	
	The application  
	\[
	d(\cdot, A) : (E, d) \longrightarrow (\mathbb{R}, |\cdot|)
	\]  
	\[
	x \longmapsto d(x, A)
	\]
	
	is continuous, because 1-Lipschitz (as a distance to a subset function in a metric space), so
	
	\[
	d(\cdot, A) : (E, \mathcal{B}(E)) \longrightarrow (\mathbb{R}, \mathcal{B}(\mathbb{R}))
	\]
	\[
	x \longmapsto d(x, A)
	\]
	
	is measurable.
	
	\bigskip
	
	Since  
	\[
	P|_{I_i^P} : \left( I_i^P,\ \mathcal{Z}(\mathcal{B}(\mathbb{R})|_{I_i^P}) \right) \longrightarrow (E, \mathcal{B}(E))
	\]
	
	is measurable, then  
	\[
	\widetilde{S} := d(\cdot, A) \circ P|_{I_i^P}
	\]
	
	is measurable, as composition of measurable functions.
	
	\bigskip
	
	Then, $S_\alpha^{P_i}=(1-\alpha)\cdot \widetilde{S}$ is measurable as a product of a measurable function by a constant.
	
	\bigskip
	
	Finally, \( \forall t \in I_i^P \),
	\begin{align}
	0 \leq S_\alpha^{P_i}(t) &\leq (1-\alpha).\Bigg(d\left( P\left( \frac{a_i + b_i}{2} \right), A \right)
	\notag\\&+ \mathrm{diam}\left( \mathrm{Im}(P|_{I_i^P}) \right)\Bigg)\end{align}
	
	so $S_\alpha^{P_i}$ is Lebesgue-integrable.\qedhere
    
        \end{proof}

    \begin{definition}Let \( (\rho, q) \in (s^\Delta)^2 \). \bigskip
	
	Then \( \exists (P, Q) \in (\mathcal{S}^\Delta)^2,\ \exists (i,j) \in \{1, \dots, N_P\} \times \{1, \dots, N_Q\} \),  
	such that \( \rho = P_i \) and \( q = Q_j \).
	
	\smallskip
	
	Let’s denote  
	\[a_i = \inf(I_i^P),\quad b_i = \sup(I_i^P),\]
	\[c_j = \inf(I_j^Q),\quad d_j = \sup(I_j^Q).
	\]
	
	\smallskip
	
	We define
	\[
	D^\alpha_\Delta(\rho, q) = D^\alpha_\Delta(P_i, Q_j) := \int_{a_i}^{b_i} W_\alpha^{P_i,Q_j}(t)\,dt
	\]
	
	and
	
	\[
	D^\alpha_\Delta(\rho, A) = D^\alpha_\Delta(P_i, A) := \int_{a_i}^{b_i} S_\alpha^{P_i}(t)\,dt
	\]
        
    \end{definition}
	
	\begin{lemma}\label{lem:metric_space_local}
	Let \( (\rho, q, r) \in (s^\Delta)^3 \). \bigskip
	
	So \( \exists (P, Q, R) \in (\mathcal{S}^\Delta)^3,\ \exists (i, j, k) \in \{1, \dots, N_P\} \times \{1, \dots, N_Q\} \times \{1, \dots, N_R\} \),  
	such that \( \rho = P_i,\ q = Q_j,\ r = R_k \)
	
	\bigskip
	
	Then we have,
	
	\begin{itemize}
		\item \( \forall i \in \{1, \dots, N_P\},\ \forall j \in \{1, \dots, N_Q\},\ D^\alpha_\Delta(P_i, Q_j) = D^\alpha_\Delta(Q_j, P_i) \)
		\item \( \forall (i,j,k) \in \{1, \dots, N_P\} \times \{1, \dots, N_Q\} \times \{1, \dots, N_R\}, \)
		\[
		D^\alpha_\Delta(P_i, Q_j) + D^\alpha_\Delta(Q_j, R_k) \geq D^\alpha_\Delta(P_i, R_k)
		\]
		\item \( \forall (i,j) \in \{1, \dots, N_P\} \times \{1, \dots, N_Q\},\ P_i = Q_j \Leftrightarrow D^\alpha_\Delta(P_i, Q_j) = 0 \)
		\item \( \forall (i,j) \in \{1, \dots, N_P\} \times \{1, \dots, N_Q\},\ D^\alpha_\Delta(P_i, Q_j) \geq 0 \)
	\end{itemize}
	
        And so, \( (s^\Delta,D^\alpha_\Delta) \) is a metric space.\end{lemma}  

        \begin{proof}Let’s denote \( a_i =\inf(I_i^P),\ b_i = \sup(I_i^P),\ c_j = \inf(I_j^Q),\ d_j = \sup(I_j^Q) \), \( e_k = \inf(I_k^R) \), and \( f_k = \sup(I_k^R) \).\begin{gather}(I_i^P,\ \mathcal{Z}(\mathcal{B}(\mathbb{R})|_{I_i^P})) \longrightarrow (\mathbb{R},\ \mathcal{B}(\mathbb{R}))\notag\\t \longmapsto (1 - \alpha) \cdot d\left( P(t),\ Q(t + c_j - a_i) \right) + \alpha \cdot |c_j - a_i|\notag\end{gather}is measurable and Lebesgue-integrable (see~\ref{prop:W_measurability}), and\begin{gather}(0, 1) \longrightarrow (a_i,b_i) \text{ is a } \mathcal{C}^1 \text{-diffeomorphism.}\notag\\t \longmapsto a_i \cdot (1 - t) + b_i \cdot t\notag\end{gather}Moreover, \begin{gather}(I_j^Q,\ \mathcal{Z}(\mathcal{B}(\mathbb{R})|_{I_j^Q})) \longrightarrow (\mathbb{R},\ \mathcal{B}(\mathbb{R}))\notag\\t \longmapsto (1 - \alpha) \cdot d\left( Q(t),\ P(t + a_i - c_j) \right) + \alpha \cdot |a_i - c_j|\notag\end{gather}is measurable and Lebesgue-integrable, and\begin{gather}(0, 1) \longrightarrow (c_j,d_j) \text{ is a } \mathcal{C}^1 \text{-diffeomorphism.}\notag\\t \longmapsto c_j \cdot (1 - t) + d_j \cdot t\notag\end{gather}

    Then, using Lebesgue integration by substitution theorem,
        \begin{flalign}&\bullet D^\alpha_\Delta(P_i, Q_j)\notag\\&= \int_{a_i}^{b_i} [ (1 - \alpha) \cdot d( P(t),\ Q(t + c_j - a_i) )\nonumber\\& + \alpha \cdot |c_j - a_i| ] dt\nonumber\\&= (b_i - a_i) \cdot \int_0^1 [ (1 - \alpha) \cdot d( P(a_i \cdot (1 - t) + b_i \cdot t),\ \nonumber\\& Q(a_i \cdot (1 - t) + b_i \cdot t + c_j - a_i) ) + \alpha \cdot |c_j - a_i| ] dt\nonumber\\&= \Delta \cdot \int_0^1 [ (1 - \alpha) \cdot d( P(a_i \cdot (1 - t) + b_i \cdot t),\ \nonumber\\&Q((b_i - a_i) \cdot t + c_j) ) + \alpha \cdot |c_j - a_i| ] dt\nonumber\\&= \Delta \cdot \int_0^1 [ (1 - \alpha) \cdot d( P(a_i \cdot (1 - t) + b_i \cdot t),\ \nonumber\\&Q(c_j \cdot (1 - t) + d_j \cdot t) ) + \alpha \cdot |c_j - a_i| ] dt\nonumber\\&= \Delta \cdot \int_0^1 [ (1 - \alpha) \cdot d( Q(c_j \cdot (1 - t) + d_j \cdot t),\ \nonumber\\& P(a_i \cdot (1 - t) + b_i \cdot t) ) + \alpha \cdot |a_i - c_j| ] dt\nonumber\\&= |d_j - c_j| \cdot \int_0^1 [ (1 - \alpha) \cdot d( Q(c_j \cdot (1 - t) + d_j \cdot t),\ \notag\\&P((b_i - a_i) \cdot t + a_i) ) + \alpha \cdot |a_i - c_j| ] dt\nonumber\\&= |d_j - c_j| \cdot \int_0^1 [ (1 - \alpha) \cdot d( Q(c_j \cdot (1 - t) + d_j \cdot t),\notag\\& P((d_j - c_j) \cdot t + a_i) ) + \alpha \cdot |a_i - c_j| ] dt\nonumber\\&= \int_{c_j}^{d_j} [ (1 - \alpha) \cdot d( Q(t),\ P(t + a_i - c_j) ) + \notag\\&\alpha \cdot |a_i - c_j| ] dt\notag\\&= D^\alpha_\Delta(Q_j, P_i)\end{flalign}
        \begin{flalign}
        &\bullet D^\alpha_\Delta(P_i, Q_j) + D^\alpha_\Delta(Q_j, R_k)\notag\\&= \Delta \cdot \int_0^1 [ (1 - \alpha) \cdot d( P(a_i \cdot (1 - t) + b_i \cdot t),\ \notag \\&Q(c_j \cdot (1 - t) + d_j \cdot t) ) + \alpha \cdot |c_j - a_i| ] dt\notag\\&+ \Delta \cdot \int_0^1 [ (1 - \alpha) \cdot d( Q(c_j \cdot (1 - t) + d_j \cdot t),\ \notag\\&R(e_k \cdot (1 - t) + f_k \cdot t) ) + \alpha \cdot |e_k - c_j| ] dt\notag\\&= \Delta \cdot \int_0^1 [ (1 - \alpha) \cdot ( d( P(a_i \cdot (1 - t) + b_i \cdot t),\ \notag\\&Q(c_j \cdot (1 - t) + d_j \cdot t) ) \notag\\& +\ d( Q(c_j \cdot (1 - t) + d_j \cdot t),\ \notag\\&R(e_k \cdot (1 - t) + f_k \cdot t) ) ) + \alpha \cdot ( |c_j - a_i| + |e_k - c_j| ) ] dt\notag\\&\geq \Delta \cdot \int_0^1 [ (1 - \alpha) \cdot d( P(a_i \cdot (1 - t) + b_i \cdot t),\ \notag\\&R(e_k \cdot (1 - t) + f_k \cdot t) ) + \alpha \cdot |e_k - a_i| ] dt\notag\\&= D^\alpha_\Delta(P_i, R_k) \quad \big(\text{by triangular inequality of \( d \) and $\mid\cdot\mid$}\big)\end{flalign}
		
        $\bullet$ If \(  P_i = Q_j \), then
        \begin{flalign}&D^\alpha_\Delta(P_i,Q_j)\notag\\&=\Delta \cdot \int_0^1 [ (1 - \alpha) \cdot d( P(a_i \cdot (1 - t) + b_i \cdot t),\ \notag\\&Q(c_j \cdot (1 - t) + d_j \cdot t) ) + \alpha \cdot |c_j - a_i| ] dt\notag\\&= \Delta \cdot \int_0^1 [ (1 - \alpha) \cdot d( P(a_i \cdot (1 - t) + b_i \cdot t),\ \notag\\&Q(a_i \cdot (1 - t) + b_i \cdot t) ) + \alpha \cdot |a_i - a_i| ] dt\notag\\&= \Delta \cdot \int_0^1 [ (1 - \alpha) \cdot d( P(a_i \cdot (1 - t) + b_i \cdot t),\ \notag\\&Q(a_i \cdot (1 - t) + b_i \cdot t) ) ] dt\notag\\&= \Delta \cdot (1 - \alpha) \cdot \int_0^1 d( P(a_i \cdot (1 - t) + b_i \cdot t),\ \notag\\&Q(a_i \cdot (1 - t) + b_i \cdot t) ) dt\notag\\&= \Delta \cdot (1 - \alpha) \cdot \int_0^1 0 \cdot dt = 0\quad \text{(Since }\notag\\&\lambda( \{ x \in I_i^P \cap I_j^Q \text{ such that } P_i(x) \neq Q_j(x) \} ) = 0,\notag\\&\text{then } \lambda( \{ t \in ]0,1[ \text{ such that }d( P(a_i \cdot (1 - t) + b_i \cdot t),\ \notag\\&Q(a_i \cdot (1 - t) + b_i \cdot t) ) \neq 0 \} ) = 0)\end{flalign}

        $\bullet$ If \( D^\alpha_\Delta(P_i, Q_j) = 0 \), then
        \begin{gather}\Delta \cdot \int_0^1 [ (1 - \alpha) \cdot d( P(a_i \cdot (1 - t) + b_i \cdot t),\  \notag\\Q(c_j \cdot (1 - t) + d_j \cdot t) ) + \alpha \cdot |c_j - a_i| ] dt = 0,\end{gather}
		
        \noindent which implies,  \( a_i = c_j \) (and so \( b_i = d_j \)), and\begin{gather}\lambda( \{ t \in (0,1) \ \text{such that} \ d( P(a_i \cdot (1 - t) + b_i \cdot t),\ \notag\\Q(a_i \cdot (1 - t) + b_i \cdot t) ) \neq 0 \} ) = 0,\end{gather}
		
	\noindent Thus \[ \lambda\left( \left\{ x \in I_i^P \cap I_j^Q\ \text{such that}\ P_i(x) \neq Q_j(x) \right\} \right) = 0,\]
    
    \noindent and \( a_i = c_j,\ b_i = d_j \), then finally,\[ P_i = Q_j.
		\] \qedhere
    \end{proof}
	\begin{lemma}\label{lem:triangular_inequality_local_distance_lemma}
	
	\bigskip
	
	Let \( (\rho, q) \in (s^\Delta)^2 \), \bigskip
	
	so \( \exists (P, Q) \in (\mathcal{S}^\Delta)^2,\ \exists (i, j) \in \{1, \dots, N_P\} \times \{1, \dots, N_Q\} \),  
	such that \( \rho = P_i,\ q = Q_j.\)
	
	\bigskip
	
	We have, $$D^\alpha_\Delta(P_i,A)+D^\alpha_\Delta(P_i,Q_j) \geq D^\alpha_\Delta(Q_j, A)$$
	\bigskip
	\end{lemma}
    
	\begin{proof} 
	With the same notation as above,
	\bigskip

	\(
	D^\alpha_\Delta(P_i, A) + D^\alpha_\Delta(P_i, Q_j)
	\)
	
	\(
	= \int_{a_i}^{b_i} (1 - \alpha) \cdot d(P(t), A) \, dt 
	+ \int_{a_i}^{b_i} \big((1 - \alpha) \cdot d\big(P(t), Q(t + c_j - a_i)\big) + \alpha | c_j - a_i| \, \big)dt
	\)
	
	\(
	\geq \int_{a_i}^{b_i} (1 - \alpha) \cdot d(P(t), A) \, dt 
	+ \int_{a_i}^{b_i} \big((1 - \alpha) \cdot d(P(t), Q(t + c_j - a_i))\big) \, dt
	\)
	
	\(
	= \int_{a_i}^{b_i} (1 - \alpha) \cdot \big[ d( P(t),A) + d(P(t), Q(t + c_j - a_i)) \big] \, dt
	\)
	
	\(
	\geq \int_{a_i}^{b_i} (1 - \alpha) \cdot d(Q(t + c_j - a_i),A) \, dt
	\)
	
	 (Because in a metric space $(E,d)$ with $A \subset E$, we have $\forall(x,y) \in E^2, \ d(x,y) + d(y,A) \geq d(x,A)$, indeed $d(.,A)$ is 1-Lipschitz)
	
	\(
	= \int_{c_j}^{d_j} (1 - \alpha) \cdot d(Q(t),A) \, dt = D^\alpha_\Delta(Q_j, A),
	\)
	
	\noindent by Lebesgue integral substitution  
	theorem, with the \( C^1 \)-diffeomorphism,
	
	\[
	\varphi : (c_j,d_j) \longrightarrow (a_i,b_i),\quad t \longmapsto t + a_i - c_j
	\]
	
	and the measurable function,
	
	\[
	((a_i,b_i),\ \mathcal{Z}(\mathcal{B}(\mathbb{R})|_{(a_i,b_i)})) \longrightarrow (\mathbb{R}, \mathcal{B}(\mathbb{R})), 
	\]
	
	\[
	t \longmapsto (1 - \alpha) \cdot d(Q(t + c_j - a_i),A)
	\]
	
	as composition of \( \mathcal{S}_\alpha^{Q_j} (see~\ref{prop:S_measurability})\) and \(\mathcal{U}\). \qedhere
	\end{proof}
        
        If $\TVPDp \in \TVPDspace^\paramDelta$, thanks to the imposed conditions, we can unambiguously denote $\TVPDp$ as a sequence $(\TVPDp_\variablei)_{1 \leq \variablei < \variableN_\TVPDp}.$ Moreover, if  $\variableN \in \{1, \dots, \variableN_\TVPDp\}$,  then $(\TVPDp_\variablei)_{1 \leq \variablei < \variableN}$  stay obviously in $\TVPDspace^\paramDelta$.

        \begin{definition}
        Set $ (\TVPDp, \TVPDq) \in (\TVPDspace^\paramDelta)^2$.  We call
$\paramDelta$-partial assignment (we omit the $\paramDelta$ when the context is
clear) between  $\TVPDp$  and  $\TVPDq$  any function $\partialAssignment$  from
 dom $\partialAssignment \subset \{1, 2, \dots, \variableN_\TVPDp\}$ to Im
$\partialAssignment$ $\subset \{1, 2, \dots, \variableN_\TVPDq\}$, such that
\partialAssignment~  is strictly increasing, i.e.  $\forall (\variablei,
\variablej)\in$ $(\text{dom }\partialAssignment)^2$, $ \variablei < \variablej
\Rightarrow \partialAssignment(\variablei) < \partialAssignment(\variablej).$ We
denote $\PASet^\paramDelta(\TVPDp, \TVPDq)$ the set of the partial assignments
between $\TVPDp$ and $\TVPDq$ in $\TVPDspace^\paramDelta$, we can see that
specifying an $\partialAssignment \in \PASet^\paramDelta(\TVPDp, \TVPDq)$
directly yields an assignment $\partialAssignment^{-1} \in
\PASet^\paramDelta(\TVPDq, \TVPDp)$.

Then, the $\CEDM^{\paramDelta}_{\paramWeight,\paramPenalty}$ between $\TVPDp=(\TVPDp_\variablei)_{1 \leq \variablei \leq \variableN_\TVPDp}$ and $\TVPDq=(\TVPDq_\variablej)_{1 \leq \variablej \leq \variableN_\TVPDq}$ is defined as: 
\begin{align}
	\CEDM^{\paramDelta}_{\paramWeight,\paramPenalty} (\TVPDp,\TVPDq)=&\mathop{\min}\limits_{\substack{\partialAssignment\in \PASet^\paramDelta(\TVPDp,\TVPDq)}} \text{cost}^\paramDelta_{(\TVPDp,\TVPDq)}(\partialAssignment):= \notag  \\ \mathop{\min}\limits_{\substack{\partialAssignment\in \PASet^\paramDelta(\TVPDp,\TVPDq)}}&\Big(\sum\limits_{\variablei \in \text{ dom }\partialAssignment} \,\,\localDistanceAppendix^\paramWeight_{\paramDelta}(\TVPDp_\variablei,\TVPDq_{\partialAssignment(\variablei)})\; \tag{2}\\ +&\;\;\sum\limits_{\variablei\, \notin \text{ dom }\partialAssignment}\,\, \paramPenalty\cdot \localDistanceAppendix^\paramWeight_{\paramDelta}(\TVPDp_\variablei,\boundarySet)\; \tag{3}\\ +&\;\:\:\,\sum\limits_{\variablej\, \notin \text{ Im }\partialAssignment}\paramPenalty\cdot \localDistanceAppendix^\paramWeight_{\paramDelta}(\TVPDq_\variablej,\boundarySet) \Big)\tag{4}
\end{align}
        \end{definition}
        \begin{proposition}
\label{prop:monotonicity_Delta}
Let $0 < \paramDeltaone, \paramDeltatwo < +\infty$ be such that
$\paramDeltatwo = k\,\paramDeltaone$ for some integer $k \geq 1$.
Let $(\TVPDp,\TVPDq) \in \TVPDspace^{\paramDeltatwo} \times \TVPDspace^{\paramDeltatwo}$.
Then $(\TVPDp,\TVPDq) \in \TVPDspace^{\paramDeltaone} \times \TVPDspace^{\paramDeltaone}$ and
\begin{equation}
  \CEDM^{\paramDeltaone}_{\paramWeight,\paramPenalty}(\TVPDp,\TVPDq)
  \;\le\;
  \CEDM^{\paramDeltatwo}_{\paramWeight,\paramPenalty}(\TVPDp,\TVPDq).
\end{equation}
\end{proposition}

\begin{proof}
By definition of $\TVPDspace^{\paramDeltatwo}$, the domain of $\TVPDp$
(resp.\ $\TVPDq$) is a finite union of consecutive intervals of length
$\paramDeltatwo$, which we call
$\paramDeltatwo$--subdivisions.
Since $\paramDeltatwo = k\,\paramDeltaone$, each such
$\paramDeltatwo$--subdivision can be partitioned into $k$ consecutive
intervals of length $\paramDeltaone$. Thus $(\TVPDp,\TVPDq)$ also admit a
representation in $\TVPDspace^{\paramDeltaone} \times \TVPDspace^{\paramDeltaone}$.

Let
\[
  \partialAssignmentBis \in \PASet^{\paramDeltatwo}(\TVPDp,\TVPDq)
\]
be a partial assignment achieving
\(
  \CEDM^{\paramDeltatwo}_{\paramWeight,\paramPenalty}(\TVPDp,\TVPDq)
\),
that is
\begin{equation}
  \CEDM^{\paramDeltatwo}_{\paramWeight,\paramPenalty}(\TVPDp,\TVPDq)
  \;=\;
  \text{cost}^{\paramDeltatwo}_{(\TVPDp,\TVPDq)}(\partialAssignmentBis).
\end{equation}

We now construct a refined partial assignment
$\partialAssignmentBis' \in \PASet^{\paramDeltaone}(\TVPDp,\TVPDq)$ as follows.
Each $\paramDeltatwo$--subdivision of $\TVPDp$ (resp.\ $\TVPDq$) is split into
$k$ consecutive $\paramDeltaone$--subdivisions.
For every matched pair $(i,\partialAssignmentBis(i))$ of
$\paramDeltatwo$--subdivisions, we match, within the corresponding blocks, the
$\ell$-th $\paramDeltaone$--subdivision of $\TVPDp$ to the $\ell$-th
$\paramDeltaone$--subdivision of $\TVPDq$, for $\ell = 1,\dots,k$.
For $\paramDeltatwo$--subdivisions that are unmatched in $\partialAssignmentBis$
(deletions or insertions), we declare the corresponding $k$
$\paramDeltaone$--subdivisions unmatched in $\partialAssignmentBis'$.
Because $\partialAssignmentBis$ is strictly increasing and the refinement inside
each block preserves the order, $\partialAssignmentBis'$ is strictly increasing
as well, hence
$\partialAssignmentBis' \in \PASet^{\paramDeltaone}(\TVPDp,\TVPDq)$.

It remains to compare the costs. Let
\[
  \TVPDp_\variablei
  \quad\text{and}\quad
  \TVPDq_\variablej
\]
be two matched $\paramDeltatwo$--subdivisions in $\partialAssignmentBis$,
with time intervals
\[
  \intervalI_\variablei^\TVPDp = (a,b),\qquad
  \intervalI_\variablej^\TVPDq = (c,d),
\]
so that $b-a = d-c = \paramDeltatwo$.
By definition of the local distance
$\localDistanceAppendix^{\paramWeight}_{\paramDeltatwo}$,
we have
\begin{align}
\label{eq:local_coarse}
  &\localDistanceAppendix^{\paramWeight}_{\paramDeltatwo}
    (\TVPDp_\variablei,\TVPDq_\variablej)
  \;\\&=\;
  \int_{a}^{b}
    \bigl(
      (1-\paramWeight)\,
        \wasserstein_2\bigl(\TVPDp(t),\TVPDq(t + c - a)\bigr)
      + \paramWeight\,|c-a|
    \bigr)\, dt.\notag
\end{align}

On the refined grid with step $\paramDeltaone$, the interval $(a,b)$ is
partitioned into
\begin{align}
  &(a_0,a_1],\ (a_1,a_2],\ \dots,\ (a_{k-1},a_k),
\quad
  \\&a_0 = a,\ a_k = b,\ a_{\ell+1} - a_\ell = \paramDeltaone,\notag
\end{align}
and similarly $(c,d)$ into
\begin{align}
  &(c_0,c_1],\ (c_1,c_2],\ \dots,\ (c_{k-1},c_k),
\quad
 \\& c_0 = c,\ c_k = d,\ c_{\ell+1} - c_\ell = \paramDeltaone,\notag
\end{align}
with $c_\ell - c = a_\ell - a$ for all $\ell$.
For each $\ell \in \{0,\dots,k-1\}$, the corresponding
$\paramDeltaone$--subdivisions $\TVPDp^{(\ell)},\TVPDq^{(\ell)}$
have time intervals $(a_\ell,a_{\ell+1})$ and $(c_\ell,c_{\ell+1})$.
Applying the definition of $\localDistanceAppendix^\paramWeight_{\paramDeltaone}$
to this pair gives
\begin{align}
\label{eq:fine_local_D1_raw}
  &\localDistanceAppendix^{\paramWeight}_{\paramDeltaone}
    (\TVPDp^{(\ell)},\TVPDq^{(\ell)})
  \\&=
  \int_{a_\ell}^{a_{\ell+1}}
    \Bigl(
      (1-\paramWeight)\,
        \wasserstein_2\bigl(\TVPDp(t),
                            \TVPDq(t + c_\ell - a_\ell)\bigr)
      \notag\\\notag&+ \paramWeight\,|c_\ell - a_\ell|
    \Bigr)\, dt.
\end{align}
Since $c_\ell - a_\ell = (c + (a_\ell-a)) - a_\ell = c - a$, \eqref{eq:fine_local_D1_raw}
simplifies to
\begin{align}
\label{eq:fine_local_D1}
  &\localDistanceAppendix^{\paramWeight}_{\paramDeltaone}
    (\TVPDp^{(\ell)},\TVPDq^{(\ell)})
  \\&\notag=
  \int_{a_\ell}^{a_{\ell+1}}
    \Bigl(
      (1-\paramWeight)\,
        \wasserstein_2\bigl(\TVPDp(t),\TVPDq(t + c - a)\bigr)
      \\\notag&+ \paramWeight\,|c-a|
    \Bigr)\, dt.
\end{align}
Summing \eqref{eq:fine_local_D1} over $\ell = 0,\dots,k-1$ and using the
additivity of the integral over the partition of $(a,b)$ yields
\begin{align}
  &\sum_{\ell=0}^{k-1}
    \localDistanceAppendix^{\paramWeight}_{\paramDeltaone}
      (\TVPDp^{(\ell)},\TVPDq^{(\ell)})
  \\&=
  \int_{a}^{b}
    \Bigl(
      (1-\paramWeight)\,
        \wasserstein_2\bigl(\TVPDp(t),\TVPDq(t + c - a)\bigr)
      \notag\\&\notag+ \paramWeight\,|c-a|
    \Bigr)\, dt
  =
  \localDistanceAppendix^{\paramWeight}_{\paramDeltatwo}
    (\TVPDp_i,\TVPDq_j).
\end{align}

An entirely similar computation holds for the deletion and insertion terms,
defined via $\localDistanceAppendix^{\paramWeight}_\paramDelta(\cdot,\boundarySet)$.
For instance, consider a deletion of a $\paramDeltatwo$--subdivision
$\TVPDp_i$ with time interval $\intervalI_i^\TVPDp = (a,b)$.
By definition,
\begin{equation}
\label{eq:coarse_delete_D2}
  \localDistanceAppendix^{\paramWeight}_{\paramDeltatwo}(\TVPDp_i,\boundarySet)
  =
  \int_{a}^{b}
    (1-\paramWeight)\,
      \wasserstein_2\bigl(\TVPDp(t),\boundarySet\bigr)\,dt.
\end{equation}
On the refined grid with step $\paramDeltaone$, the interval $(a,b)$ is
partitioned into $(a_0,a_1],\dots,(a_{k-1},a_k)$ as above, and the corresponding
$\paramDeltaone$--subdivisions $\TVPDp^{(\ell)}$ have time intervals
$(a_\ell,a_{\ell+1})$. For each $\ell\in\{0,\dots,k-1\}$, we have
\begin{equation}
\label{eq:fine_delete_D1}
  \localDistanceAppendix^{\paramWeight}_{\paramDeltaone}
    (\TVPDp^{(\ell)},\boundarySet)
  =
  \int_{a_\ell}^{a_{\ell+1}}
    (1-\paramWeight)\,
      \wasserstein_2\bigl(\TVPDp(t),\boundarySet\bigr)\,dt.
\end{equation}
Summing~\eqref{eq:fine_delete_D1} over $\ell$ and using the additivity of the
integral over the partition of $(a,b)$ yields
\begin{align}
  \sum_{\ell=0}^{k-1}
    \localDistanceAppendix^{\paramWeight}_{\paramDeltaone}
      (\TVPDp^{(\ell)},\boundarySet)
  &=
  \int_{a}^{b}
    (1-\paramWeight)\,
      \wasserstein_2\bigl(\TVPDp(t),\boundarySet\bigr)\,dt
  \notag\\&=
  \localDistanceAppendix^{\paramWeight}_{\paramDeltatwo}(\TVPDp_i,\boundarySet).
\end{align}

The case of insertions is handled in the same way, by applying the definition
of $\localDistanceAppendix^{\paramWeight}_\paramDelta(\cdot,\boundarySet)$ to a
$\paramDeltatwo$--subdivision $\TVPDq_j$ and its refinement into
$\paramDeltaone$--subdivisions $\TVPDq^{(\ell)}$.
Hence, the total deletion and insertion costs are preserved when passing from
$\paramDeltatwo$ to $\paramDeltaone$.

Therefore, the total cost of the refined assignment $\partialAssignmentBis'$ in
$\TVPDspace^{\paramDeltaone}$ satisfies
\[
  \text{cost}^{\paramDeltaone}_{(\TVPDp,\TVPDq)}(\partialAssignmentBis')
  =
  \text{cost}^{\paramDeltatwo}_{(\TVPDp,\TVPDq)}(\partialAssignmentBis)
  =
  \CEDM^{\paramDeltatwo}_{\paramWeight,\paramPenalty}(\TVPDp,\TVPDq).
\]
By minimality of $\CEDM^{\paramDeltaone}_{\paramWeight,\paramPenalty}$ over
$\PASet^{\paramDeltaone}(\TVPDp,\TVPDq)$, we obtain
\[
  \CEDM^{\paramDeltaone}_{\paramWeight,\paramPenalty}(\TVPDp,\TVPDq)
  \le
  \text{cost}^{\paramDeltaone}_{(\TVPDp,\TVPDq)}(\partialAssignmentBis')
  =
  \CEDM^{\paramDeltatwo}_{\paramWeight,\paramPenalty}(\TVPDp,\TVPDq),
\]
which proves the claim.

\end{proof}
	
        \begin{definition}Let $P = (P_i)_{1 \leq i \leq N_P} \in \mathcal{S}^\Delta$, $Q = (Q_j)_{1 \leq j \leq N_Q} \in \mathcal{S}^\Delta$. In this subsection, we will note for $v = (x, y) \in \mathbb{R}^2$, $v_1 = x$, $v_2 =y$, $\delta_\Delta^{\alpha,\beta}(v) =\delta_\Delta^{\alpha,\beta }\left((P_i)_{0 \leq i \leq v_1}, (Q_j)_{0 \leq j \leq v_2} \right)$. Then, we define recursively, $\forall K \in \{1, \dots, N_P\}$, $\forall K' \in \{1, \dots, N_Q\}$,

        \medskip

        $\begin{aligned}&\delta_\Delta^{\alpha,\beta} \left(K,K' \right)= \min \Bigr\{ \delta_\Delta^{\alpha,\beta} \left(K-1,K'\right)+ \beta\cdot \localDistanceAppendix_\Delta^\alpha (P_K,A),\\ & \delta_\Delta^{\alpha,\beta}\left(K-1, K'-1 \right) + \localDistanceAppendix_\Delta^\alpha (P_K, Q_{K'}),\\& \delta_\Delta^{\alpha,\beta} \left(K, K'-1\right)+ \beta \cdot \localDistanceAppendix_\Delta^\alpha(Q_{K'},A) \Bigr\}\end{aligned}$

        \medskip

        \noindent with initialization 
        $\delta_\Delta^{\alpha,\beta}\bigl(0,0\bigr)=0,\forall K\in\{1,\dots,N_P\},\,\delta_\Delta^{\alpha,\beta}\bigl(K,0\bigr)=\delta_\Delta^{\alpha,\beta}\bigl(K-1,0\bigr)+\beta \cdot \localDistanceAppendix_\Delta^\alpha(P_K,A),\text{ and } \forall K'\in\{1,\dots,N_Q\},\,\delta_\Delta^{\alpha,\beta}\bigl(0,K'\bigr)=\delta_\Delta^{\alpha,\beta}\bigl(0,K'-1\bigr)+\beta \cdot \localDistanceAppendix_\Delta^\alpha(Q_{K'},A).$

        Moreover, once the recursive computation of $\delta_\Delta^{\alpha,\beta} \left( (P_i)_{0 \leq i \leq N_P}, (Q_j)_{0 \leq j \leq N_Q} \right)$ has been carried out, we define recursively, with initializations $A_0 = \emptyset$, $B_0 = (N_P,N_Q)$,

\begin{itemize}
	\item $B_{z+1} = B_z - (1, 0)$ if 
	\[
	\delta_\Delta^{\alpha,\beta}(B_z) = \delta_\Delta^{\alpha,\beta}(B_z - (1,0)) +\beta\cdot \localDistanceAppendix^\alpha_\Delta(P_{B_{z_1}}, A),
	\]
	
	\item $B_{z+1} = B_z - (1, 1)$ if 
	\[
	\delta_\Delta^{\alpha,\beta}(B_z) = \delta_\Delta^{\alpha,\beta}(B_z - (1,1)) + \localDistanceAppendix^\alpha_\Delta(P_{B_{z_1}}, Q_{B_{z_2}}),
	\]
	
	\item $B_{z+1} = B_z - (0, 1)$ if 
	\[
	\delta_\Delta^{\alpha,\beta}(B_z) = \delta_\Delta^{\alpha,\beta}(B_z - (0,1)) + \beta \cdot \localDistanceAppendix^\alpha_\Delta(Q_{B_{z_2}}, A),
	\]

	\item \(A_{z+1} = A_z \cup B_z \quad \text{if} \quad B_{z+1} = B_z - (1,1)\),
	\item \(A_{z+1} = A_z \:\:\quad\qquad\text{if }B_{z+1} = B_z - (1, 0)\),
	\item \(A_{z+1} = A_z \:\:\quad\qquad\text{if } B_{z+1} = B_z - (0, 1)\).   
\end{itemize}

\medskip

We stop when $B_Z = (0,0)$ for some $Z \in \mathbb{N}$.
		    
	\end{definition}
	
	\medskip

    \begin{proposition}\label{prop:distance_DP_equality}
Let $(P,Q) \in \mathcal{S}^\Delta$. Then
\[
\delta_\Delta^{\alpha,\beta} \left( (P_i)_{0 \leq i \leq N_P^\Delta}, (Q_j)_{0 \leq j \leq N_Q^\Delta} \right) = \mathrm{CED}^\Delta_{\alpha,\beta}(P,Q).
\]
\end{proposition}

\begin{proof}
Let $\mathcal{I} = N_P$ and $\mathcal{J} = N_Q$. We prove the proposition by induction on $\mathcal{I} + \mathcal{J}$.

\noindent— Suppose that $\mathcal{I} = 1$, and $\mathcal{J} = 1$.

\vspace{-0.5em}
\begin{align*}
\delta_\Delta^{\alpha,\beta}(1,1) &= \min\Big(
\delta_\Delta^{\alpha,\beta}(0,1) + \beta \cdot D^\alpha_\Delta(P_1, A), \\
&\qquad\delta_\Delta^{\alpha,\beta}(0,0) + D^\alpha_\Delta(P_1,Q_1), \\
&\qquad\delta_\Delta^{\alpha,\beta}(1,0) + \beta \cdot D^\alpha_\Delta(Q_1, A)
\Big) \\
&= \min\Big(
\beta \cdot D^\alpha_\Delta(Q_1,A) + \beta \cdot D^\alpha_\Delta(P_1,A),\ \\&D^\alpha_\Delta(P_1,Q_1), 
\beta \cdot D^\alpha_\Delta(P_1,A) + \beta \cdot D^\alpha_\Delta(Q_1,A)
\Big) \\
&= \min_{f \in \mathcal{A}^\Delta(P,Q)} \text{cost}^\paramDelta_{(\TVPDp,\TVPDq)}(\partialAssignment) = \mathrm{CED}^\Delta_{\alpha,\beta}(P,Q).
\end{align*}

Then the proposition is true for $\mathcal{I} = 1$, and $\mathcal{J} = 1$.

\noindent— Suppose that the proposition is true for all indices $(\mathcal{I}',\mathcal{J}')$ with $\mathcal{I}' + \mathcal{J}' < \mathcal{I} + \mathcal{J}$.

Firstly, note that:
\begin{flalign}
&\mathcal{A}^\Delta(P,Q) = 
\left\{ f \in \mathcal{A}^\Delta(P,Q) \mid f(\mathcal{I}) = \mathcal{J} \right\}
\notag\\&\cup
\left\{ f \in \mathcal{A}^\Delta(P,Q) \mid \mathcal{I} \notin \mathrm{dom}(f) \right\}
\notag\\&\cup
\left\{ f \in \mathcal{A}^\Delta(P,Q) \mid \mathcal{J} \notin \mathrm{Im}(f) \right\}.
\end{flalign}

\textit{(i)} Assume that the optimal $\Delta$–partial assignment $f$ between $P$ and $Q$ belongs to\(\left\{ f \in \mathcal{A}^\Delta(P,Q) \mid f(\mathcal{I}) = \mathcal{J} \right\}.\)

Then 
\begin{flalign}&\mathrm{CED}^\Delta_{\alpha,\beta}(P,Q) = \text{cost}^\paramDelta_{(\TVPDp,\TVPDq)}(\partialAssignment)
\notag\\&= \mathrm{cost}^\Delta_{\left((P_i)_{1 \leq i \leq \mathcal{I}-1},\ (Q_j)_{1 \leq j \leq \mathcal{J}-1}\right)}(f^{(\mathcal{I},\mathcal{J})*}) + D^\alpha_\Delta(P_{\mathcal{I}}, Q_{\mathcal{J}})
\end{flalign}

\noindent
with $f^{(\mathcal{I},\mathcal{J})*}$ the $\Delta$-partial assignment between $(P_i)_{1 \leq i \leq \mathcal{I}-1}$ and $(Q_j)_{1 \leq j \leq \mathcal{J}-1}$ obtained from $f$ by removing the substitution of $P_{\mathcal{I}}$ to $Q_{\mathcal{J}}$ from $f$. Thus $f^{(\mathcal{I},\mathcal{J})*}$ is the optimal $\Delta$-partial assignment between $(P_i)_{1 \leq i \leq \mathcal{I}-1}$ and $(Q_j)_{1 \leq j \leq \mathcal{J}-1}$. Indeed if this were not the case, and if $g \neq f^{(\mathcal{I},\mathcal{J})*}$ were the optimal $\Delta$-partial assignment between $(P_i)_{1 \leq i \leq \mathcal{I}-1}$ and $(Q_j)_{1 \leq j \leq \mathcal{J}-1}$, then $g$ to which we add the substitution of $P_{\mathcal{I}}$ by $Q_{\mathcal{J}}$ would have a cost lower than $f$, which would contradict that $f$ is the optimal assignment between $P$ and $Q$.

Then
\begin{flalign}
\mathrm{CED}^\Delta_{\alpha,\beta}&\left((P_i)_{1 \leq i \leq \mathcal{I}-1}, (Q_j)_{1 \leq j \leq \mathcal{J}-1}\right) \notag\\
&= \mathrm{cost}^\Delta_{\left((P_i)_{1 \leq i \leq \mathcal{I}-1}, (Q_j)_{1 \leq j \leq \mathcal{J}-1}\right)}(f^{(\mathcal{I},\mathcal{J})*})
\end{flalign}

\noindent
and so
\begin{flalign}
\mathrm{CED}^\Delta_{\alpha,\beta}(P,Q) =
\mathrm{CED}^\Delta_{\alpha,\beta}\left((P_i)_{1 \leq i \leq \mathcal{I}-1}, (Q_j)_{1 \leq j \leq \mathcal{J}-1}\right)
\notag\\+ D^\alpha_\Delta(P_{\mathcal{I}}, Q_{\mathcal{J}}).
\end{flalign}

Thus,
\begin{flalign}
&\delta_\Delta^{\alpha,\beta}(\mathcal{I}, \mathcal{J}) =\min \Big\{\delta_\Delta^{\alpha,\beta}(\mathcal{I} - 1, \mathcal{J}) + \beta\cdot D^\alpha_\Delta(P_{\mathcal{I}}, A), \notag\\&\delta_\Delta^{\alpha,\beta}(\mathcal{I} - 1, \mathcal{J} -1) + D^\alpha_\Delta(P_{\mathcal{I}}, Q_{\mathcal{J}}), \notag\\&\delta_\Delta^{\alpha,\beta}(\mathcal{I}, \mathcal{J} - 1) +\beta \cdot D^\alpha_\Delta(Q_{\mathcal{J}}, A) \Big\}\\&= \min \Big\{\delta_\Delta^{\alpha,\beta}(\mathcal{I} - 1, \mathcal{J}) + \beta\cdot D^\alpha_\Delta(P_{\mathcal{I}}, A), \notag\\&\mathrm{CED}^\Delta_{\alpha,\beta}\big((P_i)_{1 \leq i \leq\mathcal{I} - 1}, (Q_j)_{1 \leq j \leq \mathcal{J} - 1} \big)+ D^\alpha_\Delta(P_{\mathcal{I}}, Q_{\mathcal{J}}), \notag\\&\delta_\Delta^{\alpha,\beta}(\mathcal{I}, \mathcal{J} - 1) +\beta \cdot D^\alpha_\Delta(Q_{\mathcal{J}}, A) \Big\} \text{(by induction)}\\&=\min \Big\{\delta_\Delta^{\alpha,\beta}(\mathcal{I} - 1, \mathcal{J}) + \beta\cdot D^\alpha_\Delta(P_{\mathcal{I}}, A), \mathrm{CED}^\Delta_{\alpha,\beta}(P,Q), \notag\\&\delta_\Delta^{\alpha,\beta}(\mathcal{I}, \mathcal{J} - 1) +\beta \cdot D^\alpha_\Delta(Q_{\mathcal{J}}, A)\Big\}\leq\mathrm{CED}^\Delta_{\alpha,\beta}(P,Q)\end{flalign}

\textit{(ii)} Assume that the optimal $\Delta$–partial assignment $f$ between $P$ and $Q$ belongs to
$\left\{ f \in \mathcal{A}^\Delta(P,Q) \text{ such that } \mathcal{I} \notin \mathrm{dom}(f) \right\}.$

Then 
\begin{flalign}
&\mathrm{CED}^\Delta_{\alpha,\beta}(P,Q) = \text{cost}^\paramDelta_{(\TVPDp,\TVPDq)}(\partialAssignment)
\\&= \mathrm{cost}^\Delta_{\left((P_i)_{1 \leq i \leq \mathcal{I}-1}, (Q_j)_{1 \leq j \leq \mathcal{J}} \right)}(f^{\mathcal{I}*}) 
+ \beta \cdot D^\alpha_\Delta(P_{\mathcal{I}}, A),
\end{flalign}

\noindent with $f^{\mathcal{I}*}$ the $\Delta$–partial assignment between $(P_i)_{1 \leq i \leq \mathcal{I}-1}$ and 
$(Q_j)_{1 \leq j \leq \mathcal{J}}$ obtained from $f$ by removing the deletion of $P_{\mathcal{I}}$ from $f$. Thus $f^{\mathcal{I}*}$ is the optimal $\Delta$–partial assignment between $(P_i)_{1 \leq i \leq \mathcal{I}-1}$ and $(Q_j)_{1 \leq j \leq \mathcal{J}}$. Indeed if this were not the case, and if $g \neq f^{\mathcal{I}*}$ were the optimal $\Delta$–partial assignment between $(P_i)_{1 \leq i \leq \mathcal{I}-1}$ and $(Q_j)_{1 \leq j \leq \mathcal{J}}$, then $g$ to which we add the deletion of $P_{\mathcal{I}}$ would have a cost lower than $f$, which would contradict that $f$ is the optimal assignment between $P$ and $Q$.

Then
\begin{align}
\mathrm{CED}^\Delta_{\alpha,\beta}(P,Q) =
\mathrm{CED}^\Delta_{\alpha,\beta}\left((P_i)_{1 \leq i \leq \mathcal{I}-1}, (Q_j)_{1 \leq j \leq \mathcal{J}} \right)
\notag\\+ \beta \cdot D^\alpha_\Delta(P_{\mathcal{I}}, A).
\end{align}

Thus,\begin{flalign}
\hspace{0pt}&\delta_\Delta^{\alpha,\beta}(\mathcal{I},\mathcal{J}) 
= \min \Big\{
\delta_\Delta^{\alpha,\beta}(\mathcal{I} - 1, \mathcal{J}) + \beta \cdot D^\alpha_\Delta(P_{\mathcal{I}}, A), \notag\\&\delta_\Delta^{\alpha,\beta}(\mathcal{I} - 1, \mathcal{J} - 1) + D^\alpha_\Delta(P_{\mathcal{I}}, Q_{\mathcal{J}}), \notag\\
&\delta_\Delta^{\alpha,\beta}(\mathcal{I}, \mathcal{J} - 1) + \beta \cdot D^\alpha_\Delta(Q_{\mathcal{J}}, A) \Big\} \\
&= \min \Big\{
\mathrm{CED}^\Delta_{\alpha,\beta}\left((P_i)_{1 \le i \le \mathcal{I} - 1}, (Q_j)_{1 \le j \le \mathcal{J}} \right) \notag\\&+ \beta \cdot D^\alpha_\Delta(P_{\mathcal{I}}, A), 
\delta_\Delta^{\alpha,\beta}(\mathcal{I} - 1, \mathcal{J} - 1) + D^\alpha_\Delta(P_{\mathcal{I}}, Q_{\mathcal{J}}),\notag \\
&\delta_\Delta^{\alpha,\beta}(\mathcal{I}, \mathcal{J} - 1) + \beta \cdot D^\alpha_\Delta(Q_{\mathcal{J}}, A) \Big\} (\text{by induction}) \\&= \min \Big\{
\mathrm{CED}^\Delta_{\alpha,\beta}(P,Q), \notag\\
&\delta_\Delta^{\alpha,\beta}(\mathcal{I} - 1, \mathcal{J} - 1) + D^\alpha_\Delta(P_{\mathcal{I}}, Q_{\mathcal{J}}), \notag\\
&\delta_\Delta^{\alpha,\beta}(\mathcal{I}, \mathcal{J} - 1) + \beta \cdot D^\alpha_\Delta(Q_{\mathcal{J}}, A) \Big\} \\
&\leq \mathrm{CED}^\Delta_{\alpha,\beta}(P,Q).
\end{flalign}

\textit{(iii)} Assume that the optimal $\Delta$–partial assignment $f$ between $P$ and $Q$ belongs to $\left\{ f \in \mathcal{A}^\Delta(P,Q) \text{ such that } \mathcal{J} \notin \mathrm{Im}(f) \right\}.$

\begin{align}
\hspace{0pt}\mathrm{CED}^\Delta_{\alpha,\beta}(P,Q) &= \text{cost}^\paramDelta_{(\TVPDp,\TVPDq)}(\partialAssignment)  \\
&= \mathrm{cost}^\Delta_{\left((P_i)_{1 \leq i \leq \mathcal{I}},\, (Q_j)_{1 \leq j \leq \mathcal{J}-1} \right)}(f^{\mathcal{J}*}) \notag \\
&+ \beta \cdot D^\alpha_\Delta(Q_{\mathcal{J}}, A)
\end{align}


\noindent
with $f^{\mathcal{J}*}$ the $\Delta$–partial assignment between $(P_i)_{1 \leq i \leq \mathcal{I}}$ and $(Q_j)_{1 \leq j \leq \mathcal{J}-1}$ obtained from $f$ by removing the insertion of $Q_{\mathcal{J}}$ from $f$. Thus $f^{\mathcal{J}*}$ is the optimal $\Delta$–partial assignment between $(P_i)_{1 \leq i \leq \mathcal{I}}$ and $(Q_j)_{1 \leq j \leq \mathcal{J}-1}$. Indeed if this were not the case, and if $g \neq f^{\mathcal{J}*}$ were the optimal $\Delta$–partial assignment between $(P_i)_{1 \leq i \leq \mathcal{I}}$ and $(Q_j)_{1 \leq j \leq \mathcal{J}-1}$, then $g$ to which we add the insertion of $Q_{\mathcal{J}}$ would have a cost lower than $f$, which would contradict that $f$ is the optimal assignment between $P$ and $Q$.

\begin{align}
\mathrm{CED}^\Delta_{\alpha,\beta}(P,Q) =
\mathrm{CED}^\Delta_{\alpha,\beta}\left((P_i)_{1 \leq i \leq \mathcal{I}},\, (Q_j)_{1 \leq j \leq \mathcal{J}-1} \right)\notag\\+ \beta \cdot D^\alpha_\Delta(Q_{\mathcal{J}}, A).
\end{align}
Then, 
\begin{align}
\hspace{0pt}&\delta^{\alpha,\beta}_\Delta(\mathcal{I},\mathcal{J}) 
= \min \Big\{ \delta^{\alpha,\beta}_\Delta(\mathcal{I}-1,\mathcal{J}) \notag\\
&\quad + \beta \cdot D^\alpha_\Delta(P_{\mathcal{I}},A), 
\delta^{\alpha,\beta}_\Delta(\mathcal{I}-1,\mathcal{J}-1) 
+ D^\alpha_\Delta(P_{\mathcal{I}},Q_{\mathcal{J}}), \notag\\
&\quad\delta^{\alpha,\beta}_\Delta(\mathcal{I},\mathcal{J}-1) 
+ \beta \cdot D^\alpha_\Delta(Q_{\mathcal{J}},A) \Big\} \\
&= \min \Big\{ \delta^{\alpha,\beta}_\Delta(\mathcal{I}-1,\mathcal{J}) 
+ \beta \cdot D^\alpha_\Delta(P_{\mathcal{I}},A),\notag \\
&\quad \delta^{\alpha,\beta}_\Delta(\mathcal{I}-1,\mathcal{J}-1) 
+ D^\alpha_\Delta(P_{\mathcal{I}},Q_{\mathcal{J}}), \notag\\
&\quad \mathrm{CED}^{\Delta}_{\alpha,\beta}((P_i)_{1\le i\le \mathcal{I}},(Q_j)_{1\le j\le \mathcal{J}-1}) 
+ \beta \cdot D^\alpha_\Delta(Q_{\mathcal{J}},A) \Big\} \\
&= \min \Big\{ \delta^{\alpha,\beta}_\Delta(\mathcal{I}-1,\mathcal{J}) 
+ \beta \cdot D^\alpha_\Delta(P_{\mathcal{I}},A), \notag\\
&\quad \delta^{\alpha,\beta}_\Delta(\mathcal{I}-1,\mathcal{J}-1) 
+ D^\alpha_\Delta(P_{\mathcal{I}},Q_{\mathcal{J}}), 
\mathrm{CED}^\Delta_{\alpha,\beta}(P,Q) \Big\} \notag\\
&\leq \mathrm{CED}^\Delta_{\alpha,\beta}(P,Q)
\end{align}

\textit{(i)} Conversely, if
\begin{equation}
\delta^{\alpha,\beta}_\Delta(\mathcal{I},\mathcal{J}) 
= \delta^{\alpha,\beta}_\Delta(\mathcal{I}-1,\mathcal{J}) 
+ \beta \cdot D^\alpha_\Delta(P_{\mathcal{I}},A), \\
\end{equation}
\text{then by induction } 
\begin{align}
\delta^{\alpha,\beta}_\Delta(\mathcal{I},\mathcal{J}) 
= \mathrm{CED}^\Delta_{\alpha,\beta}\big((P_i)_{1 \le i \le \mathcal{I}-1}, (Q_j)_{1 \le j \le \mathcal{J}}\big) 
\\+ \beta \cdot D^\alpha_\Delta(P_{\mathcal{I}},A) \notag\\
= \mathrm{cost}^\Delta_{((P_i)_{1 \le i \le \mathcal{I}-1}, (Q_j)_{1 \le j \le \mathcal{J}})}(f^{\mathcal{I}*}) 
\notag\\+ \beta \cdot D^\alpha_\Delta(P_{\mathcal{I}},A)
\end{align}
\noindent
with $f^{\mathcal{I}*}$ the optimal $\Delta$–partial assignment between $(P_i)_{1 \le i \le \mathcal{I}-1}$ and $(Q_j)_{1 \le j \le \mathcal{J}}$. Then append to $f^{\mathcal{I}*}$ the deletion of $P_{\mathcal{I}}$ resulting to a partial assignment between $P$ and $Q$ with cost equal to
\begin{align}
\mathrm{cost}^\Delta_{((P_i)_{1 \le i \le \mathcal{I}}, (Q_j)_{1 \le j \le \mathcal{J}})}(f^{\mathcal{I}*}) 
+ \beta \cdot D^\alpha_\Delta(P_{\mathcal{I}},A) 
\notag\\= \delta^{\alpha,\beta}_\Delta(\mathcal{I},\mathcal{J})
\end{align}
thus
\begin{align}
 \delta^{\alpha,\beta}_\Delta(\mathcal{I},\mathcal{J}) 
\ge \min_{f \in \mathcal{A}^\Delta(P,Q)} \mathrm{cost}^\Delta_{(P,Q)}(f) 
= \mathrm{CED}^\Delta_{\alpha,\beta}(P,Q)
\end{align}
\textit{(ii)} If
\begin{equation}
\delta^{\alpha,\beta}_\Delta(\mathcal{I},\mathcal{J}) = \delta^{\alpha,\beta}_\Delta(\mathcal{I}-1,\mathcal{J}-1) + D^\alpha_\Delta(P_{\mathcal{I}}, Q_{\mathcal{J}}), 
\end{equation}
then by induction,  
\begin{align}
\delta^{\alpha,\beta}_\Delta(\mathcal{I},\mathcal{J}) = \mathrm{CED}^\Delta_{\alpha,\beta}((P_i)_{1\le i \le \mathcal{I}-1},(Q_j)_{1\le j \le \mathcal{J}-1}) 
\notag\\+ D^\alpha_\Delta(P_{\mathcal{I}}, Q_{\mathcal{J}}) \\
= \mathrm{cost}^\Delta_{((P_i)_{1\le i \le \mathcal{I}-1},(Q_j)_{1\le j \le \mathcal{J}-1})}(f^{(\mathcal{I},\mathcal{J})*})  \notag\\+D^\alpha_\Delta(P_{\mathcal{I}}, Q_{\mathcal{J}})
\end{align}
\noindent with $f^{(\mathcal{I},\mathcal{J})*}$ the optimal $\Delta$–partial assignment between $(P_i)_{1\le i \le \mathcal{I}-1}$ and $(Q_j)_{1\le j \le \mathcal{J}-1}$. Then append to $f^{(\mathcal{I},\mathcal{J})*}$ the substitution of $P_{\mathcal{I}}$ by $Q_{\mathcal{J}}$, resulting to a partial assignment between $P$ and $Q$ with cost equal to
\begin{align}
\mathrm{cost}^\Delta_{((P_i)_{1\le i \le \mathcal{I}},(Q_j)_{1\le j \le \mathcal{J}})}(f^{(\mathcal{I},\mathcal{J})*}) + D^\alpha_\Delta(P_{\mathcal{I}}, Q_{\mathcal{J}})\notag\\ = \delta^{\alpha,\beta}_\Delta(\mathcal{I},\mathcal{J}),
\end{align}
thus
\begin{equation}
\delta^{\alpha,\beta}_\Delta(\mathcal{I},\mathcal{J}) \ge \min_{f \in \mathcal{A}^\Delta(P,Q)} \mathrm{cost}^\Delta_{(P,Q)}(f) = \mathrm{CED}^\Delta_{\alpha,\beta}(P,Q).
\end{equation}
\textit{(iii)} If
\begin{equation}
\delta^{\alpha,\beta}_\Delta(\mathcal{I},\mathcal{J})
= \delta^{\alpha,\beta}_\Delta(\mathcal{I},\mathcal{J}-1) 
 + \beta\cdot\, D^\alpha_\Delta(Q_{\mathcal{J}},A), \\
\end{equation}
then by induction,
\begin{align} 
\delta^{\alpha,\beta}_\Delta(\mathcal{I},\mathcal{J})
= \mathrm{CED}^\Delta_{\alpha,\beta}\Bigl(
        (P_i)_{\substack{1\le i\le \mathcal{I}}},
 (Q_j)_{\substack{1\le j\le \mathcal{J}-1}}\Bigr) \notag \\
+ \beta\cdot\, D^\alpha_\Delta(Q_{\mathcal{J}},A), \\[2pt]
= \mathrm{cost}^\Delta_{\Bigl(
        (P_i)_{\substack{1\le i\le \mathcal{I}}}, (Q_j)_{\substack{1\le j\le \mathcal{J}-1}}\Bigr)}%
        \bigl(f^{\mathcal{J}*}\bigr) \notag \\
 + \beta\cdot\, D^\alpha_\Delta(Q_{\mathcal{J}},A),
\end{align}
\noindent with $f^{\mathcal{J}*}$ the optimal $\Delta$–partial assignment between $(P_i)_{1\le i \le \mathcal{I}}$ and $(Q_j)_{1\le j \le \mathcal{J}-1}$. Then append to $f^{\mathcal{J}*}$ the insertion of $Q_{\mathcal{J}}$ resulting to a partial assignment between $P$ and $Q$ with cost equal to
\begin{align}
\mathrm{cost}^\Delta_{((P_i)_{1\le i \le \mathcal{I}}, (Q_j)_{1\le j \le \mathcal{J}})}(f^{\mathcal{J}*}) + \beta \cdot D^\alpha_\Delta(Q_{\mathcal{J}},A) \notag\\=\delta^{\alpha,\beta}_\Delta(\mathcal{I},\mathcal{J}),
\end{align}
thus
\begin{align}
\delta^{\alpha,\beta}_\Delta(\mathcal{I},\mathcal{J}) \ge \min_{f \in \mathcal{A}^\Delta(P,Q)} \mathrm{cost}^\Delta_{(P,Q)}(f) \notag\\= \mathrm{CED}^\Delta_{\alpha,\beta}(P,Q).
\end{align}

Thus if the proposition is true for all indices $(\mathcal{I}',\mathcal{J}')$ with $\mathcal{I}' + \mathcal{J}' < \mathcal{I} + \mathcal{J}$, then the proposition is true for indices $\mathcal{I}$ and $\mathcal{J}$.

\end{proof}

	\begin{theorem}\label{thm:metric_space_global}
	The space $(\mathcal{S}^\Delta, \mathrm{CED}^\Delta_{\alpha,\beta})$ is a metric space.
		    
	\end{theorem}
	\begin{proof}

\textbullet \quad Triangle inequality\\
We will prove triangle inequality by induction on $M=I+J+K$.
\begin{itemize}
    \item The triangle inequality is true for $M=0$ since $\delta_\Delta^{\alpha,\beta} \left( (P_i)_{0 \leq i \leq 0}, (R_k)_{0 \leq k \leq 0} \right)=0 \geq 0+0 = \delta_\Delta^{\alpha,\beta} \left( (P_i)_{0 \leq i \leq 0}, (Q_j)_{0 \leq j \leq 0} \right)+ \delta_\Delta^{\alpha,\beta} \left( (Q_j)_{0 \leq j \leq 0}, (R_k)_{0 \leq k \leq 0} \right)$

    \item Suppose the triangle inequality is true for all $n\in \{0,...,M-1\}$ for some $M>0$. For simplicity, we will denote $(P_I)=(P_i)_{0\leq i\leq I}$, $(Q_J)=(Q_j)_{0\leq j\leq J}$, $(R_K)=(R_k)_{0\leq k\leq K}$. Then we have different cases :\\
    $1^{st}$ case :\\
    If $\delta_\Delta^{\alpha,\beta} ((P_I),(Q_J))=\delta_\Delta^{\alpha,\beta} ((P_{I-1}) , (Q_J) ) +\beta \cdot D^\alpha_\Delta(P_I,A)$. Then $\delta_\Delta^{\alpha,\beta} ((P_I),(Q_J))+ \delta_\Delta^{\alpha,\beta} ((Q_J),(R_K))=\delta_\Delta^{\alpha,\beta} ((P_{I-1}) , (Q_J )) + \delta_\Delta^{\alpha,\beta} ((Q_J) , (R_K )) +\beta \cdot D^\alpha_\Delta(P_I,A)\\ \geq \delta_\Delta^{\alpha,\beta}((P_{I-1}) , (R_K )) +\beta \cdot D^\alpha_\Delta(P_I,A)$ (by triangular inequality of $\delta_\Delta^{\alpha,\beta}$ until $M-1$)\\
    $\geq \delta_\Delta^{\alpha,\beta}((P_{I}) , (R_K ))$ (by the recursive definition of $\delta_\Delta^{\alpha,\beta}$)
    \\
    $2^{nd}$ case :\\
    If $\left\{\begin{array}{l}
        \delta_\Delta^{\alpha,\beta}((P_I),(Q_J))=\delta_\Delta^{\alpha,\beta} ((P_{I}) , (Q_{J-1} ))\\+\beta \cdot D^\alpha_\Delta(Q_J ,A)\\
        \delta_\Delta^{\alpha,\beta} ((Q_J),(R_K))=\delta_\Delta^{\alpha,\beta}((Q_{J-1}) , (R_{K-1} ))\\ + D^\alpha_\Delta(Q_J ,R_K )
    \end{array}\right.$\\
    Then $\delta_\Delta^{\alpha,\beta}((P_I),(Q_J))+ \delta_\Delta^{\alpha,\beta} ((Q_J),(R_K))\\ \geq (\delta_\Delta^{\alpha,\beta} ((P_{I}) , (Q_{J-1} )) + \delta_\Delta^{\alpha,\beta} ((Q_{J-1}) , (R_{K-1} ))) +(\beta \cdot D^\alpha_\Delta(Q_J ,A)+ D^\alpha_\Delta(Q_J ,R_K)) \\\geq (\delta_\Delta^{\alpha,\beta} ((P_{I}) , (Q_{J-1} )) + \delta_\Delta^{\alpha,\beta} ((Q_{J-1}) , (R_{K-1} ))) +\beta \cdot (D^\alpha_\Delta(Q_J ,A)+ D^\alpha_\Delta(Q_J ,R_K))\\ \geq (\delta_\Delta^{\alpha,\beta} ((P_{I}) , (R_{K-1} ))) +\beta \cdot (D^\alpha_\Delta(Q_J ,A)+ D^\alpha_\Delta(Q_J ,R_K))$\\(by triangular inequality of $\delta_\Delta^{\alpha,\beta}$ until $M-1$)\\ $\geq(\delta_\Delta^{\alpha,\beta} ((P_{I}) , (R_{K-1} ))) +\beta \cdot D^\alpha_\Delta(R_K ,A)\text{ (by Lemma~\ref{lem:triangular_inequality_local_distance_lemma})}\\\geq \delta_\Delta^{\alpha,\beta} ((P_I),(R_K))$\\\\
    $3^{rd}$ case :\\
    If $\left\{\begin{array}{l}
       \delta_\Delta^{\alpha,\beta}((P_I),(Q_J))=\delta_\Delta^{\alpha,\beta}((P_{I-1}) , (Q_{J-1} )) \\+D^\alpha_\Delta(P_I ,Q_J)\\
        \delta_\Delta^{\alpha,\beta}((Q_J),(R_K))=\delta_\Delta^{\alpha,\beta}((Q_{J-1}) , (R_{K-1} )) \\+D^\alpha_\Delta(Q_J ,R_K )
    \end{array}\right.$\\
    $\text{Then }\delta_\Delta^{\alpha,\beta}((P_I),(Q_J))+\delta_\Delta^{\alpha,\beta}((Q_J),(R_K))\geq \delta_\Delta^{\alpha,\beta}((P_{I-1}) , (Q_{J-1} ))+\delta_\Delta^{\alpha,\beta}((Q_{J-1}) , (R_{K-1} )) +D^\alpha_\Delta(P_I ,Q_J)+D^\alpha_\Delta(Q_J ,R_K )$ \\  
    $\geq \delta_\Delta^{\alpha,\beta}((P_{I-1}) , (R_{K-1} )) +D^\alpha_\Delta(P_I ,R_K)$ (by triangular inequalitys of $\delta_\Delta^{\alpha,\beta}$ (until $M-1$) and $D^\alpha_\Delta$ )\\
    $\geq \delta_\Delta^{\alpha,\beta}((P_I),(R_K))$\\
    Then other cases are trivial or equivalent by symmetry to the first tree cases. \\
    \textbullet \quad Symmetry and positivity\\
    Since the distance $D^\alpha_\Delta$ is symmetric and positive, it is easy to show by induction that $\delta_\Delta^{\alpha,\beta}$ is symetric and positive.\\
    \textbullet \quad $P =Q \Rightarrow \mathrm{CED}^\Delta_{\alpha,\beta}(P,Q)=0$\\
    By the identity partial assignment, \\
    $\begin{array}{lll}
        &\mathrm{CED}^\Delta_{\alpha,\beta}(P,Q) \leq \text{cost~} id  = D^\alpha_\Delta(P_1 ,Q_1) + ... \\&+ D^\alpha_\Delta(P_{N_P^\Delta} ,Q_{N_Q^\Delta})\\
        & =  0 + ... + 0\\
        & =  0
    \end{array}$\\

    where $id$ is the identity partial assignment $\{(1, 1),...,({N_P^\Delta}, {N_P^\Delta})\}$.
    \\
     \textbullet \quad $\mathrm{CED}^\Delta_{\alpha,\beta}(P,Q)=0 \Rightarrow  P =Q$\\
     Suppose $ \mathrm{CED}^\Delta_{\alpha,\beta}(P,Q)=0$.\\
     Then, $\mathrm{CED}^\Delta_{\alpha,\beta}(P,Q) = \mathrm{cost}^\Delta_{(P,Q)}(f)$ for $f$ a optimal partial assignment of $(P,Q)$. Since $ \mathrm{CED}^\Delta_{\alpha,\beta}(P,Q)=0$, $f$ doesn't have any deletion or insertion (because $D^\alpha_\Delta(P_i, A)>0$ and $D^\alpha_\Delta(Q_j,A)>0$ by definition of TVPDs). Then $N_P^\Delta=N_Q^\Delta$ and $f=\{(1, 1),...,(N_P^\Delta, N_P^\Delta)\}$. Hence, $\forall i\in \{1,...,N_P^\Delta\}$, $D^\alpha_\Delta(P_i, Q_i)=0$ and thus, $\forall i\in \{1,...,N_P^\Delta\}$, $P_i = Q_i$. Finally, $P = Q$.\qedhere
\end{itemize}
    \end{proof}
	
	\medskip

\begin{lemma}
\label{lem:dp_backtracking_correct}
Let $\TVPDp \in \TVPDspace^\paramDelta$ and $\TVPDq \in \TVPDspace^\paramDelta$.
Then:
\begin{enumerate}
  \item The backtracking procedure starting from
  $\initDPTwo_0 = (\variableN_\TVPDp,\variableN_\TVPDq)$
  and constructing $(\initDPTwo_\variablez)_\variablez$ by steps
  $(-1,0)$, $(-1,-1)$, or $(0,-1)$ as specified, terminates after a finite
  number of steps at $\initDPTwo_\variableZ = (0,0)$ for some
  $\variableZ \in \NNB$.

  \item The set $\partialAssignment := \initDPOne_\variableZ$ is a
  $\paramDelta$-partial assignment between $\TVPDp$ and $\TVPDq$, i.e.\
  $\partialAssignment \in \PASet^\paramDelta(\TVPDp,\TVPDq)$.

  \item The cost of $\partialAssignment$ satisfies
  \[
    \mathrm{cost}^\paramDelta_{(\TVPDp,\TVPDq)}(\partialAssignment)
    \;=\;
    \delta_\paramDelta^{\paramWeight,\paramPenalty}
      \left( (\TVPDp_\variablei)_{0 \leq \variablei \leq \variableN_\TVPDp}, (\TVPDq_\variablej)_{0 \leq \variablej \leq \variableN_\TVPDq} \right),
  \]
  so that $\partialAssignment$ achieves the minimum in the definition of
  $\CEDM^{\paramDelta}_{\paramWeight,\paramPenalty}(\TVPDp,\TVPDq)$, and in
  particular
  \[
    \CEDM^{\paramDelta}_{\paramWeight,\paramPenalty}(\TVPDp,\TVPDq)
    \;=\;
    \mathrm{cost}^\paramDelta_{(\TVPDp,\TVPDq)}(\partialAssignment).
  \]
\end{enumerate}
\end{lemma}

\begin{proof}
Let us write $\initDPTwo_\variablez = (i_\variablez,j_\variablez)$. By construction of the backtracking: each step replaces $\initDPTwo_\variablez$ by $\initDPTwo_{\variablez+1} = \initDPTwo_\variablez-(1,0)$, or $\initDPTwo_{\variablez+1} = \initDPTwo_\variablez-(1,1)$, or $\initDPTwo_{\variablez+1} = \initDPTwo_\variablez-(0,1)$.
Hence
\[
  i_{\variablez+1} \leq i_\variablez,\quad
  j_{\variablez+1} \leq j_\variablez,\quad
  i_{z+1} + j_{z+1} \le i_z + j_z - 1,
\]
so the quantity $i_\variablez + j_\variablez$ strictly decreases at each step. Since $i_\variablez,j_\variablez$ are nonnegative (at every step, $(i_z,j_z)$ is a valid entry of the dynamic-programming table, so $0 \le i_z \le N_{\TVPDp}$ and $0 \le j_z \le N_{\TVPDq}$. In particular, $i_z$ and $j_z$ are nonnegative integers) and we start from $(i_0,j_0) = (\variableN_\TVPDp,\variableN_\TVPDq)$, this process must stop in finitely many steps, and the only possible limit in $\NNB^2$ is $(0,0)$. This proves (1).

Next, by definition, $\initDPOne_\variablez$ is obtained from $\initDPOne_{\variablez-1}$ by adding the current index pair
$\initDPTwo_{\variablez-1} = (i_{\variablez-1},j_{\variablez-1})$ if and only if the step from $\initDPTwo_{\variablez-1}$ to $\initDPTwo_\variablez$ is a diagonal move $(-1,-1)$; in the other two cases (deletion or insertion) the set $\initDPOne_\variablez$ is left unchanged. Thus $\initDPOne_\variableZ$ is exactly the set of index pairs $(i_\variablez,j_\variablez)$ along the backtracking path for which a diagonal step is taken.

Along this path, the indices $(i_\variablez,j_\variablez)$ are strictly decreasing in both coordinates when a diagonal step is taken. Reading these pairs in reverse order (from $(0,0)$ to $(\variableN_\TVPDp,\variableN_\TVPDq)$) therefore yields a family of pairs with strictly increasing first and second coordinates. Hence $\partialAssignment = \initDPOne_\variableZ$ is the graph of a strictly increasing map from a subset of $\{1,\dots,\variableN_\TVPDp\}$ to a subset of $\{1,\dots,\variableN_\TVPDq\}$, i.e.\ $\partialAssignment \in \PASet^\paramDelta(\TVPDp,\TVPDq)$. This proves (2).

We prove point (3) by induction on the backtracking index $\variablez$. Recall that the backtracking produces a sequence
$\bigl(\initDPTwo_\variablez\bigr)_{\variablez}$ with $\initDPTwo_\variablez = (i_\variablez,j_\variablez)$, starting from $(i_0,j_0) = (\variableN_\TVPDp,\variableN_\TVPDq)$ and ending at $(i_\variableZ,j_\variableZ) = (0,0)$, and at each step $\variablez$ we choose $\initDPTwo_{\variablez+1}$ so that one of the three defining equalities of the dynamic-programming recursion holds, namely
\begin{align*}
  \delta_\paramDelta^{\paramWeight,\paramPenalty}(i_\variablez,j_\variablez)
  &=
  \delta_\paramDelta^{\paramWeight,\paramPenalty}(i_\variablez-1,j_\variablez)
  + \paramPenalty \cdot \,\localDistanceAppendix^\paramWeight_\paramDelta(\TVPDp_{i_\variablez},\boundarySet),\\
  \delta_\paramDelta^{\paramWeight,\paramPenalty}(i_\variablez,j_\variablez)
  &=
  \delta_\paramDelta^{\paramWeight,\paramPenalty}(i_\variablez-1,j_\variablez-1)
  + \localDistanceAppendix^\paramWeight_\paramDelta(\TVPDp_{i_\variablez},\TVPDq_{j_\variablez}),\\
  \delta_\paramDelta^{\paramWeight,\paramPenalty}(i_\variablez,j_\variablez)
  &=
  \delta_\paramDelta^{\paramWeight,\paramPenalty}(i_\variablez,j_\variablez-1)
  + \paramPenalty \cdot \,\localDistanceAppendix^\paramWeight_\paramDelta(\TVPDq_{j_\variablez},\boundarySet),
\end{align*}
according to whether the step is a deletion, a substitution, or an insertion.

For each $\variablez\in\{0,\dots,\variableZ-1\}$, let $\mathrm{lc}_\variablez$ denote this \emph{local cost} (deletion, insertion, or substitution) incurred when moving from $\initDPTwo_\variablez$ to $\initDPTwo_{\variablez+1}$. Rearranging the corresponding equality, we obtain
\[
  \delta_\paramDelta^{\paramWeight,\paramPenalty}(i_\variablez,j_\variablez)
  =
  \delta_\paramDelta^{\paramWeight,\paramPenalty}(i_{\variablez+1},j_{\variablez+1})
  + \mathrm{lc}_\variablez.
\]

We now prove by induction on $k\in\{0,\dots,\variableZ\}$ that
\[
  \delta_\paramDelta^{\paramWeight,\paramPenalty}(i_0,j_0)
  =
  \delta_\paramDelta^{\paramWeight,\paramPenalty}(i_k,j_k)
  + \sum_{\variablez=0}^{k-1} \mathrm{lc}_\variablez,
\]
with the convention that the sum is zero for $k=0$.
\begin{itemize}
\item The equality is true for $k=0$, since both sides reduce to $\delta_\paramDelta^{\paramWeight,\paramPenalty}(i_0,j_0)$.

\item 
Assume the statement holds for some $k<\variableZ$. Using the relation above for $\variablez=k$, we have
\[
  \delta_\paramDelta^{\paramWeight,\paramPenalty}(i_k,j_k)
  =
  \delta_\paramDelta^{\paramWeight,\paramPenalty}(i_{k+1},j_{k+1})
  + \mathrm{lc}_k.
\]
Plugging this into the induction hypothesis yields
\[
  \delta_\paramDelta^{\paramWeight,\paramPenalty}(i_0,j_0)
  =
  \delta_\paramDelta^{\paramWeight,\paramPenalty}(i_{k+1},j_{k+1})
  + \sum_{\variablez=0}^{k} \mathrm{lc}_\variablez,
\]
which is exactly the desired formula at rank $k+1$. This completes the induction.

Taking $k=\variableZ$ and using the initialization
$\delta_\paramDelta^{\paramWeight,\paramPenalty}(0,0)=0$ gives
\[
  \delta_\paramDelta^{\paramWeight,\paramPenalty}
    (\variableN_\TVPDp,\variableN_\TVPDq)
  =
  \sum_{\variablez=0}^{\variableZ-1} \mathrm{lc}_\variablez.
\]
By construction of $\partialAssignment$ during backtracking, this sum of local costs coincides with
$\text{cost}^\paramDelta_{(\TVPDp,\TVPDq)}(\partialAssignment)$, which proves
\[
  \text{cost}^\paramDelta_{(\TVPDp,\TVPDq)}(\partialAssignment)
  =
  \delta_\paramDelta^{\paramWeight,\paramPenalty}
    (\variableN_\TVPDp,\variableN_\TVPDq),
\]
and therefore point \textup{(3)} of the lemma.
\end{itemize}
\end{proof}
    
\begin{proposition}\label{prop:piecewise-constant_TVPD_approximation} Let $F\in C^\Delta$.

Then  for every $\varepsilon\in\mathbb{R}^{*}_{+}$, there exists $\tilde{F}^{\eta(\varepsilon)}\in PC^{\Delta}$ such that $\mathrm{CED}^{\Delta}_{\alpha,\beta}(F,\tilde{F}^{\eta(\varepsilon)})<\varepsilon$.
\end{proposition}

\begin{proof}

    Let $F :\text{dom F} = \bigcup_{i \in \{1, \dots, N_F^\Delta\}} I_i^F \to E$, be one continuous TVPD of $S^\Delta$, that is $F\in C^\Delta$. Let $i\,\in\, \{1, \dots, N_F^\Delta\},$ then $F_i$ can be extended continuously to the closure $\overline{I_i^F}$, yielding the extension $F_i'$; because $\overline{I_i^F}$ is compact, then $F'_i$ is uniformly continuous on that set. Then $\forall\,\varepsilon> 0,\exists\,\eta_i>0,\forall(t,s)\in \overline{I_i^F},\bigl|t-s\bigr|<\eta_i\;\Longrightarrow\;
    d\,\!\bigl(F'_i(t),F'_i(s)\bigr)<\varepsilon.$ Let $\varepsilon>0$, then $\forall \,i\,\in\, \{1, \dots, N_F^\Delta\},$ there exists $\eta_i>0$, such that, $\forall(t,s)\in{I_i^F},\bigl|t-s\bigr|<\eta_i\;\Longrightarrow\;
    d\bigl(F_i(t),F_i(s)\bigr)<\varepsilon/(\mu(dom\,F)\cdot(1-\alpha))$. Fixing $\eta\,' = \min_{i \in \{1, \dots, N_F^{\Delta^{}}\}}\eta_i$, let $k_{\Delta}\in\mathbb{N}^*$, such that $\Delta/ \, k_{\Delta}<\eta'$, and define $\eta=\Delta/ \, k_{\Delta}$. For every $i \in \{1, \dots, N_F^\Delta\}$, let $M_i\in\mathbb{N}$, such that $\inf I_i^F + M_i \cdot\eta=\sup I_i^F$, and denote $\mathring{F_i}:\overline{I_i^F}\to E$ the piecewise-constant application defined by $\mathring{F_i}(t)\;=\;F'_i(a_{i,n}),
    \quad \forall \,t\in[a_{i,n},a_{i,n+1}),
    \quad$ with $a_{i,n} \;=\; \inf I_i^F + n \cdot\eta,
    \,\text{for } n=0,\dots,M_i$. Denoting $\tilde{F_i}$ the restriction of $\mathring{F_i}$ on $I_i^F$, and gluing all the $\tilde{F_i}$ together, for every $ i \in\,$$ \{1, \dots, N_F^\Delta\}$, we obtain the piecewise-constant approximation $\tilde{F}:\text{dom F }\to E$ of $F$. Obviously $\tilde{F}\in S^\Delta$ as a simple function on each $I_i^F$, moreover we have $\forall\,t\in \text{dom F},\,d(F(t),\tilde{F}(t))<\varepsilon/(\mu(dom\,F)\cdot(1-\alpha))$, which implies CED$^{\Delta}_{\alpha,\beta} (F,\tilde{F}) \leq \sum_{i \in \{1,\, \dots \,,\,N_F^\Delta\}} \,\,D^\alpha_{\Delta}(F_i,\tilde{F}_i)< N_F^\Delta\cdot\Delta\cdot(1-\alpha)\cdot(\varepsilon/\mu(dom\,F)\cdot(1-\alpha))=\mu(dom\,F)\cdot(1-\alpha)\cdot\varepsilon/(\mu(dom\,F)\cdot(1-\alpha))= \varepsilon$.
    
    \end{proof}

\begin{proposition}
\label{prop:Lipschitz_TVPD}
Let
\[
  \PDSeq_\variablen
  = \bigl(\PD_\variablen,\variableTime_\variablen\bigr)_{0\leq\variablen\leq\variableN_\PDSeq}
\]
be a sequence in $E \times \mathbb{R} $, with
\(
  \variableTime_0 < \variableTime_1 < \dots < \variableTime_{\variableN_\PDSeq}.
\)
For each
\(\variablen\in\{0,\dots,\variableN_\PDSeq-1\}\),
let
\(\WtwoGeodesic_\variablen\colon[0,1]\to E\)
be a constant--speed $d$-geodesic satisfying
\(\WtwoGeodesic_\variablen(0)=\PD_\variablen\) and
\(\WtwoGeodesic_\variablen(1)=\PD_{\variablen+1}\).
Define
\(\TVPDf\colon[\variableTime_0,\variableTime_{\variableN_\PDSeq}]\to E\)
by
\[
  \TVPDf(\variableTime)
  = \WtwoGeodesic_\variablen\!\bigl(\lambda_\variablen(\variableTime)\bigr),
  \quad
  \text{for }\variableTime\in[\variableTime_\variablen,\variableTime_{\variablen+1}),
\]
where
\[
  \lambda_\variablen(\variableTime)
  =
  \frac{\variableTime-\variableTime_\variablen}
       {\variableTime_{\variablen+1}-\variableTime_\variablen}
  \in[0,1],
\]
and set
\(\TVPDf(\variableTime_{\variableN_\PDSeq}) = \PD_{\variableN_\PDSeq}\).

Let
\[
  \variableKLips_\PDSeq
  =
  \max_{0\leq\variablen\leq\variableN_\PDSeq-1}
  \frac{
    d(\PD_\variablen,\PD_{\variablen+1})
  }{
    \variableTime_{\variablen+1}-\variableTime_\variablen
  }.
\]
Then, for all
\(\variableTime,\variableTime'\in
[\variableTime_0,\variableTime_{\variableN_\PDSeq}]\),
\[
  d\bigl(\TVPDf(\variableTime),\TVPDf(\variableTime')\bigr)
  \leq
  \variableKLips_\PDSeq\,\cdot \lvert \variableTime'-\variableTime\rvert.
\]
In particular,
\(\TVPDf\) is \(\variableKLips_\PDSeq\)-Lipschitz as a map
\(
  ([\variableTime_0,\variableTime_{\variableN_\PDSeq}],|\cdot|)
  \to (E,d).
\)
\end{proposition}

\begin{proof}
Fix \(\variablen\in\{0,\dots,\variableN_\PDSeq-1\}\) and let
\(\variableTime,\variableTime'\in
[\variableTime_\variablen,\variableTime_{\variablen+1})\).
Since \(\WtwoGeodesic_\variablen\) is a constant–speed geodesic for
\(d\), we have
\[
  d\bigl(
    \WtwoGeodesic_\variablen(u),
    \WtwoGeodesic_\variablen(v)
  \bigr)
  = |u-v|\cdot \,
    d(\PD_\variablen,\PD_{\variablen+1})
\]
for all \(u,v\in[0,1]\).
By construction of \(\TVPDf\),
\[
  \TVPDf(\variableTime)
  = \WtwoGeodesic_\variablen(\lambda_\variablen(\variableTime)),
  \quad
  \TVPDf(\variableTime')
  = \WtwoGeodesic_\variablen(\lambda_\variablen(\variableTime')).
\]

Moreover, at the breakpoint $\variableTime_{\variablen+1}$ we have
\[
  \TVPDf(\variableTime_{\variablen+1})
  = \PD_{\variablen+1}
  = \WtwoGeodesic_\variablen(1)
  = \WtwoGeodesic_\variablen\bigl(\lambda_\variablen(\variableTime_{\variablen+1})\bigr),
\]
since
\[
  \lambda_\variablen(\variableTime_{\variablen+1})
  =
  \frac{\variableTime_{\variablen+1}-\variableTime_\variablen}
       {\variableTime_{\variablen+1}-\variableTime_\variablen}
  = 1.
\]
So, the identity
\[
  \TVPDf(\variableTime)
  = \WtwoGeodesic_\variablen(\lambda_\variablen(\variableTime))
\]
also holds when $\variableTime=\variableTime_{\variablen+1}$. Thus,
$\forall \variableTime,\variableTime'\in
[\variableTime_\variablen,\variableTime_{\variablen+1}]$,
\begin{align*}
  d\bigl(\TVPDf(\variableTime),\TVPDf(\variableTime')\bigr)
  &=
  \bigl|\lambda_\variablen(\variableTime)
        -\lambda_\variablen(\variableTime')\bigr|\,
   \cdot d(\PD_\variablen,\PD_{\variablen+1}) \\
  &=
  \frac{|\variableTime-\variableTime'|}
       {\variableTime_{\variablen+1}-\variableTime_\variablen}\,
  d(\PD_\variablen,\PD_{\variablen+1}) \\
  &\leq
  \variableKLips_\PDSeq\,\cdot |\variableTime-\variableTime'|.
\end{align*}

For arbitrary
\(\variableTime,\variableTime'\in
[\variableTime_0,\variableTime_{\variableN_\PDSeq}]\)
with \(\variableTime<\variableTime'\),
let
\(\variableTime=\tau_0<\tau_1<\dots<\tau_r=\variableTime'\)
be a subdivision obtained by intersecting
\([\variableTime,\variableTime']\) with the breakpoints
\(\{\variableTime_\variablen\}_{\variablen}\).
Applying the previous bound on each subinterval and using the triangle inequality,
\begin{align*}
  d\bigl(\TVPDf(\variableTime),\TVPDf(\variableTime')\bigr)
  &\leq
  \sum_{k=0}^{r-1}
  d\bigl(\TVPDf(\tau_k),\TVPDf(\tau_{k+1})\bigr) \\
  &\leq
  \sum_{k=0}^{r-1}
  \variableKLips_\PDSeq\,\cdot|\tau_{k+1}-\tau_k| \\
  &=
  \variableKLips_\PDSeq\, \cdot |\variableTime'-\variableTime|.
\end{align*}
This holds for all \(\variableTime,\variableTime'\), which concludes the proof.
\end{proof}

For the geodesic constructions that follow, we further assume that $ (E,d) $ is a geodesic Polish space, and we slightly relax the definition of $\mathcal{S}^\Delta$. We define $\bar{\mathcal{S}}^\Delta$ as the space of functions $F$ satisfying all conditions of $\mathcal{S}^\Delta$ except the condition "$\forall t \in \textit{dom } F, d(F(t), A) > 0$". Similarly, $\bar{s}^\Delta$ denotes the corresponding space of $\Delta$-subdivisions. On $\bar{\mathcal{S}}^\Delta$, we define the equivalence relation $P \sim Q \Leftrightarrow \mathrm{CED}^\Delta_{\alpha,\beta}(P, Q) = 0$, and we work on the quotient metric space $(\bar{\mathcal{S}}^\Delta/\sim, \mathrm{CED}^\Delta_{\alpha,\beta})$. Indeed, all preceding results extend trivially to this setting. The original space $S^\Delta$ embeds naturally into $\bar{S}^\Delta/\sim$ via $P \mapsto [P]$, preserving CED distances, hence for simplicity, we continue to denote this space by $S^\Delta$ and the distance by $\mathrm{CED}^\Delta_{\alpha,\beta}$ in what follows.
    
    \begin{theorem}\label{thm:measurable_selection_1}
    
    Let \( P \in \mathcal{S}^\Delta, Q \in \mathcal{S}^\Delta \). 
    
    If \( (E,d) \) is a geodesic Polish space, then \( \forall i \in \{1,\ldots,N_P\}, \ \forall j \in \{1,\ldots,N_Q\} \),
	
	\[
	\forall T \in (0,1), 
	\]
	
	\[
	\exists G_{i,j}^T \in s^\Delta :
	\left(
	I_{i,j}^T, \completion(\mathcal{B}(\mathbb{R})|_{I_{i,j}^T})
	\right)
	\longrightarrow (E, \mathcal{B}(E))
	\]
	
	measurable, such that for $\lambda$-almost every \(  t \in I_i^P \),
	
	\[
	G_{i,j}^T(t + \inf I_{i,j}^T -a_i ) 
	\in \left\{ G(P(t), Q(t + c_j - a_i))(T) \right\}\subset E,
	\]
	
	where we denote \( a_i = \inf(I_i^P), \ b_i = \sup(I_i^P), \\
	c_j = \inf(I_j^Q) \text{, } d_j = \sup(I_j^Q), \)
	
	\[
	\inf I_{i,j}^T = (1 - T) \cdot a_i + T \cdot c_j,
	\]
	
	\[
	\sup I_{i,j}^T = (1 - T) \cdot b_i + T \cdot d_j
	\]

    and $G(P(t), Q(t + c_j - a_i))$ a d-geodesic from $P(t)$ to $Q(t + c_j - a_i)$ as defined in \autoref{sec:overview_geodesic} of the main manuscrit.

        \end{theorem}
        \begin{proof}
        Let $Geod(E)$ denote the space of geodesics in $E$ (see \autoref{sec:overview_geodesic} of the main manuscrit). Define $Geod^{*}(E)=\{\gamma\in Geod(E)\colon\forall(s,t)\in[0,1]^{2},d(\gamma(s),\gamma(t))=\vert s-t\vert \cdot d(\gamma (0),\gamma (1))\}$. We endow $Geod^{*}(E)$ with the uniform norm to make it a Souslin space.
        
        Let \( i \in \{1,\ldots,N_P\}, j \in \{1,\ldots,N_Q\}, T \in (0,1)\), and \( \Psi \) the multivalued function from \( E \times E \) to the set of nonempty subsets of $Geod^{*}(E)$, defined as,
	\begin{align}
        &\forall (y_1, y_2) \in E \times E,\Psi(y_1,y_2)\notag\\&=\left\{\gamma\in Geod^{*}(E):\gamma(0)=y_{1},\gamma(1)=y_{2}\right\}.
	\end{align}
    
        Let’s denote \( {\Gamma_\Psi}= \big\{((x,y),\gamma)\in E^2\times Geod^{*}(E): \gamma\in \Psi(x,y)\big\}\) its graph; then \({\Gamma_\Psi} \text{ is closed in }(E \times E \times Geod^{*}(E), \mathcal{T}_{(E \times E \times Geod^{*}(E))})\), where $\mathcal{T}$ denote the product topology. Indeed, let \( ((x_n, y_n), \gamma_n)_{n \in \mathbb{N}} \) a sequence in \( {\Gamma_\Psi} \) converging to \( ((x, y), \gamma) \in E^2\times Geod^{*}(E)\). Then,
	\begin{align}
        \lim_{n \to +\infty} d(x_n, x) &=\lim_{n \to +\infty} d(y_n, y) \notag\\&=\lim_{n \to +\infty}\sup_{t\in[0,1]}d(\gamma_n(t), \gamma(t)) =0.
	\end{align}
    
	So, because $E$ is closed,
	\begin{equation}
	\lim_{n \to +\infty} x_n = x\in E, \quad 
	\lim_{n \to +\infty} y_n = y\in E,
	\end{equation}
        \noindent moreover by uniform convergence, in particular, 
	\begin{equation}
        \lim_{n \to +\infty} \gamma_n(0) =\gamma(0),\quad \lim_{n \to +\infty} \gamma_n(1) = \gamma(1).
	\end{equation}
    
        Since $\gamma_n(0)=x_n\to x$ and $\gamma_n(1)=y_n\to y$, we obtain by uniqueness of the limit in a metric space,
        \begin{equation}
        \gamma(0)=x,\qquad \gamma(1)=y.
        \end{equation}
        
        Each $\gamma_n$ is a constant-speed geodesic joining $x_n$ to $y_n$, that is,
        \begin{equation}\label{eq:geod_proof_eq_1}
        \forall (s,t)\in[0,1]^{2},\quad d\big(\gamma_n(s),\gamma_n(t)\big)= |s-t|\cdot\,d(x_n,y_n).
        \end{equation}
        
        As a distance on $E$, \( d(\cdot ,\cdot ) \) is continuous. Morever, $E\times E$ is a metrizable space, then we have the sequential characterization of limits, and $\lim_{n \to +\infty} x_n = x, \lim_{n \to +\infty} y_n = y$, so
        \begin{equation}\label{eq:geod_proof_eq_2}
        \lim_{n \to +\infty} d(x_n,y_n)= d(x,y).
        \end{equation}
        
        Likewise, the uniform convergence of $\gamma_n$ to $\gamma$ implies, for every $(s,t)\in[0,1]^{2}$,
        \begin{equation}
        \lim_{n \to +\infty} \gamma_n(s) =\gamma(s),\quad \lim_{n \to +\infty} \gamma_n(t) = \gamma(t),
        \end{equation}
        thus,
        \begin{equation}\label{eq:geod_proof_eq_3}
        \lim_{n \to +\infty}d\big(\gamma_n(s),\gamma_n(t)\big) = d\big(\gamma(s),\gamma(t)\big).
        \end{equation}

        Given~\eqref{eq:geod_proof_eq_1}, \eqref{eq:geod_proof_eq_2}, and \eqref{eq:geod_proof_eq_3}, we obtain for every $(s,t)\in[0,1]^{2},$
        
        \begin{equation}
        d\big(\gamma(s),\gamma(t)\big)= |s-t|\, d(x,y),
        \end{equation}

        \noindent that is, $\gamma\in \Psi(x,y)$. Finally $x\in E,y\in E,\gamma\in \Psi(x,y)$ so $((x,y),\gamma)\in{\Gamma_\Psi}$. Thus \( {\Gamma_\Psi} \) is closed (indeed $E \times E \times Geod^{*}(E)$ is a metrizable space, as a finite product of metrizable spaces, and hence first-countable) for the product topology, and then a Borel of \( \mathcal{B}(E \times E \times  Geod^{*}(E))\).

        Moreover, \( E\times E \) is a Polish space (as a finite product of Polish spaces) and so a Souslin space; also, $Geod^{*}(E)$ is a Souslin space; thus $E\times E\times Geod^{*}(E)$ is a Souslin space as a finite product of Souslin spaces.

        Since \({\Gamma_\Psi} \) is a Borel subset of the Souslin space  $E\times E\times Geod^{*}(E)$, \( {\Gamma_\Psi} \) is a Souslin set. Then, there exists a map \(f : E \times E \longrightarrow Geod^{*}(E) \text{ that is measurable} \) with respect to, the \( \sigma \)-algebra \( \sigma(\mathcal{S}_{E \times E}) \) generated by all Souslin sets in \( E \times E \), and \( \mathcal{B}(Geod^{*}(E)) \), such that \( f(\omega) \in \Psi(\omega), \) \(\forall \omega \in E \times E\) \cite[Thm.~6.9.2]{bogachev2007measure2}.

	Let us remember that the  
	application
	\begin{align}
        V : \bigl(I_i^P,\; \mathcal{Z}\bigl(\mathcal{B}(\mathbb{R})\!\mid_{I_i^P}\bigr)\bigr)\;\longrightarrow\;\bigl(E^{2},\; \mathcal{B}(E)\otimes\mathcal{B}(E)\bigr),\notag\\ t \longmapsto V(t)=\bigl(P(t),\,Q(t+c_j-a_i)\bigr)
	\notag\end{align}
        \noindent is measurable. Recall also that \(\mathcal{B}(E)\otimes\mathcal{B}(E)=\mathcal{B}(E\times E)\), since \(E\) is separable.
	
        We denote \(\mu=\lambda\circ V^{-1}\), the push-forward measure on \(\bigl(E\times E,\mathcal{B}(E\times E)\bigr)\), where \(\lambda\) is the Lebesgue measure on \(\bigl(I_i^P,\mathcal{Z}(\mathcal{B}(\mathbb{R})\!\mid_{I_i^P})\bigr)\). Moreover, we denote \(\bigl(E\times E,\mathcal{Z}(\mathcal{B}(E\times E))\bigr)\) the completion of \(\bigl(E\times E,\mathcal{B}(E\times E)\bigr)\) for \(\mu\). Since \(E\times E\) is a Polish space, any Souslin set of \(E\times E\) is an analytic set. Also, because any analytic set is universally measurable, and \(\mu\) is a complete measure on \(\bigl(E\times E,\mathcal{Z}(\mathcal{B}(E\times E))\bigr)\) that measures all Borel sets of \(E\times E\), then  
	\begin{equation}
        \sigma(S_{E\times E}) \subset \mathcal{Z}(\mathcal{B}(E\times E)),  
	\end{equation}
	\noindent and so
	\[
	f : \bigl(E\times E,\mathcal{Z}(\mathcal{B}(E\times E))\bigr)
	\longrightarrow
	\bigl(Geod^{*}(E),\mathcal{B}(Geod^{*}(E))\bigr)
	\]
	\noindent is measurable.
	
	Additionally, there exists a measurable function
	\[
	g : (E\times E, \mathcal{B}(E\times E)) \longrightarrow (Geod^{*}(E), \mathcal{B}(Geod^{*}(E)))
	\]
	\(\mu\)-a.s equal to \(f\).
	
        Indeed, first we remark that since \(Geod^{*}(E)\) is metric and separable, then \(Geod^{*}(E)\) has a countable base \(\{A_i\}_{i\in\mathbb{N}}\) and so \(\mathcal{B}(Geod^{*}(E)) = \sigma(\{A_i\}_{i\in\mathbb{N}})\).
	
        Let \(A_i \in \{A_i\}_{i\in\mathbb{N}}\), \(f^{-1}(A_i) \in \mathcal{Z}(\mathcal{B}(E\times E))\), then \(\exists B_{A_i} \in \mathcal{B}(E\times E)\), and \(\exists N_{A_i} \subset E\times E\) a negligible part of \((E\times E, \mathcal{B}(E\times E))\) such that
	\[
	f^{-1}(A_i) = B_{A_i} \cup N_{A_i}.
	\]
	
        Let us denote \(N' = \bigcup_{i\in\mathbb{N}} N_{A_i}\), then \(N' \in \mathcal{Z}(\mathcal{B}(E\times E))\) and
	\begin{equation}
        \mu(N') = \mu\left( \bigcup_{i\in\mathbb{N}} N_{A_i} \right) \leq \sum_{i=0}^{\infty} \mu(N_{A_i}) = \sum_{i=0}^{\infty} 0 = 0.  
	\end{equation}
        Since \(N' \in \completion (\mathcal{B}(E\times E))\), then \(\exists (M,N) \in \mathcal{B}(E\times E)^2\) such that \(M\subset N'\subset N\) and \(\mu(N-M) = 0\). Because \(M\subset N'\) and \(\mu(N')=0\), then \(\mu(M)\leq 0\), so \(\mu(M)=0\). Also, \(\mu(N) = \mu(M) + \mu(N-M) = \mu(M)+0=0\). Then we have \(N'\subset N\), \(N\in\mathcal{B}(E\times E)\), \(\mu(N)=0\).
        
        We define 
        \[g : \bigl(E \times E, \mathcal{B}(E \times E)\bigr) \longrightarrow \bigl(Geod^{*}(E), \mathcal{B}(Geod^{*}(E))\bigr) \] 
        such that
	\[
	g(x) = 
        \begin{cases}f(x) & \text{if } x \in (E \times E) \setminus N,\\y_0 & \text{if } x \in N,\end{cases}
	\]
        with \( y_0 \in Geod^{*}(E) \) an arbitrary fixed element.
	
        Let \( A_i \in \{A_i\}_{i \in \mathbb{N}} \), we have \(g^{-1}(A_i) = \bigl(g^{-1}(A_i) \cap N^c\bigr) \cup \bigl(g^{-1}(A_i) \cap N\bigr).\) For \( x \in N^c \), \( g(x) = f(x) \), then
	\begin{align}
        g^{-1}(A_i) \cap N^c &= f^{-1}(A_i) \cap N^c= (B_{A_i} \cup N_{A_i}) \cap N^c\notag\\&= (B_{A_i} \cap N^c) \cup (N_{A_i} \cap N^c)\notag\\&= (B_{A_i} \cap N^c) \cup \varnothing\notag\\&= B_{A_i} \cap N^c \in \mathcal{B}(E \times E). 
	\end{align}
	
	\smallskip
	
        For \( x \in N \), \( g(x) = y_0\), then
	\begin{equation}
	g^{-1}(A_i) \cap N =
        \begin{cases}N & \text{if } y_0 \in A_i,\\\varnothing & \text{if } y_0 \notin A_i,\end{cases}
	\end{equation}
        thus in any case \( g^{-1}(A_i) \cap N\in \mathcal{B}(E \times E) \).
	
        Finally, \(\bigl(g^{-1}(A_i) \cap N^c\bigr) \cup \bigl(g^{-1}(A_i) \cap N\bigr) \in \mathcal{B}(E \times E)\) and \( g^{-1}(A_i) \in \mathcal{B}(E \times E) \), then  
	\[
        g : \bigl(E \times E, \mathcal{B}(E \times E)\bigr) \longrightarrow \bigl(Geod^{*}(E), \mathcal{B}(Geod^{*}(E))\bigr)
	\]
	is measurable.

        For each $\theta \in[0,1]$, the evaluation map
    
        $$
        e_\theta : Geod^*(E) \to E,\qquad \gamma \mapsto \gamma(\theta),
        $$
        is continuous. Indeed, let $(\gamma_n)_{n\in\mathbb{N}}$ be a sequence in $Geod^{*}(E)$ that converges uniformly to $\gamma$. Fix $\varepsilon>0$. Since $(\gamma_n)$ converges uniformly to $\gamma$, there exists $N(\varepsilon)\in\mathbb{N}$ such that for all $n\ge N(\varepsilon)$,
        $$
        d(\gamma_n(\theta),\gamma(\theta)) \le \sup_{s\in[0,1]} d(\gamma_n(s),\gamma(s)) < \varepsilon,
        $$
        \noindent hence $e_\theta$ is continuous. As a continuous map, 
        $$
        e_T : (Geod^*(E),\mathcal{B}(Geod^*(E))) \to (E,\mathcal{B}(E)),\, \gamma \mapsto \gamma(t),
        $$
        \noindent is measurable. So, as a composition of measurable maps, 
        $$
        m_T = e_{T} \circ g:\bigl(E \times E, \mathcal{B}(E \times E)\bigr) \longrightarrow \bigl(E, \mathcal{B}(E)\bigr),
        $$
        is measurable.
       
        Because \( m_T : \bigl(E \times E, \mathcal{B}(E \times E)\bigr) \longrightarrow \bigl(E, \mathcal{B}(E)\bigr) \) and \( V : \left(I_i^P, \mathcal{Z}\bigl(\mathcal{B}(\mathbb{R})|_{I_i^P}\bigr)\right) \longrightarrow \left(E^2, \mathcal{B}(E \times E)\right) \), \(t \longmapsto V(t) = \bigl(P(t), Q(t+c_j-a_i)\bigr)\)
        are measurables, then \( m_T \circ V: \left(I_i^P, \mathcal{Z}\bigl(\mathcal{B}(\mathbb{R})|_{I_i^P}\bigr)\right) \longrightarrow \left(E, \mathcal{B}(E)\right) \) is measurable. Also, the map
	\[
        \widetilde{U} : \left(I_{i,j}^T, \mathcal{Z}\left(\mathcal{B}(\mathbb{R})|_{I_{i,j}^T}\right)\right) \longrightarrow \left(I_i^P, \mathcal{Z}\left(\mathcal{B}(\mathbb{R})|_{I_{i}^P}\right)\right)
	\]
	\[
        t \longmapsto t + a_i - (1-T)\cdot a_i - T\cdot c_j
	\]
        is measurable (for the same reasons as \(U\)). Then finally, 
	\[
        G_{i,j}^T = (m_T \circ V) \circ \widetilde{U}
	\]
	is measurable and so 
	\begin{equation}
	G_{i,j}^T \in s^\Delta.		    
	\end{equation} \qedhere
        \end{proof}
        
        \begin{lemma}Let $ P \in \mathcal{S}^\Delta.$ If \( (E,d) \) is a geodesic Polish space, $A$ closed and $\forall x \in E, \,  \{y  \in A, d(x,y)=d(x,A) \}$ is non-empty.

\medskip
    
        Then $ \forall i \in \{1,\ldots,N_P\}$, there exists a measurable function $p^i$ : $\left(
        I_{i}^P, \completion(\mathcal{B}(\mathbb{R})|_{I_{i}^P})
        \right)\longrightarrow (E, \mathcal{B}(E))$, such that for $\lambda$-almost every $ t \in I_i^P , \, p^i(t)\in \{y\in A, d(P_i(t),y)=d(P_i(t),A)\}$.
	
	\smallskip
	            
        \end{lemma} 
	
	\begin{proof}

	For clarity, we present the proof together with that of the next theorem.
	    
	\end{proof}

	\bigskip

	\begin{theorem}\label{thm:measurable_selection_2}Let $ P \in \mathcal{S}^\Delta.$ If \( (E,d) \) is a geodesic Polish space, $A$ closed and $\forall x \in E, \, \{y  \in A, d(x,y)=d(x,A) \}$ \\ is non-empty.
	
	Then $ \forall i \in \{1,\ldots,N_P\}$,
	
	\[
	\forall T \in (0,1), 
	\]
	
	\[
	\exists G_{i}^T \in s^\Delta :
	\left(
	I_{i}^P, \completion(\mathcal{B}(\mathbb{R})|_{I_{i}^P})
	\right)
	\longrightarrow (E, \mathcal{B}(E))
	\]
	
	measurable, such that for $\lambda$-almost every \(  t \in I_i^P \),
	
	\[
	G_{i}^T(t) 
	\in \left\{ G(P(t), p^i(t))(T) \right\}.
	\]
    
	where $G(P(t), p^i(t))$ a d-geodesic from $P(t)$ to $p^i(t)$ as defined in \autoref{sec:overview_geodesic} of the main manuscrit.
		    
	\end{theorem}

	\begin{proof}

	Let \( i \in \{1, \dots, N_P\} \),  
	and let \( \Pi \) the multivalued function  
	from \( E \) to the set of non-empty subsets  
	of \( E \), defined as \( \forall x \in E \),  
	\[
	\Pi(x) = \left\{ y \in A, \text{ such that } d(x,y) = d(x,A) \right\}
	\]
	
	Let’s denote \( \Gamma_\Pi \) its graph, then \( \Gamma_\Pi \) is closed in \( (E \times E, \mathcal{B}(E \times E)) \)

\medskip
    
	Indeed, we have  
	\[
	\Gamma_\Pi = \left\{ (x,y) \in E^2, y \in A, d(x,y) = d(x,A) \right\}
	\]
	
	Let \( (x_n, y_n)_{n \in \mathbb{N}} \) a sequence in \( \Gamma_\Pi \) converging  
	to \( (x,y) \in E \times E \).  
	
	Then,  
	\[
	\lim_{n \to +\infty} d(x_n,x) = \lim_{n \to +\infty} d(y_n,y) = 0
	\]
	and so  
	\[
	\lim_{n \to +\infty} x_n = x, \quad \lim_{n \to +\infty} y_n = y
	\]
	
	Moreover \( y \in A \) because \( A \) is closed in  
	\( (E,d) \) first-countable space (because metric),  
	and because \( \forall n \in \mathbb{N}, y_n \in A \).  
	
	Also, \( d(x_n, y_n) = d(x_n, A) \) because  
	
	\[
	(x_n, y_n) \in \Gamma_\Pi
	\]
	
	Then  
	\[
	\lim_{n \to +\infty} d(x_n,y_n) = \lim_{n \to +\infty} d(x_n, A)
	\]
	
	So, because \( d(\cdot,\cdot) \) and \( d(\cdot,A) \) are  
	continuous (because \( (E,d) \) is metric, then $d(\cdot,\cdot)$ is continuous and $d(\cdot,A)$ is $1$-Lipschitz, and hence continuous) and \( \lim_{n \to +\infty} x_n = x \),  
	and \( \lim_{n \to +\infty} y_n = y \), then \( d(x,y) = d(x,A) \) by the sequential characterization of limits in metrizable spaces.
	
	Finally, we have \( (x,y) \in \Gamma_\Pi \) and  
	so \( \Gamma_\Pi \) is closed in \( E^2 \) for the product topology, and then \( \Gamma_\Pi \in \mathcal{B}(E^2) \).  
	
	Moreover \( E \) is a Polish space,  
	then \( E^2 \) is a Polish space, and  
	so a Souslin space.  
	
	Since \( \Gamma_\Pi \) is a Borel subset of  
	the Souslin space \( E^2 \),  
	\( \Gamma_\Pi \) is a Souslin set.  
	
	Then, there exists a mapping  
	\[
	\widetilde{\rho} : E \longrightarrow E \text{ that is measurable}
	\]
	with respect to the \( \sigma \)-algebra \( \sigma(\mathcal{S}_E) \)  
	generated by all Souslin sets in \( E \),  
	and \( \mathcal{B}(E) \), such that \( \widetilde{\rho}(\omega) \in \Pi(\omega) \),  
	\( \forall \omega \in E \,\text{\cite{bogachev2007measure2}}\).
	
	Let us remember that the application \( P_i : \left( I_i^P, \mathcal{Z}(\mathcal{B}(\mathbb{R})|_{I_i^P}) \right) \longrightarrow (E, \mathcal{B}(E)) \),  
	\( t \mapsto P(t) \) is measurable by definition.  
	
	We denote \( \mu^* = \lambda \circ P_i^{-1} \), the pushforward  
	measure on \( (E, \mathcal{B}(E)) \), where \( \lambda \) is  
	the Lebesgue measure on \( \left( I_i^P, \mathcal{Z}(\mathcal{B}(\mathbb{R})|_{I_i^P}) \right) \). We denote \( (E, \mathcal{Z}(\mathcal{B}(E))) \) the completion  
	of \( (E, \mathcal{B}(E)) \) for \( \mu^* \).

	Since $E$ is a Polish space, any Souslin set of $E$ is an analytic set. Also, because any analytic set is universally measurable, and $\mu^*$ is a complete measure on $(E, \mathcal{Z}(\mathcal{B}(E)))$ that measures all Borel of $E$, then $\sigma(\mathcal{S}_E) \subset \mathcal{Z}(\mathcal{B}(E))$, and so $\tilde{\rho} : (E, \mathcal{Z}(\mathcal{B}(E))) \rightarrow (E, \mathcal{B}(E))$ is measurable.
	
    Then there exists a measurable function $\rho : (E, \mathcal{B}(E)) \rightarrow (E, \mathcal{B}(E))$ which is $\mu^*$-a.s. equal to $\tilde{\rho}$ (See the proof in the previous theorem).

    So $p_i:=\rho \circ P_i$ is measurable, with respect to $\big(
        I_{i}^P, \completion(\mathcal{B}(\mathbb{R})|_{I_{i}^P})
        \big)$, and $(E, \mathcal{B}(E))$, as a composition of measurable functions, and we can apply theorem~\ref{thm:measurable_selection_1} with $P_i$ and $p_i=\rho\, \circ P_i$ (the assumption $p_i \in s^\Delta$ is not needed in order to apply the theorem) to conclude the proof. \qedhere
    	    
	\end{proof} 
	
	\textbf{Remark :} Many conditions allow the hypothesis, $\forall x \in E, \,  \{y  \in A, d(x,y)=d(x,A) \}$ non-empty, to be satisfied. Some examples are : A compact; $\forall x \in E, A\cap B(x,r+\epsilon)$ compact (with any $\epsilon \in \mathbb{R}^*_+$, and $r =d(x,A$)); (E,d) complete and $\forall x \in E , A\cap B(x,r+\epsilon)$ relatively compact; Heine-Borel E space; E complete connected Riemannian manifold; (E,d) complete length-metric space and locally compact.

\begin{lemma}\label{lem:delta-subdivision_geodesic_space}
If $(E,d)$ is a geodesic Polish space, then $(s^{\Delta},D^\alpha_\Delta)$ is a geodesic metric space.
\end{lemma}

\begin{proof}
Let $\rho,q \in s^\Delta$. Then there exist $P,Q \in \mathcal{S}^\Delta$ and indices $i \in \{1,\dots,N_P\}$, $j \in \{1,\dots,N_Q\}$ such that $\rho = P_i$ and $q = Q_j$. Write $ a_i = \inf(I_i^P), \ b_i = \sup(I_i^P), c_j = \inf(I_j^Q)$, and  $d_j = \sup(I_j^Q)$. By the definition of $\mathcal{S}^\Delta$,
\[
  b_i - a_i = d_j - c_j = \Delta.
\]

For $T \in (0,1)$, let $G_{i,j}^T \in s^\Delta$ be given by Theorem~\ref{thm:measurable_selection_1}, with
\(
  \inf I_{i,j}^T
  =
  \bigl(1-T)\cdot a_i + T\cdot c_j,\;\text{and } \sup I_{i,j}^T=(1-T)\cdot b_i + T\cdot d_j\bigr.
\).  For $t \in (0,1)$, set
\[
  \phi_P(t) = a_i \cdot (1-t) + b_i \cdot t \in I_i^P,
  \phi_Q(t) = c_j \cdot (1-t) + d_j \cdot  t \in I_j^Q.
\]
Since $b_i-a_i = d_j-c_j = \Delta$, we have
\[
  \phi_Q(t) = \phi_P(t) + c_j - a_i,\quad t\in(0,1).
\]

By Theorem~\ref{thm:measurable_selection_1}, for every $s\in I_i^P$,
\[
  G_{i,j}^T\bigl(s + \inf I_{i,j}^T - a_i\bigr)
  \in
  \bigl\{\gamma\bigl(P(s),Q(s+c_j-a_i)\bigr)(T)\bigr\},
\]
where $\gamma(x,y)$ denotes a constant–speed $d$-geodesic in $E$ from $x$ to $y$ (for clarity in the following proof, we will no longer use the letter $G$ to denote $d$-geodesics). Taking $s=\phi_P(t)$ and using $\phi_P(t)+c_j-a_i=\phi_Q(t)$ gives, for all $t\in(0,1)$,
\[
  G_{i,j}^T\Bigl(\inf I_{i,j}^T \cdot (1-t)+\sup I_{i,j}^T\, \cdot t\Bigr)
 \] \[ =
  G_{i,j}^T\bigl((1-T) \cdot \phi_P(t)+T \cdot \phi_Q(t)\bigr)
  \]\[
  = \gamma\bigl(P(\phi_P(t)),Q(\phi_Q(t))\bigr)(T).
\]

Now let $0 \leq S < T \leq 1$. For $S,T\in(0,1)$, by the constant–speed property of $\gamma$ we obtain, for all $t\in(0,1)$,
\begin{align*}
  d\Bigl(&
    G_{i,j}^S\bigl(\inf I_{i,j}^S(1-t)+\sup I_{i,j}^S\,t\bigr),
   \\& G_{i,j}^T\bigl(\inf I_{i,j}^T(1-t)+\sup I_{i,j}^T\,t\bigr)
  \Bigr) \\
  &\qquad=
  \bigl|T-S\bigr|\,
  d\bigl(P(\phi_P(t)),Q(\phi_Q(t))\bigr).
\end{align*}
Moreover, the left endpoints of $I_{i,j}^S$ and $I_{i,j}^T$ satisfy
\begin{align*}
  \inf I_{i,j}^T - \inf I_{i,j}^S
  &=
  \bigl[(1-T) \cdot a_i+T \cdot c_j\bigr] \\&- \bigl[(1-S) \cdot a_i+S \cdot c_j\bigr]
  \\&=
  (T-S)(c_j-a_i),
\end{align*}
so the temporal shift in $D^\alpha_\Delta$ between $G_{i,j}^S$ and $G_{i,j}^T$ is $|T-S|\cdot|c_j-a_i|$.

By the definition of $D^\alpha_\Delta$ (using the common parametrization $t\mapsto \inf I_{i,j}^{\boldsymbol{\cdot}}(1-t)+\sup I_{i,j}^{\boldsymbol{\cdot}}t$) we therefore have
\begin{align}
  &D^\alpha_\Delta\bigl(G_{i,j}^S,G_{i,j}^T\bigr)
  = \Delta \int_0^1
    (1-\alpha)\,
    d\bigl(
      G_{i,j}^S(\cdots),G_{i,j}^T(\cdots)
    \bigr)\,dt
    \;\notag\\&+\;
    \alpha\,\Delta\,\bigl| \inf I_{i,j}^T - \inf I_{i,j}^S \bigr|
  \\\notag
  &= \Delta \int_0^1
    (1-\alpha)\,
    |T-S|\,
    d\bigl(P(\phi_P(t)),Q(\phi_Q(t))\bigr)\,dt
    \;\\&+\;
    \alpha\,\Delta\,|T-S|\,|c_j-a_i|
  \\\notag
  &= |T-S|\,
    \Bigl[
      \Delta \int_0^1
      (1-\alpha)\,
      d\bigl(P(\phi_P(t)),Q(\phi_Q(t))\bigr)\,dt
      \;\\&+\;
      \alpha\,\Delta\,|c_j-a_i|
    \Bigr]
  \\
  &= |T-S| \cdot D^\alpha_\Delta(P_i,Q_j)\notag.
\end{align}

We now extend the definition to the endpoints by setting
\[
  G_{i,j}^0 := P_i,
  \qquad
  G_{i,j}^1 := Q_j,
\]
and define
\[
  g_{i,j}:[0,1]\to s^\Delta,\qquad
  g_{i,j}(T) = G_{i,j}^T.
\]
Then for all $S,T\in[0,1]$ we have
\[
  D^\alpha_\Delta\bigl(g_{i,j}(S),g_{i,j}(T)\bigr)
  = |T-S|\;D^\alpha_\Delta(P_i,Q_j),
\]
so $g_{i,j}$ is a constant–speed geodesic in $(s^\Delta,D^\alpha_\Delta)$ joining $\rho=P_i$ to $q=Q_j$.

Since $\rho,q\in s^\Delta$ were arbitrary, this shows that $(s^\Delta,D^\alpha_\Delta)$ is a geodesic metric space.

\bigskip

\textit{Remark :} We used in this proof that, for every $T\in(0,1)$, and  for every $t\in(0,1)$, we have

\begin{gather}
G_{i,j}^T\big(\inf I_{i,j}^T \cdot (1-t)+\sup I_{i,j}^T  \cdot t\big)\in
\big\{\gamma\big(P(a_i \cdot (1-t)\notag\\+\,b_i \cdot  t),\,Q(c_j \cdot (1-t)+d_j  \cdot t)\big)(T)\big\}. 
\end{gather}

Indeed, as said previously, set $T\in(0,1)$, by Theorem~\ref{thm:measurable_selection_1}, $\forall s\in(a_i,b_i)$ we have

\begin{gather} G_{i,j}^T\big(s+\inf I_{i,j}^T-a_i\big)\notag\\\in
\big\{\gamma(P(s),\,Q(s+c_j-a_i))(T)\big\}. 
\end{gather}

Then, because,
\begin{align}\phi_P:(0,1)\to(a_i,b_i),\quad \phi_P(t)=a_i\cdot (1-t)+b_i\cdot t,\\\phi_Q:(0,1)\to(c_j,d_j),\quad \phi_Q(t)=c_j\cdot (1-t)+d_j\cdot t,
\end{align}

\noindent we have for $t\in(0,1)$, 

\begin{align*}
&G_{i,j}^T\big(s+\inf I_{i,j}^T-a_i\big)\notag\\&\in
\big\{\gamma(P(s),\,Q(s+c_j-a_i))(T)\big\}\\& \implies G_{i,j}^T\big(\phi_P(t)+\big((1-T)a_i+Tc_j\big)-a_i \big)\notag\\&\in
\big\{\gamma(P(\phi_P(t)),\,Q(\phi_P(t)+c_j-a_i))(T)\big\}\\& \implies G_{i,j}^T\big((1-T)\phi_P(t)+T\big(c_j+\phi_P(t)-a_i\big))\notag\\&\in\big\{\gamma(P(\phi_P(t)),\,Q(c_j+(b_i-a_i)t))(T)\big\}\\&\implies G_{i,j}^T\big((1-T)\phi_P(t)+T\big(c_j+(b_i-a_i)t)\big))\notag\\&\in\big\{\gamma(P(\phi_P(t)),\,Q(c_j+(b_i-a_i)t))(T)\big\}\\&\implies G_{i,j}^T\big((1-T)\phi_P(t)+T\big(c_j+(d_j-c_j)t)\big))\notag\\&\in\big\{\gamma(P(\phi_P(t)),\,Q(c_j+(d_j-c_j)t))(T)\big\}\\&\implies G_{i,j}^T\big((1-T)\phi_P(t)+T\big(c_j(1-t)+d_j t )\big))\notag\\&\in\big\{\gamma(P(\phi_P(t)),\,Q(c_j(1-t)+d_j t ))(T)\big\}\\&\implies G_{i,j}^T\big((1-T)\phi_P(t)+T\big(\phi_Q(t))\big))\notag\\&\in\big\{\gamma(P(\phi_P(t)),\,Q(\phi_Q(t)))(T)\big\}\\&\implies G_{i,j}^T\big((1-T)\big(a_i(1-t)+b_i t\big)+T\big(c_j(1-t)+d_j t\big)\big))\notag\\&\in\big\{\gamma(P(\phi_P(t)),\,Q(\phi_Q(t)))(T)\big\}\\&\implies G_{i,j}^T\big([(1-T)a_i+Tc_j](1-t)+[(1-T)b_i+Td_j]t\big))\notag\\&\in\big\{\gamma(P(\phi_P(t)),\,Q(\phi_Q(t)))(T)\big\}\\&\implies G_{i,j}^T\big(\inf I_{i,j}^T(1-t)+\sup I_{i,j}^T t\big))\notag\\&\in\big\{\gamma(P(\phi_P(t)),\,Q(\phi_Q(t)))(T)\big\}. 
\end{align*}

To make the proof more concrete, we now compute the length of $g_{i,j}$ using uniform partitions of $[0,1]$, indeed let $N\in\mathbb{N}^{*}$, then
\begin{align}
&\sum_{k=0}^{N-1} D^\alpha_\Delta\!( G_{i,j}^{k/N},\, G_{i,j}^{(k+1)/N} )
\\& \displaybreak[2]=\sum_{k=0}^{N-1} \Delta \cdot\int_0^1 (1-\alpha)\cdot d\Big( G_{i,j}^{k/N}\!\big( \inf I^{k/N}_{i,j}\cdot(1-t)\notag\\&+\sup I^{k/N}_{i,j}\cdot t\big),\, 
G_{i,j}^{(k+1)/N}\!\big( \inf I^{(k+1)/N}_{i,j}\cdot(1-t)\notag\\&+\sup I^{(k+1)/N}_{i,j}\cdot t\big) \Big)
\,dt+ \alpha\cdot |c_j-a_i|\cdot \frac{\Delta + k\cdot(\Delta-\Delta)}{N}\,
\\&= \sum_{k=0}^{N-1} \Delta \cdot \int_0^1 (1-\alpha)\cdot \frac{1}{N}\, d( P(a_i\cdot(1-t)+b_i\cdot t),\, \notag\\&Q(c_j\cdot(1-t)+d_j\cdot t) )
\,dt+ \alpha\cdot \frac{|c_j-a_i|}{N}\, 
\\&= \sum_{k=0}^{N-1} \frac{1}{N}\cdot D^\alpha_\Delta(P_i,Q_j)= D^\alpha_\Delta(P_i,Q_j).
\end{align}
Moreover, we have
\begin{align}
\lim_{N\to+\infty} D^\alpha_\Delta&\!\left( G_{i,j}^{k/N},\, G_{i,j}^{(k+1)/N} \right)
\\&= \lim_{N\to+\infty} \frac{1}{N}\cdot D^\alpha_\Delta(P_i,Q_j)=0,
\end{align}
Thus finally, with $g_{i,j}$ the map defined by,
$$
g_{i,j}:[0,1]\longrightarrow s^\Delta,\qquad T\longmapsto g_{i,j}(T)=G_{i,j}^T,
$$
we have $g_{i,j}$ continuous, and
\begin{align}
\sup \sum\,_{i=0}^{k}\,D^\alpha_\Delta\!\left( g_{i,j}(T_{i}),\, g_{i,j}(T_{i+1}) \right)= D^\alpha_\Delta(P_i,Q_j),
\end{align}
\noindent (where the supremum is taken over all $k\in \mathbb{N}^{*}$, and all sequences $T_0=0< T_1 < \dots < T_k=1$ in $[0,1])$, so $g_{i,j}$ is a geodesic.

\end{proof}

    \begin{theorem}\label{thr:final}
If $(E,d)$ is a geodesic Polish space, and $A\subsetneq E$ is closed such that, $\forall x \in E, \{y \in A ,\, d(x,y)=d(x,A)\}$ is non-empty, then \[(S^\Delta, \mathrm{CED}^\Delta_{\alpha,1}) \]

is a geodesic space.
\end{theorem}

\begin{proof}

Let $P \in S^\Delta$, $Q \in S^\Delta$, with optimal partial assignment $f \in \mathcal{A}^\Delta(P,Q)$. If $\ell \in [0, L_{1/3})$, then let $\widetilde{G}^\ell = (\widetilde{G}^\ell_i)_{1 \le i \le N_P^\Delta} \in S^\Delta$ defined by
\begin{equation}
\forall i \in \mathcal{S}_P,\quad 
\widetilde{G}^\ell_i = P_i,\text{ and  } 
\forall i \in \mathcal{D}_P,\quad 
\widetilde{G}^\ell_i = G_i^{\frac{\ell}{L_{1/3}}},
\label{eq:def_G_l}
\end{equation}

\noindent as defined in Theorem~\ref{thm:measurable_selection_2} for $P_i$, from $P_i$ to $A\,(\text{i.e., }p_i=\rho \circ P_i)$.

\medskip

If $\ell = L_{1/3}$, let
\begin{equation}
\widetilde{G}^\ell = P \setminus \{P_i,\ i \in \mathcal{D}_P\}.
\label{eq:G_l_at_L}
\end{equation}

\medskip

Let $(\ell_j)$ be a sequence in $[0, L_{1/3}]$ such that
\[
\ell_0 = 0 < \ell_1 < \cdots < \ell_K = L_{1/3}.
\]

Then,
\begin{multline}
\sum_{j=0}^{K-1} \mathrm{CED}_{\alpha,1}^\Delta\big( \widetilde{G}^{\ell_j}, \widetilde{G}^{\ell_{j+1}} \big)
\\
\le
\sum_{j=0}^{K-2} \sum_{i=0}^{N_P^\Delta}
D^\alpha_\Delta\big(\widetilde{G}_i^{\ell_j}, \widetilde{G}_i^{\ell_{j+1}}\big)
+ \sum_{i \in \mathcal{D}_P} D^\alpha_\Delta\big(\widetilde{G}_i^{\ell_{K-1}}, A\big).
\label{eq:ced_sum_bound}
\end{multline}

Moreover, 
\begin{align*}
&
\sum_{j=0}^{K-2} \sum_{i \in \mathcal{D}_P}
D^\alpha_\Delta\big(
  \widetilde{G}^{\ell_j}_i,
  \widetilde{G}^{\ell_{j+1}}_i
\big)
+ \sum_{i \in \mathcal{D}_P}
D^\alpha_\Delta\big(
  \widetilde{G}^{\ell_{K-1}}_i,
  A
\big)
\\[0.3em]
&=
\sum_{j=0}^{K-2} \sum_{i \in \mathcal{D}_P}
\frac{\ell_{j+1} - \ell_j}{L_{1/3}}
\cdot D^\alpha_\Delta(P_i, A)
+ \sum_{i \in \mathcal{D}_P}
D^\alpha_\Delta\big(
  \widetilde{G}_i^{\ell_{K-1}},
  A
\big)
\\[0.3em]
&=
\frac{1}{L_{1/3}}
\sum_{i \in \mathcal{D}_P}
D^\alpha_\Delta(P_i, A)
\sum_{j=0}^{K-2} (\ell_{j+1} - \ell_j)
\\
&\quad
+ \sum_{i \in \mathcal{D}_P}
\frac{\ell_K - \ell_{K-1}}{L_{1/3}}
\cdot D^\alpha_\Delta(P_i, A)
\\[0.3em]
&=
\frac{1}{L_{1/3}}
\sum_{i \in \mathcal{D}_P}
\Big[
D^\alpha_\Delta(P_i, A)
\big(
  \sum_{j=0}^{K-2} (\ell_{j+1} - \ell_j)
  \\&+ (\ell_K - \ell_{K-1})
\big)
\Big]
\\[0.3em]
&=
\frac{1}{L_{1/3}}
\sum_{i \in \mathcal{D}_P}
D^\alpha_\Delta(P_i, A)
\cdot L_{1/3}
\\[0.3em]
&=
\sum_{i \in \mathcal{D}_P}
D^\alpha_\Delta(P_i, A)
\end{align*}

Then,
\begin{equation}
\begin{aligned}
\sup
\sum
\mathrm{CED}_{\alpha,1}^\Delta\big(
  \widetilde{G}^{\ell_j},
  \widetilde{G}^{\ell_{j+1}}
\big)
&\leq
\sum_{i \in \mathcal{D}_P}
D^\alpha_\Delta(P_i, A)
\end{aligned}
\end{equation}

\noindent
where the supremum is taken over all
$K \in \mathbb{N}^*$, and all sequences\\
$\ell_0 = 0 < \ell_1 < \ldots < \ell_K = L_{1/3}$
in $[0, L_{1/3}]$.

\medskip

Moreover,
\begin{equation}
\begin{aligned}
\sum_{j=0}^{K-1}
\mathrm{CED}_{\alpha,1}^\Delta(
  \widetilde{G}^{\ell_j},
  \widetilde{G}^{\ell_{j+1}}
)
&\ge
\mathrm{CED}_{\alpha,1}^\Delta(
  \widetilde{G}^0,
  \widetilde{G}^{L_{1/3}}
)
\\&=
\mathrm{CED}_{\alpha,1}^\Delta(P, \widetilde{G}^{L_{1/3}})
\\&=\mathrm{CED}_{\alpha,1}^\Delta\!\left(
  P,\,
  P \setminus \{ P_i,\ i \in \mathcal{D}_P \}
\right)
\\
&=
\sum_{i \in \mathcal{D}_P} D^\alpha_\Delta(P_i, A)
\end{aligned}
\end{equation}

Indeed, if \( f \in \mathcal{A}^\Delta(P, Q) \) is optimal, then 
\[ f^* \in \mathcal{A}^\Delta\left(P, \bigl\{\TVPDq_{\partialAssignment(\variablei)}\mid \variablei\in\substitutionSet_\TVPDp\bigr\}\right) \] defined as
\begin{equation}
\begin{aligned}
f^* : \{1, \ldots, N_P^\Delta\}
&\to
\{ j \text{ such that } \exists i \in \{1,\ldots, N_P^\Delta\},\\
&\qquad j = f(i) \},
\\
i &\mapsto f(i)
\end{aligned}
\end{equation}
and it is optimal for \( (P, \bigl\{\TVPDq_{\partialAssignment(\variablei)}, \variablei\in\substitutionSet_\TVPDp\bigr\}) \). Indeed, by contradiction, if another one were optimal instead, this would induce an optimal assignment for \( (P, Q) \) with a cost smaller than \( f \), which would contradict the fact that \( f \) is optimal for \( (P, Q) \).

Moreover, since \( f^* \) is optimal for \( (P, \bigl\{\TVPDq_{\partialAssignment(\variablei)}, \variablei\in\substitutionSet_\TVPDp\bigr\}) \), we have
\[
f^{**} \in \mathcal{A}^\Delta(P, \widetilde{G}^{L_{1/3}}) \text{ defined as}
\]
\begin{equation}
\begin{aligned}
f^{**} : i \in \{1, \ldots, N_P^\Delta\}
&\to  \mathcal{S}_P,
\\
i &\mapsto i, \text{ if } i \in \mathcal{S}_P
\end{aligned}
\end{equation}
and it is optimal. Indeed, \( f^* \) is optimal for \( (P, \bigl\{\TVPDq_{\partialAssignment(\variablei)}, \variablei\in\substitutionSet_\TVPDp\bigr\}) \). Thus, for every \( i \in \mathcal{S}_P \), if we move \( Q_{f(i)} \) towards \( P_i \), then the cost of the assignment decreases and remains optimal; in the limit, we reach \( P_i \), which gives the optimal assignment \( f^{**} \) between \( P \) and \( \widetilde{G}^{L_{1/3}} \), with cost
\[
\sum_{i \in \mathcal{D}_P} D^\alpha_\Delta(P_i, A).
\]

Thus, we finally obtain
\begin{equation}
\begin{aligned}
\sup
\sum_{j=0}^{K-1}
\mathrm{CED}^{\Delta}_{\alpha,1}\big(
  \widetilde{G}^{\ell_j},
  \widetilde{G}^{\ell_{j+1}}
\big)
&=
\sum_{i \in \mathcal{D}_P}
D^\alpha_\Delta(P_i, A).
\end{aligned}
\end{equation}

Moreover, the map
\[
G_D : [0, L_{1/3}] \to \mathcal{S}^\Delta,
\qquad
\ell \mapsto \widetilde{G}^\ell
\]
is trivially uniformly continuous by construction, since
\begin{multline}
\lim_{\ell \to L}
\mathrm{CED}^{\Delta}_{\alpha,1}\big(
  \widetilde{G}^{\ell},
  \widetilde{G}^L
\big)
\le
\lim_{\ell \to L}
\sum_{i \in \mathcal{D}_P}
D^\alpha_\Delta\big(
  \widetilde{G}^{\ell}_{i},
  \widetilde{G}^L_{i}
\big)
\\
=
\lim_{\ell \to L}
\frac{1}{L_{1/3} }
\sum_{i \in \mathcal{D}_P}
|\ell - L|\,
D^\alpha_\Delta(P_i, A)
= 0.
\end{multline}

Next, if \( \ell \in [L_{1/3}, L_{2/3}] \), then let
\[
\begin{aligned}
\widetilde{G}^\ell
&=
(\widetilde{G}^\ell_i)_{1 \le i \le N_P^\Delta,\ i \in \mathcal{S}_P}
\in \mathcal{S}^\Delta
\end{aligned}
\quad \text{defined by}
\]
\[
\forall i \in \mathcal{S}_P,\quad
\widetilde{G}^\ell_i
=
G_{i,f(i)}^{
   \frac{\ell - L_{1/3}}{L_{2/3} - L_{1/3}}
}
\quad \] as defined in Theorem~\ref{thm:measurable_selection_1}, from
$P_i$ to $ Q_{f(i)}$.
\medskip

Let $(\ell_j)_{0 \le j \le K \in \mathbb{N}^*}$
be a sequence in $[L_{1/3}, L_{2/3}]$
such that
\[
\ell_0 = L_{1/3}
< \ell_1 < \cdots < \ell_K = L_{2/3}.
\]

Then,
\begin{equation}
\sum_{j=0}^{K-1}
\mathrm{CED}^{\Delta}_{\alpha,1}\big(
  \widetilde{G}^{\ell_j},
  \widetilde{G}^{\ell_{j+1}}
\big)
\le
\sum_{j=0}^{K-1}
\sum_{i \in \mathcal{S}_P}
D^\alpha_\Delta\big(
  \widetilde{G}^{\ell_j}_i,
  \widetilde{G}^{\ell_{j+1}}_i
\big).
\end{equation}

Moreover,
\begin{equation}
\begin{aligned}
&\sum_{j=0}^{K-1}
\sum_{i \in \mathcal{S}_P}
D^\alpha_\Delta\big(
  \widetilde{G}^{\ell_j}_i,
  \widetilde{G}^{\ell_{j+1}}_i
\big)
\\&=
\sum_{j=0}^{K-1}
\sum_{i \in \mathcal{S}_P}
\left(
  \frac{\ell_{j+1} - \ell_j}{L_{2/3} - L_{1/3}}
\right)
D^\alpha_\Delta(P_i, Q_{f(i)})
\\
&=
\frac{1}{L_{2/3} - L_{1/3}}
\sum_{i \in \mathcal{S}_P}
D^\alpha_\Delta(P_i, Q_{f(i)})
\sum_{j=0}^{K-1}
(\ell_{j+1} - \ell_j)
\\
&=
\frac{1}{L_{2/3} - L_{1/3}}
\sum_{i \in \mathcal{S}_P}
D^\alpha_\Delta(P_i, Q_{f(i)})
\big( L_{2/3} - L_{1/3} \big)
\\
&=
\sum_{i \in \mathcal{S}_P}
D^\alpha_\Delta(P_i, Q_{f(i)}).
\end{aligned}
\end{equation}

Thus,
\begin{equation}
\begin{aligned}
\sup
\sum_{j=0}^{K-1}
\mathrm{CED}^\Delta_{\alpha,1}\big(
  \widetilde{G}^{\ell_j},
  \widetilde{G}^{\ell_{j+1}}
\big)
&\le
\sum_{i \in \mathcal{S}_P}
D^\alpha_\Delta(P_i, Q_{f(i)}).
\end{aligned}
\end{equation}

where the supremum is taken over all $K \in \mathbb{N}^*$ and all sequences
\[
\ell_0 = L_{1/3} < \ell_1 < \cdots < \ell_K = L_{2/3}
\quad\text{in } [L_{1/3}, L_{2/3}].
\]

Moreover,
\begin{equation}
\begin{aligned}
\sum_{j=0}^{K-1}
\mathrm{CED}^\Delta_{\alpha,1}\big(
  \widetilde{G}^{\ell_j},
  \widetilde{G}^{\ell_{j+1}}
\big)
&\ge
\mathrm{CED}^\Delta_{\alpha,1}\big(
  \widetilde{G}^{L_{1/3}},
  \widetilde{G}^{L_{2/3}}
\big)
\\
&=
\mathrm{CED}^\Delta_{\alpha,1}\Big(
  P \setminus \{ P_i,\ i \in \mathcal{D}_P \},
\\[-0.3em]
&\qquad\qquad
  \{ Q_{f(i)},\ i \in \mathcal{S}_P \}
\Big)
\\
&=
\sum_{i \in \mathcal{S}_P}
D^\alpha_\Delta\big(
  P_i,
  Q_{f(i)}
\big),
\end{aligned}
\end{equation}
again by contradiction.

Then,
\begin{equation}
\begin{aligned}
\sup
\sum_{j=0}^{K-1}
\mathrm{CED}^\Delta_{\alpha,1}\big(
  \widetilde{G}^{\ell_j},
  \widetilde{G}^{\ell_{j+1}}
\big)
&=
\sum_{i \in \mathcal{S}_P}
D^\alpha_\Delta\big(
  P_i,
  Q_{f(i)}
\big).
\end{aligned}
\end{equation}

Moreover,
\[
G_S : [L_{1/3}, L_{2/3}] \rightarrow \mathcal{S}^\Delta,
\qquad
\ell \mapsto \widetilde{G}^\ell
\]
is trivially continuous by construction, since
\begin{equation}
\begin{aligned}
&\lim_{\ell \rightarrow L}
\mathrm{CED}^\Delta_{\alpha,1}\big(
  \widetilde{G}^\ell,
  \widetilde{G}^L
\big)
\\&\leq
\lim_{\ell \rightarrow L}
\sum_{i \in \mathcal{S}_P}
D^\alpha_\Delta\big(
  \widetilde{G}^{\ell}_i,
  \widetilde{G}^L_i
\big)
\\
&=
\frac{1}{L_{2/3} - L_{1/3}}
\lim_{\ell \rightarrow L}
\sum_{i \in \mathcal{S}_P}
|\ell - L|\,
D^\alpha_\Delta\big(
  P_i,
  Q_{f(i)}
\big)
\\
&= 0.
\end{aligned}
\end{equation}

Finally, if $\ell \in ( L_{2/3}, L_{3/3} ]$, then let
\[
\widetilde{G}^\ell
=
\big(
  \widetilde{G}^\ell_i
\big)_{1 \leq i \leq N_Q^\Delta}
\in \mathcal{S}^\Delta
\text{ defined by}
\]

$\forall i \notin \mathcal{I}_Q,\quad 
\widetilde{G}^\ell_i = Q_i,\text{ and  }\forall i \in \mathcal{I}_Q,\quad
\widetilde{G}^\ell_i
=
G^{\frac{{\ell} - L_{2/3}}{L_{3/3} - L_{2/3}}}_{i}
\quad 
$
\medskip

\noindent as defined in Theorem~\ref{thm:measurable_selection_2} for $Q_i$, but from $A\,(\text{i.e., }p_i=\rho \circ Q_i)$ to $Q_i$.

\medskip

Let $(\ell_j)_{0 \leq j < K \in \mathbb{N}^*}$ be a sequence in
$[L_{2/3}, L_{3/3}]$ such that
\[
\ell_0 = L_{2/3}
< \ell_1 < \cdots < \ell_K = L_{3/3}.
\]

Then,
\begin{equation}
\begin{aligned}
\sum_{j=0}^{K-1}
\mathrm{CED}^\Delta_{\alpha,1}\big(
  \widetilde{G}^{\ell_j},
  \widetilde{G}^{\ell_{j+1}}
\big)
&\leq
\sum_{j=1}^{K-1}
\sum_{i=0}^{N_Q^\Delta}
D^\alpha_\Delta\big(
  \widetilde{G}^{\ell_j}_i,
  \widetilde{G}^{\ell_{j+1}}_i
\big)
\\
&\quad+
\sum_{i \in \mathcal{I}_Q}
D^\alpha_\Delta\big(
  \widetilde{G}^{\ell_1}_i,
  A
\big).
\end{aligned}
\end{equation}

\begin{align*}
&=
\left(
  \sum_{j=1}^{K-1}
  \sum_{i \in \mathcal{I}_Q}
  D^\alpha_\Delta\big(
    \widetilde{G}^{\ell_j}_i,
    \widetilde{G}^{\ell_{j+1}}_i
  \big)
\right)
+
\sum_{i \in \mathcal{I}_Q}
  D^\alpha_\Delta\big(
    \widetilde{G}^{\ell_1}_i,
    A
  \big)
\\[0.4em]
&=
\left(
  \sum_{j=1}^{K-1}
  \sum_{i \in \mathcal{I}_Q}
  \frac{(\ell_{j+1} - \ell_j)}{L_{3/3}-L_{2/3}}\,
  D^\alpha_\Delta(Q_i, A)
\right)
\\&\qquad+
\sum_{i \in \mathcal{I}_Q}
  D^\alpha_\Delta\big(\widetilde{G}^{\ell_1}_i,A
  \big)
\\[0.4em]
&=
\left(
  \frac{1}{L_{3/3} - L_{2/3}}
  \sum_{i \in \mathcal{I}_Q}
  D^\alpha_\Delta(Q_i, A)
  \sum_{j=1}^{K-1}
  (\ell_{j+1} - \ell_j)
\right.
\\
&\qquad\left.
  + \sum_{i \in \mathcal{I}_Q}
    \frac{\ell_1 - \ell_0}{L_{3/3} - L_{2/3}}\,
    D^\alpha_\Delta(Q_i, A)
\right)
\\[0.4em]
&=
\frac{1}{L_{3/3} - L_{2/3}}
\sum_{i \in \mathcal{I}_Q}
\Big(
  D^\alpha_\Delta(Q_i, A)
  \big(
    \sum_{j=1}^{K-1}
    (\ell_{j+1} - \ell_j)
    \\&+ (\ell_1 - \ell_{0})
  \big)
\Big)
\\[0.4em]
&=
\sum_{i \in \mathcal{I}_Q} D^\alpha_\Delta(Q_i, A).
\end{align*}

Then
\begin{equation}
\sup \sum_{j=0}^{K-1}
\mathrm{CED}^\Delta_{\alpha,1}\big(
  \widetilde{G}^{\ell_j},
  \widetilde{G}^{\ell_{j+1}}
\big)
\leq
\sum_{i \in \mathcal{I}_Q}
D^\alpha_\Delta(Q_i, A).
\end{equation}

\noindent
where the supremum is taken over all
$K \in \mathbb{N}^\star$ and all sequences
$\ell_0 = L_{2/3} < \ell_1 < \ldots < \ell_K = L_{3/3}$
in $[L_{2/3}, L_{3/3}]$.

\medskip
\noindent
Moreover,
\begin{equation}
\sum_{j=0}^{K-1}
\mathrm{CED}^\Delta_{\alpha,1}\big(
  \widetilde{G}^{\ell_j},
  \widetilde{G}^{\ell_{j+1}}
\big)
\geq
\mathrm{CED}^\Delta_{\alpha,1}\big(
  \widetilde{G}^{L_{2/3}},
  \widetilde{G}^{L_{3/3}}
\big).
\end{equation}

\begin{align*}
&=
\mathrm{CED}^\Delta_{\alpha,1}\big(
  \widetilde{G}^{L_{2/3}},
  Q
\big)
\\
&=
\mathrm{CED}^\Delta_{\alpha,1}\big(
  \{ Q_{f(i)} \mid i \in \mathcal{S}_P \},
  Q
\big)
\\
&=
\sum_{i \in \mathcal{I}_Q} D^\alpha_\Delta(Q_i, A)
\end{align*}

\noindent
(by the same argument as previously).

\medskip
\noindent
We thus obtain
\begin{equation}
\sum_{j=0}^{K-1}
\mathrm{CED}^\Delta_{\alpha,1}\big(
  \widetilde{G}^{\ell_j},
  \widetilde{G}^{\ell_{j+1}}
\big)
=
\sum_{i \in \mathcal{I}_Q} D^\alpha_\Delta(Q_i, A).
\end{equation}

\noindent
Moreover,
$G_I : [L_{2/3}, L_{3/3}] \to \mathcal{S}^\Delta$,
$\ell \mapsto \widetilde{G}^\ell$ is continuous, again by the same argument as previously.

\medskip
\noindent
By concatenating the three continuous paths $G_D$, $G_S$, and $G_I$,
we obtain a continuous path $G$ of length
$\mathrm{CED}^\Delta_{\alpha,1}(P, Q)$
such that $G(0) = P$ and $G(L_{3/3}) = Q$.

    \end{proof}

	\section{Data specification}

    This section catalogues the ensemble datasets used in our study. For each dataset, we report its provenance, the format in which we handle it, any preprocessing applied, and the associated ground-truth classification. All ensembles were obtained from publicly available sources.

    To streamline reproducibility, we supply scripts that automatically (i) retrieve the data, (ii) run the TTK pipeline for preprocessing, and (iii) export standardized VTK files embedding the ground-truth labels as VTK “Field Data.” For convenience, we also publish a ready-to-use archive containing the curated ensembles in VTK format. All scripts and curated data are available at: \href{https://github.com/sebastien-tchitchek/ContinuousEditDistance}{https://github.com/sebastien-tchitchek/ContinuousEditDistance}. In addition, the code package ships with the full sets of TVPDs computed from these inputs.

    \subsection{Asteroid impact}

    Asteroid Impact (SciVis Contest 2018) comprises six timed PL-scalar fields—YA11, YB11, YC11, YA31, YB31, YC31 sequences—totaling approximately 300 GB. The raw data are available at: \url{https://oceans11.lanl.gov/deepwaterimpact/}. Each member simulates either a direct ocean impact or an atmospheric airburst whose blast wave interacts with the sea surface. Two asteroid diameters are explored (first digit 1 → 100 m; 3 → 250 m) and three impact/airburst altitudes A: sea-surface impact; B: explosion at 5 km; C explosion at 10 km). We analyze the matter-density scalar field, which clearly separates asteroid, water, and ambient air. This ensemble is a parameter study; here we examine how the asteroid’s size affects the resulting wave. In our pipeline, salient maxima capture effectively the asteroid and large water splashes; thus, each member is represented as a time series of persistence diagrams of maxima. The ground-truth labels group members by asteroid diameter, so the classification task is to assign each series to its correct diameter class. The ground-truth classification is as follows:

    \begin{itemize}
    \item \textbf{Class 1} (3 members): yA11, yB11, yC11
    \item \textbf{Class 2} (3 members): yA31, yB31, yC31
    \end{itemize}

    \subsection{Sea surface height}

    This ensemble contains 48 members provided as 2D regular grids at 1440×720 resolution. Each member is a global sea-surface height observation acquired in January, April, July, and October 2012. The raw data can be found at the following address: \url{https://apdrc.soest.hawaii.edu/erddap/griddap/hawaii_soest_90b3_314d_ab45.html}. The ensemble has been used in prior work \cite{favelier2018persistenceatlascriticalpoint,vidal2019progressive, pont2021wasserstein}, and corresponding VTK files and persistence diagrams are available at \url{https://github.com/julesvidal/wasserstein-pd-barycenter}. The features of interest are ocean eddy centers, which are reliably captured by height extrema. Accordingly, each observation can be represented by a persistence-diagram representation of the ssh field. The ground truth groups observations by month—four classes (January, April, July, October) representative of seasons—so the task is to identify, for any given observation, its correct month/season. Concretely in our TVPD setting, we consider the four time series of 12 observations/diagrams (one series per month). Next, each monthly series is further split by taking alternating time stamps (even/odd indices), yielding eight sub-series of six observations each; these naturally cluster into four groups—one per month—each grouping the two sub-series derived from the same month. The ground-truth labels are:

    \begin{itemize}
    \item \textbf{Class 1} (2 members): 201201-even, 201201-odd
    \item \textbf{Class 2} (2 members): 201204-even, 201201-odd
    \item \textbf{Class 3} (2 members): 201207-even, 201201-odd
    \item \textbf{Class 4} (2 members): 201210-even, 201201-odd
    \end{itemize}

    \subsection{VESTEC}

    In VESTEC \cite{flatken2023vestec}, space weather—understood here as the collection of physical phenomena in the solar system, particularly near Earth, with emphasis on magnetic and radiative effects—is investigated through the analysis of magnetic reconnection events in the magnetosphere. Given the complexity of these processes, only standard theoric simulations are currently performed. Accordingly, an ensemble of 3D magnetic-reconnection simulations was generated under varied initial conditions, with the project’s goal being to build a simulation-and-analysis pipeline for decision support during catastrophic events. In previous work \cite{flatken2023vestec}, four simulation runs were executed on the same 3D domain (128 × 64 × 64) for 2,500 time steps, with distinct input parameter sets (variations of the initial magnetic field and particle types present in the domain). The code used for these simulations is available at: \url{https://github.com/KTH-HPC/iPIC3D}. For every time step in each run, the persistence diagram of the magnetic-field magnitude were computed, yielding a corpus of 10,000 diagrams archived in a Cinema database. In our TVPDs application context, we embedded these four simulations into three dimensions using multidimensional scaling (MDS); within each run, we then selected a subsequence so that the runs separated clearly into two clusters under the MDS embedding. Specifically, we retained similar subsequences for simulations 1–2 and for simulations 3–4, ensuring that these two groups were well separated. The ground-truth classification is:

    \begin{itemize}
    \item \textbf{Class 1} (2 members): VESTEC1, VESTEC2
    \item \textbf{Class 2} (3 members): VESTEC3, VESTEC4
    \end{itemize}

    \section{Parameter settings}

    \noindent
    This appendix details the parameter settings used in the experiments reported in the main text.
    Beyond these specific choices, the discussion is intended as a practical guideline for selecting
    \(\paramWeight\), \(\paramDelta\), \(\eta\) and \(\beta\) when
    applying CED to other TVPD datasets.

    \sebastien{\textbf{Parameter settings.} In all experiments, we fix \(\paramWeight < 10^{-4}\). This choice is motivated by the fact that, across all datasets considered, temporal distances between time samples in the TVPDs to be compared are several orders of magnitude larger than the \(\wasserstein_2\) distances between the persistence diagrams composing these TVPDs. Setting \(\paramWeight\) to such a small value lets the Wasserstein term
dominate the metric, so that \(\CEDM^{\paramDelta}_{\paramWeight,\paramPenalty}\) essentially compares TVPDs through their topological content while still respecting their temporal
structure. The same value of \(\alpha\) was enforced for the other dissimilarity measures
(L2, Fréchet, TWED, DTW) in our MDS-based clustering experiment.}

\sebastien{As discussed
\julien{in the main manuscript} (see \sebastienBis{\autoref{sec:global_distance} and \autoref{sec:pc_TVPD}}), from a theoretical standpoint the
most relevant choice is to take $\paramDelta$ as small as possible, in order to
increase the resolution of the assignments, and to set $\eta = \paramDelta$,
which simplifies the computations. However, in practice this has to be balanced
against reasonable
running times. \sebastienBis{A simple practical strategy is to fix a target average number of $\paramDelta$-subdivisions per input TVPD in the sample (e.g., 100), and then choose $\paramDelta$ so that the resulting TVPDs have approximately this average number of subdivisions.} \sebastienBis{In our experiments, we adopt a more data-adaptive protocol in which the number of
$\paramDelta$-subdivisions, and hence the values of $\paramDelta$ and $\eta$, adapts to the input data in order to obtain more accurate numerical
approximations. For these experiments, the protocol is as follows:} for each dataset in our
clustering study, we choose a common approximation tolerance \(\scalarVariableOne>0\). We then select a subdivision step \(\paramDelta\) satisfying \(
  \paramDelta < \min_{\PDSeq} \frac{\scalarVariableOne}
           {(1-\paramWeight)\,\variableKLips_{\PDSeq}\,
    (\variableTime_{\variableN_\PDSeq}-\variableTime_0)}\), (recall that $\eta\leq\paramDelta$ by construction) where the minimum is taken over all TVPDs $\PDSeq$ in the dataset \sebastienBis{(see \autoref{sec:pc_TVPD} of the main manuscript for the definitions of \(\variableKLips_{\PDSeq}\) and \(\variableTime_{\variableN_\PDSeq}\))}, and we set \(\eta = \paramDelta\). This guarantees that, for all input TVPDs in the dataset, the corresponding piecewise-constant approximation incurs a \(\CEDM^{\paramDelta}_{\paramWeight,\paramPenalty}\) error on the order of \(\scalarVariableOne\). To select the tolerance \(\scalarVariableOne\) for a dataset, we first fixed a moderately small subdivision step \(\paramDelta\) (and set \(\eta = \paramDelta\), as discussed above) and computed all the pairwise \(\CEDM^{\paramDelta}_{\paramWeight,\paramPenalty}\) distances between the input TVPDs in the sample. We then inspected the typical scale of these distances and chose \(\scalarVariableOne\) as a small fraction of that scale (e.g., significantly smaller than the typical pairwise distance, typically on the order of one fiftieth to one hundredth of that scale in practice, depending on the available computational resources), so that the error introduced by the piecewise-constant approximation remained small relative to the scale of the \(\CEDM^{\paramDelta}_{\paramWeight,\paramPenalty}\) distances within the dataset. This procedure yields \(\paramDelta = 0.1\) for sea-surface height, \(\paramDelta = 0.25\) for VESTEC, and \(\paramDelta = 1000\) for asteroid impact.}

\sebastien{For the tracking experiments, which were all conducted on the asteroid impact dataset, the subdivision step \(\paramDelta\) (and hence \(\eta\) as well, with
\(\eta = \paramDelta\)) was further reduced in order to increase the resolution of the optimal assignments and to more precisely highlight the accuracy of pattern synchronization achieved by \CEDPP, whether in the temporal-shift–recovery (\(\paramDelta=250\)) or pattern-search setting (\(\paramDelta=500\)).}

\sebastien{Finally, we fix \(\beta = 1\) by default in all experiments, consistently with our geodesic construction, so that all reported results rely on the same CED geometry and are directly comparable.}